\documentclass{article}
\usepackage{amsthm, caption, subcaption, graphicx, algorithm, algorithmic, nicematrix, authblk}
\usepackage[hidelinks]{hyperref}
\usepackage{cleveref}
\usepackage[margin=1in]{geometry}

\newtheorem{definition}{Definition}[section]
\newtheorem{assumption}{Assumption}[section]
\newtheorem{lemma}{Lemma}[section]

\newtheorem{theorem}{Theorem}[section]

\numberwithin{equation}{section}
\numberwithin{figure}{section}
\numberwithin{algorithm}{section}

\usepackage{amsmath, amssymb, bm}

\newcommand{\N}[0]{{\mathbb{N}}}

\newcommand{\R}[0]{{\mathbb{R}}}

\newcommand{\cA}[0]{{\mathcal{A}}}
\newcommand{\cB}[0]{{\mathcal{B}}}
\newcommand{\cC}[0]{{\mathcal{C}}}

\newcommand{\cF}[0]{{\mathcal{F}}}

\newcommand{\cH}[0]{{\mathcal{H}}}
\newcommand{\cI}[0]{{\mathcal{I}}}
\newcommand{\cJ}[0]{{\mathcal{J}}}

\newcommand{\cL}[0]{{\mathcal{L}}}

\newcommand{\cN}[0]{{\mathcal{N}}}
\newcommand{\cO}[0]{{\mathcal{O}}}
\newcommand{\cP}[0]{{\mathcal{P}}}
\newcommand{\cQ}[0]{{\mathcal{Q}}}
\newcommand{\cR}[0]{{\mathcal{R}}}
\newcommand{\cS}[0]{{\mathcal{S}}}
\newcommand{\cT}[0]{{\mathcal{T}}}
\newcommand{\cU}[0]{{\mathcal{U}}}
\newcommand{\cV}[0]{{\mathcal{V}}}

\newcommand{\cX}[0]{{\mathcal{X}}}

\newcommand{\vc}[0]{{\mathbf{c}}}

\newcommand{\vi}[0]{{\mathbf{i}}}

\newcommand{\vr}[0]{{\mathbf{r}}}

\newcommand{\vt}[0]{{\mathbf{t}}}

\newcommand{\vv}[0]{{\mathbf{v}}}

\newcommand{\vx}[0]{{\mathbf{x}}}
\newcommand{\vy}[0]{{\mathbf{y}}}
\newcommand{\vz}[0]{{\mathbf{z}}}
\newcommand{\vzero}[0]{{\mathbf{0}}}

\newcommand{\vxi}[0]{{\bm{\xi}}}

\newcommand{\mA}[0]{{\mathbf{A}}}			
\newcommand{\mB}[0]{{\mathbf{B}}}
\newcommand{\mC}[0]{{\mathbf{C}}}
\newcommand{\mD}[0]{{\mathbf{D}}}

\newcommand{\mG}[0]{{\mathbf{G}}}

\newcommand{\mI}[0]{{\mathbf{I}}}

\newcommand{\mL}[0]{{\mathbf{L}}}

\newcommand{\mP}[0]{{\mathbf{P}}}
\newcommand{\mQ}[0]{{\mathbf{Q}}}
\newcommand{\mR}[0]{{\mathbf{R}}}
\newcommand{\mS}[0]{{\mathbf{S}}}
\newcommand{\mT}[0]{{\mathbf{T}}}
\newcommand{\mU}[0]{{\mathbf{U}}}
\newcommand{\mV}[0]{{\mathbf{V}}}

\newcommand{\mX}[0]{{\mathbf{X}}}
\newcommand{\mY}[0]{{\mathbf{Y}}}
\newcommand{\mZ}[0]{{\mathbf{Z}}}

\newcommand{\mTheta}[0]{{\mathbf{\Theta}}}

\newcommand{\mLambda}[0]{{\mathbf{\Lambda}}}

\newcommand{\mPi}[0]{\mathbf{{\Pi}}}
\newcommand{\mSigma}[0]{{\mathbf{\Sigma}}}

\newcommand{\mZero}[0]{\mathbf{0}}

\newcommand{\rvV}[0]{{\vec{V}}}
\newcommand{\rvW}[0]{{\vec{W}}}
\newcommand{\rvX}[0]{{\vec{X}}}
\newcommand{\rvY}[0]{{\vec{Y}}}
\newcommand{\rvZ}[0]{{\vec{Z}}}

\newcommand{\tildemu}[0]{{\widetilde{\mu}}}

\newcommand{\suchthat}[0]{\, : \,}      
\newcommand{\bigunion}[0]{\bigcup}      

\DeclareMathOperator{\rank}{rank}        
\DeclareMathOperator{\range}{range}      
\DeclareMathOperator{\trace}{tr}         
\newcommand{\tp}[0]{^\mathsf{T}}         
\newcommand{\stp}[0]{^{\,\mathsf{T}}}    
\newcommand{\frob}[0]{{\mathrm{F}}}      

\newcommand{\prob}[0]{\mathrm{Pr}}
\newcommand{\E}[0]{\mathbb{E}}

\DeclareMathOperator{\cov}{Cov}
\DeclareMathOperator{\var}{Var}

\newcommand{\sign}[0]{\mathrm{sign}}
\newcommand{\idx}[0]{\mathrm{idx}}
\newcommand{\poly}[0]{\mathrm{Poly}}
\newcommand{\dist}[0]{\mathrm{dist}}
\newcommand{\diam}[0]{\mathrm{diam}}
\newcommand{\depth}[0]{\mathrm{depth}}
\newcommand{\ind}[0]{\bm{1}}                                 
\DeclareMathOperator*{\argmin}{arg\, min}                    

\newcommand{\nth}[0]{{^\text{th}}}                           
\newcommand{\mcol}[0]{\,\mathord{:}\,}                       

\makeatletter                                                
\newcommand*{\defeq}{\mathrel{\rlap{%
                     \raisebox{0.3ex}{$\m@th\cdot$}}%
                     \raisebox{-0.3ex}{$\m@th\cdot$}}%
                     =}
\makeatother

\newcommand\Tstrut{\rule{0pt}{2.6ex}}         

\title{Estimating Hierarchically Rank Structured Covariance Matrices}
\author[1]{Robin Armstrong\footnote{Email: \texttt{robin.armstrong@uni-potsdam.de} (corresponding author)}}
\affil[1]{\small Universit\"at Potsdam, Institut f\"ur Mathematik, 14476 Potsdam, Germany}
\author[2]{Anil Damle\footnote{Email: \texttt{damle@cornell.edu}}}
\affil[2]{Cornell University, Department of Computer Science, Ithaca, NY 14853, USA}
\author[3]{Samuel E.\ Otto\footnote{Email: \texttt{s.otto@cornell.edu}}}
\affil[3]{Cornell University, Sibley School of Mechanical and Aerospace Engineering, Ithaca, NY 14853, USA}
\date{}

\begin{document}

    \maketitle

    \begin{abstract}
        We consider the problem of estimating a high-dimensional covariance matrix from a very limited number of samples. This problem is ubiquitous in computational fluid dynamics, where a small number of fluid snapshots must be used to construct a Gramian matrix determining a reduced-order model, as well as in computational geoscience, where a small ensemble of Earth system forecasts must be used to estimate the covariance matrix associated with the forecast uncertainty. It is common practice to regularize the small-sample covariance by imposing a ``localization'' structure that enforces a physically realistic correlation lengthscale, imposing a sparsity constraint, ``shrinking'' towards a prescribed target, or attenuating small correlations. We propose an alternate technique that regularizes the small-sample covariance by imposing hierarchical rank structure. Compared to regularization methods that assume sparsity such as spatial localization, hierarchical rank structure accommodates a wider range of covariance matrices, roughly corresponding to situations where long-range correlations vary more smoothly than short-range ones. It also results in a data-sparse matrix format that permits highly efficient matrix-vector products. We present theory and algorithms which show how to efficiently estimate a high-dimensional, hierarchically rank structured covariance matrix from limited samples. Through an error analysis and numerical experiments with a variety of model problems, we demonstrate that these techniques are effective at reducing sampling errors, and that in many cases they achieve smaller estimation error than conventional techniques.
    \end{abstract}
    
    \section{Introduction} \label{section:hrs_cov_intro}

\indent

How many independent and identically distributed (i.i.d.) samples does it take to estimate the covariance matrix of an $n$-dimensional probability distribution? In the absence of any prior knowledge of the covariance structure, well-known statistical theory \cite{vershynin_hdp,wainwright_hds} has established that $\cO(n)$ samples are required. On the other hand, practitioners in physical science routinely encounter situations where $\cO(n)$ is far too much to ask. This is an especially acute challenge in computational fluid dynamics and Earth systems modeling; these fields rely on model reduction and data assimilation algorithms that require estimates of a covariance matrix with as many as $\approx 10^9$ rows and columns. Each of the sample vectors used to estimate this matrix is produced by a costly numerical fluids simulation, and as a result, computational resource constraints limit the number of samples to $\approx 10^2$. This dramatic gap between the theoretically required and practically available sample sizes, a difference of $7$ orders of magnitude, must be bridged by regularizing the sample statistics with \emph{a priori} knowledge of the covariance structure.

To design a regularized covariance estimator that is both practical and performant, one must carefully balance the flexibility of their prior structural assumptions against the computational constraints that a given structure introduces. Illustrative examples of this principle can be found in the so-called ``localization'' strategies developed by meteorological research communities. On one end of the spectrum, spatial localization \cite{gaspari_cohn,hamill_distance_dependent_filtering,houtekamer_mitchell_enkf} imposes the fairly strict assumption that covariances between widely separated points of space are vanishingly small, thus eliminating the need to estimate matrix elements far from the diagonal. This method produces estimates that have a range of useful properties including positive definiteness and data sparsity, but it requires sophisticated modifications to handle multiscale covariance structures \cite{gilpin_generalzed_gc,wang_mlgetkf}, and its performance depends sensitively on how the localization bandwidth is tuned. In contrast, correlation-based localization methods \cite{bishop_hodyss_ecorap_1,bishop_hodyss_ecorap_2,vishny_covariance} can accommodate highly general covariance structures, but the resulting estimates are difficult to represent in a data-sparse format and are rarely positive definite. Statistical research communities have developed their own arsenal of regularized estimators that fall at various points along this spectrum, including methods based on thresholding \cite{bickel_levina_thresholding,rothman_generalized_thresholding}, sparse precision matrix estimation \cite{bickel_levina_regularized_covariance,friedman_graphical_lasso,wu_pourahmadi_cholesky}, and shrinkage \cite{ledoit_wolf_shrinkage,ledoit_wolf_analytical_shrinkage,ledoit_wolf_shrinkage_review}, to name a few.

By assuming that point-to-point covariances are essentially zero at long distances, spatial localization imposes a banded structure on the covariance matrix. If one makes the weaker assumption that point-to-point covariances vary more smoothly at long distances than at short distances, the result is a matrix with \emph{hierarchical rank structure}. We will define this type of structure formally in later sections of this paper (cf.\ \cref{subsec:hrs}), and comprehensive introductions to this subject can also be found in texts such as \cite{bebendorf,hackbusch}. The purpose of this paper is to develop estimators of hierarchically rank-structured covariance matrices. Our reasons for doing so are twofold: first, matrices of this type can be represented in compressed formats that allow more efficient storage and calculations than a generic dense array, including the $\cH$-matrix and $\cH^2$-matrix formats \cite{bebendorf,hackbusch}, hierarchically semiseparable format \cite{martinsson_hss,xia_fast_hss_algs}, and recursively skeletonized format \cite{ho_recursive_skeletonization,miden_rs_gp,minden_strong_rs}. Several authors have shown that compressing a covariance matrix into one of these formats can improve the efficiency of tasks such as, for example, optimization loops for Gaussian process regression \cite{ballani_kressner_sparseinverse_hierarchical} and Kalman filtering \cite{saibaba_hrs_geostats}; from this perspective, the basic toolkit we are using is not new. Our second motive for examining hierarchical rank structure is what differentiates this study clearly from the previous literature: we seek to use these matrix formats as regularizing constraints within the covariance estimation procedure itself.

The theory of high-dimensional statistics provides strong evidence that this should be possible. This body of literature has established that low-rank covariance matrices can be estimated from relatively few i.i.d.\ samples, regardless of the size of the matrix itself; see, e.g., \cite[remark 3.1.5]{vershynin_prob_methods_for_data}. Hierarchical rank structure is a far less restrictive assumption than low-rank structure, but remarkably, computational complexity and storage efficiency results for this matrix class usually differ from the corresponding results for low-rank matrices by only polylogarithmic factors \cite{bebendorf,hackbusch}. It is reasonable to ask whether a similar relationship could be proven between the sample complexities of estimating low-rank and hierarchically rank structured covariance matrices. In \cref{section:hcov}, we will provide an affirmative answer for covariance matrices possessing so-called ``asymptotically smooth'' structure, in the form of an efficient estimation algorithm that we call \text{HCov} (\cref{alg:hcov}) and an accompanying error bound (\cref{thm:hcov_bound_global}). \Cref{subsec:tidal_covar,subsec:qg_climatology} will show how asymptotically smooth covariance structures can naturally arise in the stationary distributions of dynamical systems and spatial processes, and will show experimentally that \text{HCov} can recover these structures from far fewer samples than traditional methods. Because \text{HCov} is not guaranteed to return a positive semidefinite matrix, we will also develop an estimator called \text{RSCov} (\cref{alg:rscov}) in \cref{section:rscov} that guarantees positivity using a factorization-based representation of hierarchical rank structure.

\subsection{Code Availability}

\indent

The experiments in this paper can be reproduced using our publicly available code: \url{https://github.com/robin-armstrong/hrs-covariance-experiments/releases/tag/arxiv-v1}.

    \section{Background} \label{section:hrs_cov_background}

\indent

In this section, we will explain in greater detail the dynamical systems modeling contexts that motivate this study. We will then survey the relevant covariance estimation strategies from existing literature, and introduce the foundational concepts of hierarchically rank structured matrices.

\subsection{The Need for Small-Sample Covariance Estimators} \label{subsec:mor_da_context}

\indent

In computational fluid dynamics, there are a number of settings where one must characterize a fluid's main directions of variability within its extremely high-dimensional phase space. Two such instances, of great current interest to modeling communities, are the following.

\begin{enumerate}
\item \underline{Model reduction.} In this setting, the fluid is represented by a discretized state vector $\vx \in \R^n$ that evolves according to a set of model equations $\dot{\vx}(t) = f(\vx(t))$, $f : \R^n \to \R^n$. The goal is to construct a simplified model of the fluid's dynamics in terms of a much smaller state vector $\vz \in \R^r$, $r \ll n$, given a number of ``snapshots''
\begin{equation*}
\vx_i = \vx(t_i),\quad 0 \leq t_1 < t_2 \ldots < t_m \leq T,
\end{equation*}
representing $m$ observations of the fluid's state across a time interval $[0,T]$. Proper orthogonal decomposition \cite{berkooz_pod}, or POD, does this by identifying a linear subspace $\cV_r \subseteq \R^n$, $\dim \cV_r = r$, that contains the most significant directions of variability observed in the snapshots. This subspace is constructed from the empirical covariance matrix of the snapshots; specifically, $\cV_r = \range \mV_r$, where $\mV_r = [\vv_1 \:\: \cdots \:\: \vv_r ] \in \R^{n \times r}$ is an orthonormal basis obtained from the spectral decomposition
\begin{equation}
\cov[\vx_1,\, \ldots,\, \vx_m] \defeq \frac{1}{m-1}\sum_{i = 1}^m (\vx_i - \overline{\vx})(\vx_i - \overline{\vx})\tp = \sum_{i = 1}^n \lambda_i \vv_i\vv_i\tp, \label{eqn:pod_snapshot_covariance}
\end{equation}
where $\overline{\vx} \defeq \frac{1}{m} \sum_{i = 1}^m \vx_i$ and $\lambda_1 \geq \ldots \geq \lambda_n \geq 0$. It is useful to think of $\cov[\vx_1,\, \ldots,\, \vx_m]$ as approximating a ``true'' covariance matrix
\begin{equation*}
\mC = \lim_{M \to \infty} \cov[\vx_1 \:\cdots\cdots \vx_M],
\end{equation*}
where $\{ \vx_i \}_{i = 1}^\infty$ is an infinite collection of snapshots whose empirical distribution converges to an ergodic distribution for the fluid system as a whole. If this approximation is accurate enough, then a reduced order model can be constructed from \cref{eqn:pod_snapshot_covariance} using (for example) Galerkin projection, yielding
\begin{equation*}
\vx(t) \approx \mV_r\vz(t),\quad \dot{\vz}(t) = \mV_r\tp f(\mV_r\vz(t)).
\end{equation*}

\item \underline{Ensemble data assimilation.} Here the goal is to estimate the state of a partially observed fluid (e.g., the atmosphere) in a Bayesian manner, by using empirical observations (e.g., from satellites, weather balloons, etc.) to condition a prior probability distribution representing uncertain knowledge of the fluid's current state \cite{asch_bocquet_maelle_data_assimilation,daley,law_da_intro}. The starting point is a ``forecast ensemble'' of state vectors, $\vx_1,\, \ldots,\, \vx_m \in \R^n$, representing (ideally) an i.i.d.\ sample from the prior at a fixed time $t$. Corresponding to this ``state space'' ensemble is an ``observation space'' ensemble, $h(\vx_1),\, \ldots,\, h(\vx_m) \in \R^d$, where $h : \R^n \to \R^d$ is a forward operator that associates each $n$-dimensional state with a $d$-dimensional observation. The variability that we wish to identify is that of the forecast ensemble in joint state-observation space, encoded by the prior covariance matrix
\begin{equation}
\cov[\vz_1,\, \ldots,\, \vz_m] \in \R^{(n + d) \times (n + d)},\quad \vz_i = \begin{bmatrix} \vx_i \\ h(\vx_i) \end{bmatrix}, \label{eqn:ensemble_da_covariance}
\end{equation}
which approximates the ``true'' prior covariance
\begin{equation*}
\mC = \lim_{M \to \infty} \cov[\vz_1,\, \ldots,\, \vz_M],
\end{equation*}
where $\{ \vz_i \}_{i = 1}^\infty$ is an infinite sequence of i.i.d.\ draws from the joint state-observation prior. This matrix reveals directions of high and low uncertainty in the forecast, and more importantly, it allows information contained in the observations to be propagated to unobserved state variables through state-observation correlations.
\end{enumerate}

Sample approximations of covariance matrices, such as those in \cref{eqn:pod_snapshot_covariance,eqn:ensemble_da_covariance}, are only statistically meaningful when the number of samples is large enough to represent the main directions of variability in the underlying probability distribution. The number of such directions is quantified by the \emph{effective dimension} of the distribution's covariance matrix $\mC$, defined as
\begin{equation*}
D(\mC) = \frac{\trace(\mC)}{\| \mC \|_2} = \sum_{i \geq 1} \frac{\lambda_i(\mC)}{\lambda_1(\mC)},
\end{equation*}
where $\trace(\cdot)$ is the trace, $\| \cdot \|_2$ is the spectral norm, and $\lambda_i(\cdot)$ is the $i\nth$ largest eigenvalue. The following result, which we paraphrase from \cite[section 5.6]{vershynin_hdp}, establishes the requisite sample size in terms of $D(\mC)$.
\begin{theorem} \label{theorem:classical_covariance_bound}
Let $\rvX$ be a random variable on $\R^n$ satisfying $\E[\rvX] = \vzero$, $\E[\| \rvX \|_2^2] < \infty$, and
\begin{equation*}
\| \rvX \|_2^2 \leq C_1 \E [\| \rvX \|_2^2] \:\:\mathrm{a.s.}
\end{equation*}
for some $C_1 > 0$. Let $\rvX_1,\, \ldots,\, \rvX_m$ be i.i.d.\ draws from the distribution of $\rvX$, and let $\widetilde{\mC}_m = m^{-1} \sum_i \rvX_i\rvX_i\tp$ be a sample approximation of $\mC \defeq \cov[\rvX]$. Then,
\begin{equation*}
\E \| \widetilde{\mC}_m - \mC \|_2 \leq C_2 \| \mC \|_2 \left( \sqrt{\frac{C_1 D(\mC) \log n}{m}} + \frac{C_1 D(\mC) \log n}{m} \right)
\end{equation*}
where $C_2 > 0$ is a universal constant.
\end{theorem}

Model reduction and ensemble data assimilation share a common challenge, which is that the number of available snapshots or ensemble members is almost always too small to meet the $\cO(D(\mC) \log n)$ threshold established by \cref{theorem:classical_covariance_bound}. This is largely due to computational limitations, since in both settings, snapshot/ensemble vectors represent outputs of complex numerical simulations that are costly to run. In this situation, obtaining a useful covariance estimate requires regularizing the noisy small-sample statistics with information known \emph{a priori} about either the underlying covariance matrix, or the error statistics of the sample covariance estimator itself. These regularization strategies are what we will review next.

\subsection{Regularizing Small-Sample Covariance Estimates} \label{subsection:regularization_methods}

\indent

The literature on regularization strategies for small-sample covariance estimation is vast, and the review we provide here is not meant to be complete. Pourahmadi \cite{pourahmadi_covariance} and Vishny et al.\ \cite{vishny_covariance} provide excellent overviews of regularized covariance estimation across statistical science, while Bannister \cite{bannister_covar_1, bannister_covar_2} reviews covariance estimation strategies for atmospheric data assimilation specifically. The review by Vishny et al.\ \cite{vishny_covariance} particularly influenced our own presentation. For our purposes, the relevant strategies for regularized covariance estimation can be broadly grouped into five categories, described below. In what follows,
\begin{equation*}
\widetilde{\mC}_m \defeq \frac{1}{m} \sum_{i = 1}^m \rvX_i\rvX_i\stp
\end{equation*}
denotes the unregularized sample covariance of an $n$-dimensional random variable $\rvX$, satisfying $\E[\rvX] = \vzero$, given i.i.d.\ samples $\rvX_1,\, \ldots,\, \rvX_m \sim \rvX$.
\begin{enumerate}
\item \underline{Shrinkage.} These methods are based on the observation that when $m \ll n$, the spectrum of $\widetilde{\mC}_m$ tends to severely overestimate the largest eigenvalues of $\mC$ and underestimate the smallest ones. Shrinkage attempts to correct this by regularizing the eigenvalue spectrum toward a more realistic decay profile. The simplest form of this method, first described in \cite{ledoit_wolf_shrinkage}, involves linear shrinkage towards the identity matrix, resulting in the estimator 
\begin{equation}
\widehat{\mC} = \beta \mI + (1 - \beta)\widetilde{\mC}_m,\quad \beta \in (0,1), \label{eqn:identity_linear_shrinkage}
\end{equation}
whose eigenvalue spectrum is that of $\widetilde{\mC}_m$ relaxed towards a constant value $\beta$. The optimal value of $\beta$ is relatively simple to estimate from the sample vectors \cite{ledoit_wolf_shrinkage}. Data assimilation frequently makes use of ``hybrid'' covariance estimators of the form
\begin{equation}
\widehat{\mC} = \beta \mB + (1 - \beta)\widetilde{\mC}_m,\quad \beta \in (0,1), \label{eqn:general_linear_shrinkage}
\end{equation}
where $\mB$ is a ``static background'' covariance matrix (not dependent on $\rvX_1,\, \ldots,\, \rvX_m$) derived from physics-based models \cite{bannister_covar_1, bannister_covar_2,bannister_envar_review,hamill_snyder_covariance,weaver_courtier_diffusion}. In this context, $\beta$ is usually tuned empirically. If $\beta = (m+n)^{-1}n$, then \cref{eqn:identity_linear_shrinkage,eqn:general_linear_shrinkage} are maximum \emph{a posteriori} estimators for $\mC$ under inverse Wishart priors centered on $\mI$ and $\mB$, respectively \cite{webber_morzfeld_bayesian_localization}. A more general nonlinear shrinkage estimator is
\begin{equation}
\widehat{\mC} = \mV\varphi(\mLambda)\mV\tp, \label{eqn:nonlinear_shrinkage}
\end{equation}
where $\widetilde{\mC}_m = \mV\mLambda\mV\tp$ is a spectral decomposition and $\varphi$ is a regularizing function that modifies the diagonal of $\mLambda$ \cite{pourahmadi_covariance}. Deriving optimal forms for $\varphi$ often involves the use of sophisticated tools from random matrix theory \cite{ledoit_wolf_analytical_shrinkage}, leading to regularization functions whose definitions can be quite complex. \Cref{eqn:identity_linear_shrinkage} is a special case of \cref{eqn:nonlinear_shrinkage} corresponding to $\varphi(\mLambda) = \beta\mI + (1 - \beta)\mLambda$. For more details on shrinkage, we refer to the excellent review paper by Ledoit and Wolf \cite{ledoit_wolf_shrinkage_review}.

\item \underline{Sparsification.} This encompasses a broad range of methods that impose sparsity on the covariance matrix or its inverse (the so-called \emph{precision matrix}). Sparsity in the covariance matrix encodes the assumption that each component of $\rvX$ correlates with only a small number of other components. Sparsity in the precision matrix, on the other hand, encodes an assumption about the conditional dependence structure among components of $\rvX$. Specifically, if $\rvX$ is multivariate Gaussian and $\mC^{-1}(i,j) = 0$, then $\rvX(i)$ and $\rvX(j)$ are conditionally independent given $\{ \rvX(r) \suchthat r \not\in \{ i,j \} \}$. Thus, sparsity in the precision matrix of a multivariate Gaussian $\rvX$ implies that its components are related to one another through a graphical structure, wherein $\rvX(i)$ depends on the remaining components only through a small number of ``neighbors'' corresponding to nonzero entries of $\mC^{-1}(:,i)$.

Thresholding methods enforce sparsity in the covariance matrix by forming an estimator $\widehat{\mC} = T(\widetilde{\mC}_m)$, where $T$ is a function acting element-wise to eliminate covariances that fall below a threshold magnitude $t > 0$. Common choices for $T$ include
\begin{equation*}
T_1(x) = x \cdot \ind_{\{ |x| \geq t \}} \quad\text{and}\quad T_2(x) = \sign(x)(|x| - t)_+,
\end{equation*}
called ``hard thresholding'' for $T_1$ \cite{bickel_levina_thresholding} and ``soft thresholding'' for $T_2$ \cite{rothman_generalized_thresholding}. Like most regularization strategies that act element-wise, these do not guarantee a positive-semidefinite estimate \cite{vishny_covariance}. Soft thresholding is, however, equivalent to solving an $\ell_1$-constrained optimization problem whose objective function can be modified to ensure positivity \cite{xue_l1_penalty}.

The most well-known algorithm for enforcing sparsity in the precision matrix is the graphical lasso \cite{friedman_graphical_lasso}. This method produces the estimate
\begin{equation}
\mC^{-1} \approx \widehat{\mTheta} = \argmin_{\mTheta \in \R^{n \times n},\, \mTheta \succ \mZero} \trace(\mZ_m\tp\mTheta\mZ_m) - \log\det(\mTheta) + \lambda | \mTheta |_1, \label{eqn:graphical_lasso}
\end{equation}
where $\mZ_m \defeq \frac{1}{\sqrt{m}}[\rvX_1 \:\cdots\: \rvX_m]$, $\lambda > 0$ is a tuning parameter and $| \mTheta |_1 \defeq \sum_{i,j} |\mTheta(i,j)|$. The first two terms in the cost function minimized in \cref{eqn:graphical_lasso} are the negative log-likelihood of the observed samples under the given precision model, while the last term enforces sparsity. Sparse precision matrix estimators can also be derived from the observation that in the LDL decomposition $\mC^{-1} = \mL\mD\mL\tp$, the $i\nth$ row of $\mL$ contains linear regression coefficients for $\rvX(i)$ onto $\rvX(1),\, \ldots,\, \rvX(i-1)$, with residual variances appearing on the diagonal of $\mD^{-1}$. This has led to a range of methods that estimate $\mL$ and $\mD^{-1}$ row-wise via linear regression procedures with sparsity-promoting regularization \cite{bickel_levina_regularized_covariance,wu_pourahmadi_cholesky}. 

\item \underline{Spatial localization.} These methods are targeted towards applications in spatial statistics, particularly geophysical data assimilation, where there is a natural ``correlation length scale'' associated with the underlying system. Spatial localization eliminates ``spurious'' long-range correlations by forming the estimator
\begin{equation*}
\widehat{\mC} = \mL \circ \widetilde{\mC}_m,\quad \mL(i,j) = \ell(d_{ij}),
\end{equation*}
where $d_{ij}$ is the spatial distance between grid points $i$ and $j$, $\ell : [0,\infty) \to [0,1]$ is a non-increasing function of distance, and $\circ$ is the element-wise (i.e., Schur or Hadamard) product \cite{hamill_distance_dependent_filtering,houtekamer_mitchell_enkf}. The most common choice for $\ell$ is the Gaspari-Cohn localizing function \cite{gaspari_cohn}, denoted here as $\ell_{\mathrm{GC},\mu}$, which decays smoothly from $\ell_{\mathrm{GC},\mu}(0) = 1$ to $\ell_{\mathrm{GC},\mu}(\mu) = 0$. An important advantage of $\ell_{\mathrm{GC},\mu}$ is that it produces a positive semidefinite $\mL$, thus ensuring that $\widehat{\mC}$ is positive semidefinite as well by the Schur product theorem \cite[theorem 7.5.3]{horn_johnson}. In general, $\mu$ must be carefully tuned, and its optimal value depends on the sample size $m$.

\item \underline{Correlation-based localization.} These techniques originate in the geophysical data assimilation literature, which broadly uses the term ``localization'' to describe regularization strategies for large covariance matrices; but despite their name, these methods do not rely on a locality assumption \emph{per se}. Instead of estimating element-wise correction factors based on distance, they estimate correction factors from the samples themselves. Thus, correlation-based localization produces estimators of the form
\begin{equation*}
\widehat{\mC} = \cL(\rvX_1,\, \ldots,\, \rvX_m) \circ \widetilde{\mC}_m,
\end{equation*}
where $\cL : (\R^n)^m \to [0,1]^{n \times n}$. Usually, $\cL$ serves to make small covariances smaller while leaving larger covariances relatively unaffected, following the principle that empirical estimates of small covariances tend to be more corrupted with noise. Power law correction \cite{vishny_covariance}, or PLC, does this by taking powers of the empirical correlations:
\begin{equation}
\cL_\text{PLC}^{(\beta)}(\rvX_1,\, \ldots,\, \rvX_m) = \begin{bmatrix}
|\hat{\rho}_{11}|^\beta & \cdots & |\hat{\rho}_{1n}|^\beta, \\
\vdots & \ddots & \vdots \\
|\hat{\rho}_{n1}|^\beta & \cdots & |\hat{\rho}_{nn}|^\beta
\end{bmatrix}, \label{eqn:plc_estimator}
\end{equation}
where $\hat{\rho}_{ij}$ is the $(i,j)$ empirical correlation coefficient defined by
\begin{equation}
\hat{\rho}_{ij} \defeq \frac{\hat{c}_{ij}}{\sqrt{\hat{c}_{ii} \hat{c}_{jj}}},\qquad \hat{c}_{ij} \defeq \frac{1}{m} \sum_{s=1}^m \rvX_s(i)\rvX_s(j), \label{eqn:empirical_corr}
\end{equation}
and $\beta \geq 1$. If $\beta$ is an even integer, then the Schur product theorem guarantees a positive semidefinite covariance estimate. The NICE (Noise Informed Covariance Estimation) method uses PLC localization in conjunction with a tuning procedure for $\beta$ based on the Morozof discrepancy principle \cite{vishny_covariance}. Ensemble correlations raised to a power \cite{bishop_hodyss_ecorap_1,bishop_hodyss_ecorap_2}, or ECO-RAP, modifies PLC by incorporating a smoothing transformation for robustness to noise. Correlation-based localization has many more variants that we will not mention here.
\end{enumerate}

\subsection{Hierarchical Rank Structure} \label{subsec:hrs}

\indent

Each of the estimation strategies described above depends on a regularizing structural assumption for the covariance model. The purpose of this work is to introduce estimators based on a new regularizing assumption: \emph{hierarchical rank structure}. This is a structure commonly observed in matrices that represent interactions across space, including kernel evaluation matrices \cite{rebrova_hrs_kernel}, discretizations of elliptic operators \cite{martinsson_fast_elliptic}, and most fortunately for this paper, covariance matrices \cite{saibaba_hrs_geostats}. Hierarchical rank structure has commonly been used to compress a covariance matrix for easier manipulation \emph{after} it has been estimated. We will instead use hierarchical rank structure as a regularization within the estimation procedure itself.

Informally, hierarchical rank structure corresponds to the following structural property:
\begin{center}
\emph{Interaction strengths vary more smoothly at long distances than at short distances.}
\end{center}
A somewhat more formal description of hierarchical rank structure is:
\begin{center}
\emph{Interactions between well-separated domains of space are represented by low-rank submatrices.}
\end{center}
To fully specify these ideas, we consider a spatial domain $\Omega \subseteq [-1,1)^d$ with\footnote{Restricting to $d \leq 3$ spatial dimensions is not technically necessary, but we are not aware of any realistic applications of our methods with $d \geq 4$. Furthermore, our error analysis (cf.\ \cref{thm:hcov_bound_global}) features dimension-dependent constants that grow exponentially fast as $d \to \infty$, thus limiting our theory to low-dimensional settings.} $d \in \{ 1,2,3 \}$. We will also need to specify a discretization of space; thus, we consider a set of points
\begin{equation*}
\vr_1,\, \ldots,\, \vr_n \in \Omega,
\end{equation*}
representing, for example, grid points or finite volume centers for some model of a spatial process on $\Omega$.

``Well-separated domains'' refers to subdomain pairs that are admissible.
\begin{definition} \label{def:admissibility}
Given $\Theta_1, \Theta_2 \subseteq \Omega$ and $\eta > 0$, we say that $(\Theta_1, \Theta_2)$ is an \emph{$\eta$-admissible} subdomain pair provided that
\begin{equation}
\dist(\Theta_1, \Theta_2) \geq \eta^{-1} \max\{ \diam(\Theta_1), \diam(\Theta_2) \}, \label{eqn:strong_admissibility}
\end{equation}
where $\dist(\cdot), \diam(\cdot)$ denote (respectively) Euclidian set distance and diameter.
\end{definition}
Note that other notions of admissibility, particularly weak admissibility \cite{hackbusch}, are widely used in hierarchical matrix theory as well. \Cref{def:admissibility} corresponds to so-called ``strong admissibility,'' with $\eta$ defining the ``strength'' of the admissibility criterion.

Hierarchically rank structured matrices exhibit low-rank structure at a hierarchy of spatial scales; rank-deficiency is observed in long-distance interactions between large subdomains, as well as shorter-distance interactions between small subdomains. This scale hierarchy is described by a set of Cartesian rectangles $R_i \subseteq [-1,1)^d,\, i \geq 1$, defined by recursively dividing $[-1,1)^d$ into halves ($d=1$), quadrants ($d=2$), or octants ($d=3$). Concretely, we set $R_1 \defeq [-1,1)^d$, and if $R_i = \prod_{k = 1}^d [a_k, b_k)$ for some $i \geq 1$ (where $\prod$ denotes here the Cartesian product), then
\begin{equation}
R_{2^d i + r} \defeq \prod_{k = 1}^d B([a_k,b_k),\, q_k),\qquad r = 0, 1, \ldots ,2^d-1, \label{eqn:clustertree_indexing}
\end{equation}
where $r = \sum_{k = 1}^d q_k 2^{k-1},\, q_k \in \{ 0,1 \},$ is the binary expansion of $r$, and
\begin{equation*}
B([a_k,b_k),\, q_k) \defeq \begin{cases}
\left[ a_k,\, \frac{a_k+b_k}{2} \right) & q_k = 0 \\
\left[ \frac{a_k+b_k}{2},\, b_k \right) & q_k = 1
\end{cases}
\end{equation*}
defines the bisection scheme. \Cref{fig:hrs_schematics} illustrates this division of space in $d=1$ dimensions. The division is tree-structured: if an edge is drawn between $R_i$ and $R_j$ whenever $R_j \subseteq R_i$, then $\{ R_i \suchthat i \geq 1 \}$ are the nodes of a degree-$2^d$ tree rooted at $R_1$. The children of any non-leaf domain $R_i$ are its subdomains, equal to $\{ R_j \suchthat j \in \mathrm{ch}(i) \}$ where
\begin{equation*}
\mathrm{ch}(i) \defeq \{ 2^d i + r \suchthat 0 \leq r < 2^d \}.
\end{equation*}
We may now define a hierarchy of spatial scales through the cluster tree and block tree.
\begin{definition} \label{def:hrs_trees}
The \emph{cluster tree} of $\Omega$ is a finite, complete tree $\cT_\mathrm{C}$ of degree $2^d$ whose nodes are subdomains $\Omega_i \subseteq \Omega$, $i \geq 1$, defined by the following conditions.
\begin{enumerate}
\item The root node of $\cT_\mathrm{C}$ is $\Omega_1 \defeq \Omega = R_1 \cap \Omega$.
\item The children of any non-leaf $\Omega_i$ are $\Omega_{2^d i + r} \defeq R_{2^d i + r} \cap \Omega$ for $r = 0, 1, \ldots, 2^d-1$.
\end{enumerate}
The \emph{block tree} of $\Omega$ is a finite tree $\cT_\mathrm{B}$ of degree $2^{2d}$ with nodes of the form $(\Omega_i, \Omega_j)$, where $\Omega_i, \Omega_j$ are nodes of $\cT_\mathrm{C}$. It is defined by the following conditions.
\begin{enumerate}
\item The root node of $\cT_\mathrm{B}$ is $(\Omega_1, \Omega_1)$.
\item The children of any non-leaf node $(\Omega_i, \Omega_j)$ are $\{ (\Omega_s, \Omega_t) \suchthat (s,t) \in \mathrm{ch}(i) \times \mathrm{ch}(j) \}$.
\item There exists $\eta > 0$ such that a node $(\Omega_i, \Omega_j)$ is a leaf if and only if it is $\eta$-admissible or its level\footnote{We define the level of a node recursively: the root node is at level 1, and the children of a level-$\ell$ node are at level $\ell+1$. We define the tree depth as the maximum level of any node.} is $\depth(\cT_\mathrm{C})$. Consequently, $\depth(\cT_\mathrm{B}) = \depth(\cT_\mathrm{C})$.
\end{enumerate}
\end{definition}
A cluster tree is illustrated in \cref{subfig:cluster_tree}. Note that for a cluster tree $\cT_\mathrm{C}$, the corresponding block tree $\cT_\mathrm{B}$ is completely determined by $\eta$.

There are many distinct families of matrices that can be considered ``hierarchically rank structured,'' and many different data structures have been developed to efficiently represent them and compute with them. The most well known are $\cH$ and $\cH^2$-matrices \cite{bebendorf,hackbusch} and hierarchically semiseparable matrices \cite{martinsson_hss,xia_fast_hss_algs}, and for our purposes, $\cH$-matrices are the relevant class. To define these matrices, denote for each $\Theta \subseteq [-1,1)^d$
\begin{equation*}
\idx(\Theta) \defeq [i_1,\, \ldots,\, i_t],\quad\text{where}\quad \Theta \cap \{ \vr_1,\, \ldots,\, \vr_n \} = \{ \vr_{i_1},\, \ldots,\, \vr_{i_t} \},\quad i_1 < \ldots < i_t.
\end{equation*}
Thus, $\idx(\Theta)$ denotes the vector of point indices contained in $\Theta$. If $\cI(\ell)$ denotes the indices of all cluster tree nodes at level $\ell$, then $\{ \idx(\Omega_i) \}_{i \in \cI(\ell)}$ partitions $\{ 1,\, \ldots,\, n \}$ for each $\ell$.

We may now define the simplest class of hierarchically rank-structured matrices, namely, $\cH$-matrices.
\begin{definition}
Let $\mA \in \R^{n \times n}$. Suppose the existence of an integer $k \geq 1$, a cluster tree $\cT_\mathrm{C}$ over $\Omega$, and a corresponding block tree $\cT_\mathrm{B}$ such that
\begin{equation*}
\rank\left( \mA(\idx(\Omega_i),\idx(\Omega_j)) \right) \leq k
\end{equation*}
for all admissible leaves $(\Omega_i, \Omega_j)$ of $\cT_\mathrm{B}$. We then say that $\mA$ is an \emph{$\cH$-matrix} of (hierarchical) rank $k$.
\end{definition}
\Cref{subfig:hrs_matrices} illustrates cluster trees, block trees, and $\cH$-matrices with respect to a uniform grid on $[-1,1)$.
\begin{figure}
    \begin{subfigure}{\textwidth}
        \centering
        \includegraphics[scale=.45]{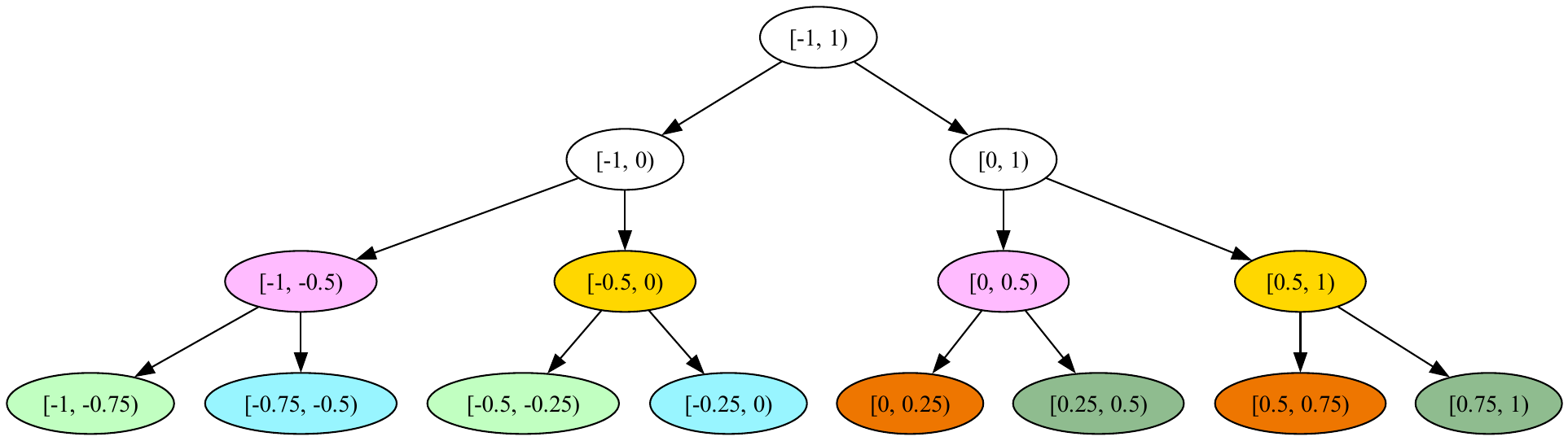}
        \caption{The first four levels of a cluster tree on $[-1,1)$. From the third level down, subdomains with the same color form a 1-admissible leaf on the block tree. The coloring scheme omits certain subdomain pairs that are also 1-admissible leaves, such as $([-1,-0.5), [0.5,1))$, $([-0.5, -0.25), [0, 0.25))$, $([-0.25, 0), [0.25, 0.5))$, and several others.} \label{subfig:cluster_tree}
    \end{subfigure}
    \medbreak
    \medbreak
    \medbreak
    \begin{subfigure}{\textwidth}
        \centering
        \includegraphics[scale=.55]{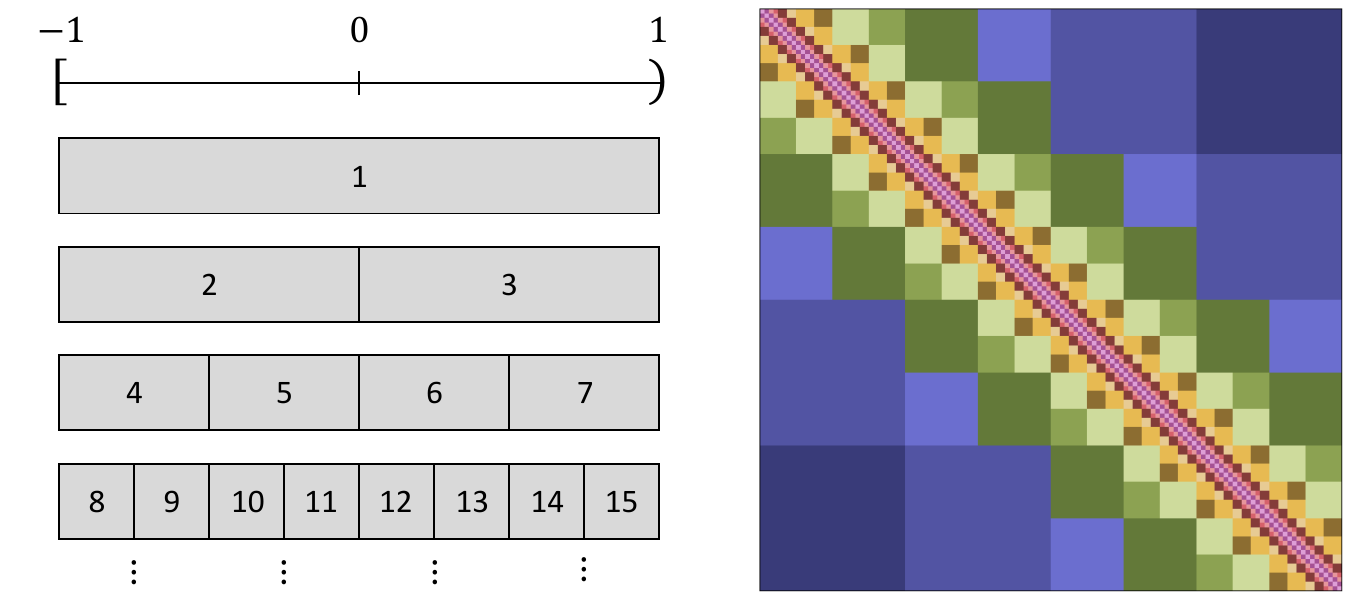}
        \caption{\underline{Left:} the recursive division of $[-1,1)$ defining the cluster tree in \cref{subfig:cluster_tree}. Each subdomain is labeled with its index as given in \cref{eqn:clustertree_indexing} with $d = 1$. \underline{Right}: an $\cH$-matrix constructed from the corresponding block tree with $\eta=1$ for the admissibility criterion. Colored panels correspond to low-rank submatrices. The rows and columns of this matrix are indexed over a uniform grid on $[-1,1)$.} \label{subfig:hrs_matrices}
    \end{subfigure}
    \caption{Schematics depicting the structure of $\cH$-matrices.}
    \label{fig:hrs_schematics}
\end{figure}

As mentioned at the outset, we will use $\cH$-matrix structure as a regularizing assumption for covariance estimation. Historically, however, $\cH$-matrices have mainly been studied for their computational efficiency properties; this motivates a great deal of our own interest as well. From this perspective, the key advantage of $\cH$-matrices is that they can be represented as trees instead of arrays. Labeling each leaf of the block tree with two factors of a rank-$k$ decomposition amounts to fully specifying an $\cH$-matrix. From here, a great deal of standard matrix operations can be performed using loops over the admissible leaves, and/or recursions over the levels of the block tree. Because the block tree depth is only logarithmic in $n$, many of the resulting algorithms have very low asymptotic complexity. \Cref{table:hmatrix_complexities} summarizes some of the standard complexity results for $\cH$-matrices; a much more thorough treatment of this subject can be found in \cite{bebendorf,hackbusch}.
\begin{table}
\centering
\def\arraystretch{1.3}
\begin{tabular}{|c|c|c|c|}
\hline
                      & Unstructured & Rank-$k$  & Rank-$k$ $\cH$-matrix \\
\hline
Storage               & $\cO(n^2)$   & $\cO(nk)$ & $\cO(nk \log n)$ \\
\hline
Matrix-vector product & $\cO(n^2)$   & $\cO(nk)$ & $\cO(nk \log n)$ \\
\hline
Inversion             & $\cO(n^3)$   & N.A.       & $\cO(n^2 k (\log n)^2)$ \\
\hline
\end{tabular}
\caption{Storage and time complexities associated with unstructured matrices, low-rank matrices, and $\cH$-matrices, all of size $n \times n$.} \label{table:hmatrix_complexities}
\end{table}
    \section{An $\cH$-Matrix Covariance Estimator} \label{section:hcov}

\indent

Having established the basic framework of hierarchical rank structure, we are now ready to show how these ideas can be used for covariance matrix estimation. In this section we will develop a method that estimates the covariance as an $\cH$-matrix. We will find that this estimator has excellent approximation accuracy for appropriately structured problems, though it will not support rapid sample generation, and it will not in general be positive semidefinite. These difficulties will be resolved in \cref{section:rscov}, where we develop an estimator using a different hierarchical matrix format.

\subsection{Problem Formulation}

\indent

We seek to estimate a covariance matrix $\mC \in \R^{n \times n}$ defined as
\begin{equation*}
\mC(i,j) = \cC(\vr_i,\vr_j),\qquad 1 \leq i,j \leq n,
\end{equation*}
where $\cC : \Omega \times \Omega \to \R$ is a symmetric positive definite\footnote{This condition means that the matrix $[\cC(\vr_i,\vr_j)]_{i,j=1}^n$ is symmetric positive semidefinite for any choice of $\vr_1,\, \ldots,\, \vr_n$.} covariance kernel. The available data are modeled as a collection of samples
\begin{equation*}
\rvX_1,\, \ldots,\, \rvX_m \in \R^n,\qquad \rvX_j(i) \defeq F_j(\vr_i),
\end{equation*}
where $F_1,\, \ldots,\, F_m : \Omega \to \R$ are drawn i.i.d.\ from the Gaussian process $\mathrm{GP}(0,\cC)$. Our key regularizing assumption is that of \emph{asymptotic smoothness}; this is the property that will endow $\mC$ with hierarchical rank structure, allowing it to be closely approximated by an $\cH$-matrix. A formal definition follows.
\begin{assumption} \label{assumption:asymptotic_smoothness}
We assume that $\cC$ is the restriction to $\Omega \times \Omega$ of a kernel
\begin{equation*}
\cC_* : [-1,1)^d \times [-1,1)^d \to \R
\end{equation*}
that is analytic on $\{ (\vx,\vy) \in (0,1)^d \times (0,1)^d \suchthat \vx \neq \vy \}$. Furthermore, we assume that there exist $c,\gamma > 0$ such that for any multi-index $\alpha \in \N^d$ and for all $\vx,\vy \in \Omega$,
\begin{equation}
\max\{ |\partial_\vx^\alpha \cC_*(\vx,\vy)|, |\partial_\vy^\alpha \cC_*(\vx,\vy)| \} \leq c \cdot \frac{|\alpha|!\gamma^{|\alpha|}\| \cC \|_\infty}{\| \vx - \vy \|_2^{|\alpha|}} \label{eqn:asymptotic_smoothness}
\end{equation}
where $\| \cC \|_\infty \defeq \sup_{\vx,\vy \in \Omega} |\cC(\vx,\vy)| = \sup_{\vx \in \Omega} |\cC(\vx,\vx)|$. We refer to \cref{eqn:asymptotic_smoothness} as an \emph{asymptotic smoothness} condition.
\end{assumption}

\subsection{General Framework}

\indent

As explained in \cref{subsec:hrs}, specifying an $\cH$-matrix is equivalent to specifying two factors of a low-rank approximation for each admissible leaf of the block-tree, and specifying an arbitrary matrix (albeit with appropriate size) for each non-admissible leaf. Hence, our $\cH$-matrix estimator amounts to $\cO(n)$ separate estimators run in parallel, each estimating an admissible cross-covariance submatrix or inadmissible submatrix at the bottom of the block tree. We call this algorithm \text{HCov}; \cref{alg:hcov} shows, in pseudocode, its main steps. Note that for each leaf $(\Omega_i,\Omega_j)$ of the block tree, symmetry is enforced between the submatrices corresponding to $(\Omega_i,\Omega_j)$ and $(\Omega_j,\Omega_i)$.

The most important work happens in line \ref{line:estimate_block}, where the function \text{RegCC} (``regularized cross-covariance'') forms matrices $\mV_i \in \R^{|\idx(\Omega_i)| \times r},\, \mV_j \in \R^{|\idx(\Omega_j)| \times r}$, and $\mT_{i,j} \in \R^{r \times r}$ constituting a rank-$r$ factorization of the estimated $\Omega_i \times \Omega_j$ cross-covariance block. We will explain this function shortly, in \cref{subsec:polyappx_block_estimator}.
\begin{algorithm}
\caption{\text{HCov} ($\cH$-matrix covariance estimator).} \label{alg:hcov}
\begin{algorithmic}[1]
\STATE \textbf{input:} samples $\rvX_1,\, \ldots,\, \rvX_m \in \R^n$, approximation degree $k \geq 1$, spatial domain $\Omega \subseteq [-1,1)^d$, spatial dimension $d \geq 1$, admissibility strength $\eta > 0$.
\STATE \textbf{output:} an $\cH$-matrix $\widehat{\mC}_\mathrm{HM}$ approximating $\cov[\rvX]$.
\item[]
\STATE $\cT \leftarrow \text{block\_tree}(\Omega,d,n,\eta,\mu)$.
\STATE $L \leftarrow \text{depth}(\cT)$.
\STATE $\widehat{\mC}_\mathrm{HM} \leftarrow \text{zeros}(n,n)$.
\FOR{all leaves $(\Omega_i,\Omega_j)$ of $\cT$ with $i \leq j$}
\STATE $\rvZ_{i,s} \leftarrow \rvX_s(\idx(\Omega_i))$ for $s = 1,\, \ldots,\, m$.
\STATE $\rvZ_{j,s} \leftarrow \rvX_s(\idx(\Omega_j))$ for $s = 1,\, \ldots,\, m$.
\IF{$\text{level}((\Omega_i,\Omega_j)) = \depth(\cT)$}
\STATE $\widehat{\mC}_\mathrm{HM}(\idx(\Omega_i),\idx(\Omega_j)) \leftarrow \frac{1}{m}\sum_{s = 1}^m \rvZ_{i,s}\rvZ_{j,s}\stp$. \label{line:hcov_bottomlevel_estim_1}
\STATE $\widehat{\mC}_\mathrm{HM}(\idx(\Omega_j),\idx(\Omega_i)) \leftarrow \frac{1}{m}\sum_{s = 1}^m \rvZ_{j,s}\rvZ_{i,s}\stp$. \label{line:hcov_bottomlevel_estim_2}
\ELSE
\STATE $\mV_i,\mT_{i,j},\mV_j \leftarrow \text{RegCC}([\rvX_1 \:\cdots\: \rvX_m], (\Omega_i,\Omega_j), k)$. \textit{\# cf.\ \cref{alg:polyappx_crosscovar_estimator}} \label{line:estimate_block}
\STATE $\widehat{\mC}_\mathrm{HM}(\idx(\Omega_i),\idx(\Omega_j)) \leftarrow \mV_i\mT_{i,j}\mV_j\tp$.
\STATE $\widehat{\mC}_\mathrm{HM}(\idx(\Omega_j),\idx(\Omega_i)) \leftarrow \mV_j\mT_{i,j}\tp\mV_i\tp$.
\ENDIF
\ENDFOR
\item[]
\RETURN $\widehat{\mC}_\mathrm{HM}$
\end{algorithmic}
\end{algorithm}

Note that, to implement \text{HCov} in a space-efficient manner, the estimator should not be stored as an $n \times n$ array; rather, it should be stored as a tree (isomorphic to the block tree) where each admissible leaf $(\Omega_i,\Omega_j)$ stores the factors $\mV_i,\mT_{i,j}$ and $\mV_j$ from line \ref{line:estimate_block}, and each nonadmissible leaf directly stores the estimated submatrix. Because of symmetry, we need only store the leaves $(\Omega_i,\Omega_j)$ where $i \leq j$. We have omitted these storage details from \cref{alg:hcov} for simplicity, but our numerical experiments will use this compressed representation to benchmark storage complexity.

\subsection{Estimation via Polynomial Approximation} \label{subsec:polyappx_block_estimator}

\indent

To estimate a large cross-covariance submatrix under sample size constraints, it is necessary to identify a representation of the underlying covariance matrix such that the number of parameters to estimate are significantly reduced. Spatial localization, for example, reduces the number of free parameters from $\cO(n^2)$ to $\cO(n)$ by assuming that the nonzero elements of $\mC$ are all located within a certain bandwidth from the diagonal (cf.\ \cref{subsection:regularization_methods}). Our assumption of asymptotic smoothness (\cref{assumption:asymptotic_smoothness}) allows the possibility that most, or even all, of the elements of a cross-covariance block are nonzero. Luckily, \cref{lemma:local_polynomial_approximation} demonstrates that asymptotically smooth cross-covariance blocks possess low-dimensional structure of a different kind. This section describes the \text{RegCC} routine (line \ref{line:estimate_block} of \cref{alg:hcov}), which leverages this structure to estimate cross covariance matrices with provable accuracy guarantees.

To state \cref{lemma:local_polynomial_approximation}, let $\poly(d, j)$ denote the set of $d$-variate polynomials of degree up to $j$.
\begin{lemma} \label{lemma:local_polynomial_approximation}
Suppose that $\eta < (\gamma \sqrt{d})^{-1}$. Under \cref{assumption:asymptotic_smoothness}, if $(\Omega_i,\Omega_j)$ is an admissible leaf of the block tree, then  there exists for each $\vy \in \Omega_j$ and $k \geq 1$ a polynomial $T_\vy^{(k)} \in \poly(d,k-1)$ satisfying
\begin{equation*}
|\cC(\vx,\vy) - T_\vy^{(k)}(\vx)| \leq \frac{c(\gamma \eta \sqrt{d})^{k}}{1 - \gamma \eta \sqrt{d}} \| \cC \|_\infty
\end{equation*}
for all $\vx \in \Omega_i$.
\end{lemma}

\begin{proof}
This lemma is a slight modification of \cite[lemma 3.15]{bebendorf}, and its proof is essentially identical. For completeness, we give the full proof in \cref{section:local_poly_approx_proof}.
\end{proof}
\Cref{lemma:local_polynomial_approximation} shows that under asymptotic smoothness, long-distance covariance structures can be approximated using a small number of low-degree polynomial basis functions. Representing long-distance covariance structures in this basis significantly reduces the number of parameters needed to specify them. Indeed, using this parameterization, estimating the covariance structure between an admissible pair of subdomains requires estimating only a small number of polynomial coefficients.

To that end, we associate each node $\Omega_i$ of the cluster tree with the linear subspace
\begin{equation*}
\cP_i^{(k)} \defeq \left\{ \begin{bmatrix} f(\vr_{i_1}) \\ \vdots \\ f(\vr_{i_t}) \end{bmatrix} \suchthat f \in \poly(d, k-1) \right\} \subseteq \R^{|\idx(\Omega_i)|},
\end{equation*}
where $[i_1,\, \ldots,\, i_t] \defeq \idx(\Omega_i)$, and we let
\begin{equation*}
t_{k,i} \defeq \dim \cP_i^{(k)}.
\end{equation*}
We define rank-$t_{k,i}$ orthogonal projection matrices $\mP_i^{(k)} \in \R^{|\idx(\Omega_i)| \times |\idx(\Omega_i)|}$ by
\begin{equation}
\mP_i^{(k)} \vx \defeq \argmin_{\vz \in \cP_i^{(k)}} \| \vx - \vz \|_2 \label{eqn:polynomial_projector}
\end{equation}
for all $\vx \in \R^{|\idx(\Omega_i)|}$. For each admissible leaf $(\Omega_i,\Omega_j)$ of the block tree, we wish to estimate the cross-covariance submatrix
\begin{equation*}
\mC_{i,j} \defeq \mC(\idx(\Omega_i), \idx(\Omega_j))
\end{equation*}
from samples $\rvX_1,\, \ldots,\, \rvX_m \in \R^n$, recalling that these are i.i.d.\ samples from $\mathrm{GP}(0, \cC)$ evaluated on $\vr_1,\, \ldots,\, \vr_n$. To do so, we define for each $\Omega_i$ vectors
\begin{equation*}
\rvZ_{i,s} \defeq \rvX_s(\idx(\Omega_i)) \in \R^{|\idx(\Omega_i)|},\quad 1 \leq s \leq m.
\end{equation*}
Our estimator for $\mC_{i,j}$ is then defined as the orthogonal projection of the unregularized cross-covariance estimate into $\cP_i^{(k)}$ and $\cP_j^{(k)}$:
\begin{equation}
\widehat{\mC}_{i,j} \defeq \frac{1}{m} \sum_{s = 1}^m (\mP_i^{(k)}\rvZ_{i,s})(\mP_j^{(k)}\rvZ_{j,s})\tp. \label{eqn:admissible_crosscovar_projected_avg}
\end{equation}
Because of the definition of $\mP_i^{(k)}$, we have $\rank \widehat{\mC}_{i,j} \leq \min\{ t_{k,i}, t_{k,j} \}$. A low-rank factorization of $\widehat{\mC}_{i,j}$ is
\begin{equation}
\widehat{\mC}_{i,j} = \mV_i\mT_{i,j}\mV_j\tp, \label{eqn:admissible_crosscovar_estimator}
\end{equation}
where $\mV_i \in \R^{|\idx(\Omega_i)| \times t_{k,i}}$ and $\mV_j \in \R^{|\idx(\Omega_j)| \times t_{k,j}}$ are orthonormal bases for $\cP_i^{(k)}$ and $\cP_j^{(k)}$, respectively, and
\begin{equation*}
\mT_{i,j} \defeq \frac{1}{m} \sum_{s = 1}^m (\mV_i\tp\rvX_s(\idx(\Omega_i)))(\mV_j\tp\rvX_s(\idx(\Omega_j)))\tp \in \R^{t_{k,i} \times t_{k,j}}.
\end{equation*}
This procedure, which comprises \text{RegCC}, is summarized in \cref{alg:polyappx_crosscovar_estimator}.
\begin{algorithm}
\caption{\text{RegCC} (Regularized Cross Covariance Estimator)} \label{alg:polyappx_crosscovar_estimator}
\begin{algorithmic}[1]
\STATE \textbf{input:} samples $\rvX_1,\, \ldots,\, \rvX_m \in \R^n$, admissible block tree node $(\Omega_i,\Omega_j)$, approximation degree bound $k \geq 1$.
\STATE \textbf{output:} matrices $\mV_i \in \R^{|\idx(\Omega_i)| \times t_{k,i}},\, \mV_j \in \R^{|\idx(\Omega_j)| \times t_{k,j}}$, and $\mT_{i,j} \in \R^{t_{k,i} \times t_{k,j}}$ giving a low-rank factorization of an estimate of $\mC(\idx(\Omega_i),\idx(\Omega_j))$.
\item[]
\STATE $\mV_i,\mV_j \leftarrow$ orthonormal bases for $\cP_i^{(k)},\cP_j^{(k)}$.
\STATE $\mT_{i,j} \leftarrow \frac{1}{m} \sum_{s = 1}^m (\mV_i\tp\rvZ_{i,s})(\mV_j\tp\rvZ_{j,s})\tp$. \label{line:poly_coeff_estim}
\item[]
\RETURN $\mV_i,\mT_{i,j},\mV_j$.
\end{algorithmic}
\end{algorithm}

\subsection{Error Analysis}

\indent

\text{HCov} permits an error analysis in terms of the constants controlling problem structure ($n,\, c,\, d$, and $\gamma$), as well as its hyperparameters. These are $k \geq 1$, the polynomial approximation degree bound, and $\eta > 0$, the admissibility parameter controlling the construction of the block tree. In \cref{thm:hcov_bound_global} below, we show how the error of \text{HCov} depends asymptotically on these hyperparameters. To state this theorem, we define
\begin{equation*}
r_{k,d} \defeq \dim \poly(d,k-1) = \binom{k+d-1}{d}.
\end{equation*}
This equality follows from a calculation in \cite[example 3.9]{bebendorf}. Also, we use the notation $x \lesssim_{\,a,b,\ldots} y$ to mean that $x \leq \alpha y$, where $\alpha \geq 0$ is a constant depending only on $a,b,\ldots$, and so on. Similarly, $x \gtrsim_{\,a,b,\ldots} y$ means that $y \lesssim_{\,a,b,\ldots} x$. Finally, recall that $\cI(\ell)$ denotes the set of indices for cluster tree nodes at level $\ell$.
\begin{theorem}\label{thm:hcov_bound_global}
Suppose that $\eta < \min\{ 1, (\gamma \sqrt{d})^{-1} \}$. Then, under \cref{assumption:asymptotic_smoothness}, the estimate $\widehat{\mC}_\mathrm{HM}$ produced by \text{HCov}\footnote{The subscript ``HM'' stands for ``$\cH$-matrix.''} satisfies
\begin{equation}
\begin{split}
\E \| \widehat{\mC}_\mathrm{HM} - \mC \|_\frob^2 &\lesssim_{\,c,d} \| \cC \|_\infty^2 \cdot \frac{Ln^2 (\eta \gamma \sqrt{d})^{2k}}{(1 - \eta \gamma \sqrt{d})^2} \\
&+ \| \mC \|_\frob^2 \cdot L\eta^{-d} \left( \frac{\max\{ r_{k,d},p \}}{m} + \sqrt{\frac{\max\{ r_{k,d},p \}}{m}} \right)^2,
\end{split} \label{eqn:hcov_global_rate}
\end{equation}
where $L$ is the depth of the cluster and block trees, and $p \defeq \max_{i \in \cI(L)} |\idx(\Omega_i)|$. Consequently, given a relative error tolerance $\epsilon \in (0,1)$, taking
\begin{equation}
k \gtrsim_{\,c, d} \frac{\log\left( \epsilon^{-1} n \sqrt{L} \right) + \log\left( \frac{1}{1 - \eta \gamma \sqrt{d}} \right) + 1}{\log\left( \frac{1}{\eta \gamma \sqrt{d}} \right)} \quad\text{and}\quad m \gtrsim_{\,c,d} \frac{L\eta^{-d} \max\{ p,r_{k,d} \}}{\epsilon^2} \label{eqn:rank_and_samplesize_rates}
\end{equation}
guarantees that $\E[\| \widehat{\mC}_\mathrm{HM} - \mC \|_\frob^2] \leq \epsilon^2 \cdot [\| \cC \|_\infty^2 + \| \mC \|_\frob^2]$.
\end{theorem}

\begin{proof}
See \cref{section:hcov_analysis}.
\end{proof}

For readers who wish to know the hidden constants in \cref{thm:hcov_bound_global}, a non-asymptotic statement of the result can be found in the appendix; see \cref{thm:hcov_global_bound_explicit}.

It is instructive to show how the asymptotic rates in \cref{thm:hcov_bound_global} reflect inherent properties of \cref{alg:hcov}. Let us discuss this term-by-term.
\begin{itemize}
\item \underline{Regularization bias.} The first term in \cref{eqn:hcov_global_rate} represents the baseline error incurred by truncating the covariance structure into a low-degree polynomial basis. Its lack of dependence on the $m$ reflects the fact that this error stems from our regularizing assumptions, rather than from anything having to do with the data; as a result, it cannot be mitigated by increasing the sample size. Notably, this error decreases exponentially fast with the truncation rank $k$, owing to \cref{lemma:local_polynomial_approximation}. The specific exponential rate depends on the admissibility strength $\eta$; this is because decreasing $\eta$ forces a greater degree of smoothness in the admissible cross-covariance blocks, thus lowering their effective rank.

\item \underline{Sample complexity.} The second term in \cref{eqn:hcov_global_rate} represents the contribution of sampling noise to the total error. Because we have parameterized the covariance structure in terms of rank-$r_{k,d}$ polynomial subspaces, this error scales with $r_{k,d}$ instead of with the total problem dimension $n$. The use of a standard outer product estimator in the polynomial subspace (cf.\ line \ref{line:poly_coeff_estim} of \cref{alg:polyappx_crosscovar_estimator}) gives \cref{eqn:rank_and_samplesize_rates} a Monte-Carlo rate with respect to $m$ and $\epsilon$. In the numerator, $p$ bounds the size of bottom-level (i.e., non-admissible) cross covariance blocks. The presence of this term reflects the fact that non-admissible blocks are estimated without any regularization, meaning that the error for these blocks is controlled by their total size, rather than by the truncation rank as for admissible blocks.

\item \underline{Dependence on admissibility strength.} The $\eta^{-d}$ factor in the second term reflects the fact that, for each cross-covariance block, \text{HCov} must ``disentangle'' two sources of variation in the sample vectors: variation due to the autocovariance of a subdomain with itself, and variation due to the cross-covariance between a subdomain pair. For each autocovariance node $(\Omega_i,\Omega_i)$ of the block tree, the number of cross-covariance leaves $(\Omega_i,\Omega_j)$ for which this error disentanglement must occur scales as $\eta^{-d}$. These issues are dealt with formally in \cref{thm:local_error_bound,lemma:diagblock_overcounting}.

\item \underline{Dependence on tree depth.} The factor of $L$ appearing in both terms arises from a technicality in the proof, relating to the accumulation of error across different levels of the block tree.
\end{itemize}

    \section{Enforcing Positive Definiteness} \label{section:rscov}

\indent

Entry-by-entry knowledge of $\mC$ is, by itself, insufficient for many practical purposes. Some situations require the ability to quickly draw new samples from a distribution (say, a multivariate Gaussian) with a specified covariance, a task which requires a computationally demanding Cholesky factorization of the covariance estimate. The $\cH$-matrix format used by \text{HCov} is somewhat helpful in this respect, since Cholesky factorization algorithms exist for positive-definite $\cH$-matrices whose asymptotic complexity is less than the characteristic $\cO(n^3)$ rate for dense matrix factorizations \cite[section 7.6]{hackbusch}. \text{HCov}, however, tends to produce estimates with negative eigenvalues, precluding an exact Cholesky factorization by \emph{any} means. The purpose of this section is to develop a rank-structured estimator for $\mC$ that is symmetric\footnote{Symmetry is already guaranteed by \text{HCov}; the challenge in this section is to enforce positivity.} positive definite (SPD) and that supports efficient factorization and sample generation, while retaining the accuracy of \text{HCov}.

\subsection{Recursive Skeletonization} \label{subsec:rskel}

\indent

Our overall strategy is to factorize $\widehat{\mC}_\mathrm{HM}$ (the \text{HCov} estimate, cf.\ \cref{section:hcov}) in a manner that reveals its spectral properties, and then to modify the factors so that they represent a factorization of an SPD matrix. For the high-dimensional problems that motivate this paper, \emph{bona fide} spectral decompositions are computationally infeasible, and an efficient alternative is needed. For this reason, we will base our method off of the recursive skeletonization factorization \cite{ho_rs_factorization,miden_rs_gp,minden_strong_rs}.

Whereas a spectral decomposition of $\widehat{\mC}_\mathrm{HM}$ reduces the matrix exactly to diagonal form via an orthogonal similarity transform, a recursive skeletonization factorization reduces $\widehat{\mC}_\mathrm{HM}$ approximately to block diagonal form via sparse, nonorthogonal congruence transforms:
\begin{equation}
\widehat{\mC}_\mathrm{HM} \approx \mL\begin{bmatrix}
\mD_1 &        &       &       \\
      & \ddots &       &       \\
      &        & \mD_p &       \\
      &        &       & \mS \\
\end{bmatrix} \mL\tp, \label{eqn:generic_rs_simplified}
\end{equation}
where $\mL$ is a product of invertible sparse matrices, $\mD_1,\, \ldots,\, \mD_p$, are small symmetric matrices whose dimension is bounded by a constant set by the user, and $\mS$ is a moderately-sized symmetric ``residual'' matrix. We will explain the details of this factorization shortly. Unlike a spectral decomposition, this factorization can be efficiently scaled to extremely large matrices; this is possible because the target form is merely block-diagonal instead of diagonal, the transformation matrices are sparse, they are not required to be orthogonal, and the factorization itself is only approximate. Nevertheless, a factorization of the form \cref{eqn:generic_rs_simplified} provides the needed information to enforce positivity; by Sylvester's principle of inertia for congruences \cite[corollary 4.5.11]{horn_johnson}, the right-hand side of \cref{eqn:generic_rs_simplified} is SPD if, and only if, $\mD_1,\, \ldots,\, \mD_p,$ and $\mS$ are each SPD. Exploiting the relatively small sizes of these matrices, we will design a method that enforces this condition explicitly.

First it is necessary to describe how a recursive skeletonization factorization is computed for an arbitrary symmetric $\mA \in \R^{n \times n}$ with respect to a given cluster tree; this procedure was originally developed in \cite{ho_rs_factorization}. We assume that $\mA$ is invertible. Consider a partition of $\{ 1,\, \ldots,\, n \}$ into subsets $\cB_1,\, \ldots,\, \cB_p$. For each $\cB_i$, a sequence of sparse transforms will ``decouple'' the majority of grid points in $\cB_i$ from the rest of $\{ 1,\, \ldots,\, n \}$. This is done as follows: let $\mPi_1$ be a permutation matrix that moves the indices in $\cB_i$ to the front, and write
\begin{equation*}
\mPi_1\mA\mPi_1\tp = \begin{bmatrix}
\mA(\cB_i,\cB_i) & \mA(\cB_i,\cF) \\
\mA(\cF,\cB_i) & \mA(\cF,\cF)
\end{bmatrix},\qquad \cF \defeq \bigcup_{j \neq i} \cB_j.
\end{equation*}
The interaction matrix between $\cB_i$ and its complement is compressed using an interpolative decomposition:
\begin{equation*}
k,\, \mT,\, \mPi_\mathrm{ID} \leftarrow \text{InterpDecomp}(\mA(\cF,\cB_i), \epsilon),
\end{equation*}
where $\epsilon > 0$ is a relative error tolerance, and \text{InterpDecomp} is an algorithm detailed in \cref{section:select_cpqr}. This yields the approximate decomposition
\begin{equation}
\mA(\cF,\cB_i)\mPi_\mathrm{ID} = \begin{bNiceMatrix}[last-row] \mA(\cF,\cR_i) & \mA(\cF,\cS_i) \\
\mbox{\scriptsize $|\cB_i| - k$} & \mbox{\scriptsize $k$}\end{bNiceMatrix} \approx \mA(\cF,\cS_i) \begin{bmatrix} \mT & \mI \end{bmatrix}. \label{eqn:generic_id}
\end{equation}
Here $k$ is an approximation rank chosen adaptively based on $\epsilon$, $\mPi_\mathrm{ID} \in \{ 0,1 \}^{|\cB_i| \times |\cB_i|}$ is a permutation that selects a group of $k$ ``skeleton'' columns comprising $\mA(\cF,\cS_i)$, and $\mT \in \R^{k \times (|\cB_i| - k)}$ is an ``interpolation'' matrix that approximates the ``redundant'' columns $\mA(\cF,\cR_i)$ in terms of the skeleton columns.

\Cref{eqn:generic_id} allows us to eliminate the $(\cF,\cR_i)$ interactions via a sparse change of variables:
\begin{align*}
\mPi_2\mA\mPi_2\tp &\approx \left[ \begin{array}{cc|c}
\mA(\cR_i,\cR_i) & \mA(\cR_i,\cS_i) & \mT\tp\mA(\cS_i,\cF) \\
\mA(\cS_i,\cR_i) & \mA(\cS_i,\cS_i) & \mA(\cS_i,\cF) \\
\hline
\Tstrut \mA(\cF,\cS_i)\mT & \mA(\cF,\cS_i) & \mA(\cF,\cF)
\end{array} \right] \\
&= (\mL_i^{(1)})\tp \left[ \begin{array}{cc|c}
\mD_i & \mY_i & \mZero \\
\mY_i\tp & \mA(\cS_i,\cS_i) & \mA(\cS_i,\cF) \\
\hline
\Tstrut \mZero & \mA(\cF,\cS_i) & \mA(\cF,\cF)
\end{array} \right] \mL_i^{(1)},
\end{align*}
where $\mPi_2$ composes the action of $\mPi_1$ and $\mPi_\mathrm{ID}$, and
\begin{align}
\mL_i^{(1)} &\defeq \left[\begin{array}{cc|c}
\mI  & \mZero  & \mZero \\
-\mT & \mI    & \mZero \\
\hline
\Tstrut \mZero & \mZero & \mI
\end{array}
\right], \label{eqn:rskel_elimination_1} \\
\mY_i &\defeq \mA(\cR_i,\cS_i) - \mA(\cR_i,\cS_i)\mT, \label{eqn:rskel_elimination_2} \\
\mD_i &\defeq \mA(\cR_i,\cR_i) - \mA(\cR_i,\cS_i)\mT - \mT\tp\mA(\cS_i,\cR_i) - \mT\tp\mA(\cS_i,\cS_i)\mT. \label{eqn:rskel_elimination_3}
\end{align}
Next, we decouple the $(\cR_i,\cR_i)$ block via a sparse transform that uses $\mD_i$ to eliminate $\mY_i$:
\begin{equation*}
\mPi_2\mA\mPi_2\tp = (\mL_i^{(1)}\mL_i^{(2)})\tp \left[ \begin{array}{cc|c}
\mD_i & \mZero & \mZero \\
\mZero & \mX_i & \mA(\cS_i,\cF) \\
\hline
\Tstrut \mZero & \mA(\cF,\cS_i) & \mA(\cF,\cF)
\end{array} \right] \mL_i^{(1)}\mL_i^{(2)},
\end{equation*}
where
\begin{equation}
\mL_i^{(2)} \defeq \left[\begin{array}{cc|c}
\mI            & -\mD_i^{-1}\mY_i & \mZero \\
\mZero         & \mI              & \mZero \\
\hline
\Tstrut \mZero & \mZero           & \mI
\end{array}
\right],\qquad \mX_i \defeq \mA(\cS_i,\cS_i) - \mY_i\mD_i^{-1}\mY_i. \label{eqn:rskel_elimination_4}
\end{equation}
By applying these elimination steps for each of $\cB_1,\, \ldots,\, \cB_p$, we arrive at a decomposition
\begin{align*}
\mPi\mA\mPi\tp &= \mL\tp \begin{bmatrix}
\mD_1  & \mZero & \cdots & \mZero & \mZero \\
\mZero & \mD_2  & \cdots & \mZero & \mZero \\
\vdots & \vdots & \ddots & \vdots & \vdots \\
\mZero & \mZero & \cdots & \mD_p  & \mZero \\
\mZero & \mZero & \cdots & \mZero & \mS
\end{bmatrix} \mL, \quad\text{where}\quad \mS \defeq \begin{bmatrix}
\mX_1                      & \mA(\cS_1,\cS_2) & \cdots & \mA(\cS_1,\cS_p) \\
\mA(\cS_2,\cS_1) & \mX_2                      & \cdots & \mA(\cS_2,\cS_p) \\
\vdots                     & \vdots                     & \ddots & \vdots                     \\
\mA(\cS_p,\cS_1) & \mA(\cS_p,\cS_2) & \cdots & \mX_p
\end{bmatrix}.
\end{align*}
What these steps have accomplished is to reduce a dense $n \times n$ matrix to a block diagonal combination of smaller matrices, $\mD_1,\, \ldots,\, \mD_p$ and $\mS$. The key point is that the sizes of $\mD_1,\, \ldots,\, \mD_p$ are controlled by the partition; $\mD_i$ has at most $|\cB_i|$ rows and columns, and $\max_i |\cB_i|$ can be made arbitrarily small by making the partition finer. As a result, operations with the $\mathrm{diag}(\mD_1,\, \ldots,\, \mD_p)$ block can be performed very rapidly, e.g., most matrix decompositions of it can be performed in $\cO(p \max_i|\cB_i|^3)$ time. Operations with $\mS$, however, may still be expensive; this block has $\sum_{i=1}^p |\cS_i|$ rows and columns, a number which can be quite large.

To see how $\mS$ is dealt with, let us write the procedure just described using the expression
\begin{equation*}
\mPi,\,\mL,\, \{\mD_i\}_{i = 1}^p,\, \mS,\, \{\cS_i\}_{i = 1}^p \leftarrow \mathrm{skeletonize}(\mA,\{ \cB_i \}_{i = 1}^p,\, \epsilon),
\end{equation*}
where $\epsilon > 0$ is an error tolerance for the approximation in \cref{eqn:generic_id}. Recursive skeletonization applies these reductions in a multilevel fashion determined by the cluster tree. First, $\mA$ is skeletonized with respect to the partition of $\{ 1,\, \ldots,\, n \}$ induced by domains at the bottom level:
\begin{equation*}
\mPi^{(L)},\,\mL^{(L)},\, \{\mD_i^{(L)}\}_{i \in \cI(L)},\, \mS^{(L)},\, \{ \cS_i^{(L)} \}_{i \in \cI(L)} \leftarrow \mathrm{skeletonize}(\mA,\, \{\cB_i^{(L)}\}_{i \in \cI(L)},\, \epsilon),
\end{equation*}
where $\cI(\ell)$ denotes the set of cluster tree indices at depth $\ell$ and $\cB_i^{(L)} \defeq \idx(\Omega_i)$. To deal with the potentially large size of $\mS^{(L)}$, this matrix is skeletonized with respect to the partition of $\bigcup_{i \in \cI(L)} \cS_i^{(L)}$ induced by the next higher level of the cluster tree. The elements of this partition are
\begin{equation*}
\cB_i^{(L-1)} \defeq \bigunion_{j \in \mathrm{ch}(i)} \cS_j^{(L)},\qquad i \in \cI(L-1),
\end{equation*}
recalling that $\mathrm{ch}(i)$ are the child indices of node $i$. Thus, the next level of skeletonization proceeds as
\begin{align*}
\mPi^{(L-1)},\,\mL^{(L-1)},\, \{\mD_i^{(L-1)}\}_{i \in \cI(L-1)},\, \mS^{(L-1)},\, \{ \cS_i^{(L-1)} \}_{i \in \cI(L-1)} \leftarrow \mathrm{skeletonize}(\mS^{(L)},\, \{ \cB_i^{(L-1)} \}_{i \in \cI(L-1)},\, \epsilon).
\end{align*}
We then skeletonize $\mS^{(L-1)}$ with respect to the partition of $\bigcup_{i \in \cI(L-1)} \cS_i^{(L-1)}$ induced by the next higher level of the cluster tree, and so on. The end result is a factorization
\begin{equation}
\mA \approx \mL\tp \begin{bmatrix}
\mD^{(L)} &             &        &           & \\
          & \mD^{(L-1)} &        &           & \\
          &             & \ddots &           & \\
          &             &        & \mD^{(2)} & \\
          &             &        &           & \mS^{(2)}
\end{bmatrix} \mL, \label{eqn:generic_rs_factorization}
\end{equation}
where $\mD^{(\ell)} \defeq \mathrm{diag}(\mD_i^{(\ell)} \suchthat i \in \cI(\ell))$, and $\mL$ composes the actions of $\{ \mPi^{(\ell)} \}_{\ell = 1}^{L-1}$ and $\{ \mL^{(\ell)} \}_{\ell = 1}^{L-1}$. Recursion halts after $\ell = 2$ because the $\ell=1$ level of the cluster tree has only a single node, $\Omega_1 = \Omega$. The partition of $\cB^{(1)} \defeq \bigunion_{i \in \cI(2)}\cS_i^{(2)}$ induced by the top level therefore consists only of $\cB^{(1)}$ itself, meaning that this level has no far field interactions left to compress.

Pseudocode for this factorization is given in \cref{alg:rskel}.

\begin{algorithm}
\caption{Recursive Skeletonization Factorization}\label{alg:rskel}
\begin{algorithmic}[1]
\STATE \textbf{input:} symmetric matrix $\mA \in \R^{n \times n}$, cluster tree of depth $L$ (cf.\ \cref{def:hrs_trees}), error tolerance $\epsilon > 0$.
\item[]
\STATE \Tstrut $\cB_i^{(L)} \leftarrow \idx(\Omega_i)$ \textbf{for} $i \in \cI(L)$.
\STATE $\mS^{(L+1)} \leftarrow \mA$.
\FOR{$\ell = L,\, L-1,\, \ldots,\, 2$}
\STATE $\mPi^{(\ell)},\, \mL^{(\ell)},\, \{ \mD_i^{(\ell)} \}_{i \in \cI(\ell)},\, \mS^{(\ell)},\, \{ \cS_i^{(\ell)} \}_{i \in \cI(\ell)} \leftarrow \mathrm{skeletonize}(\mS^{(\ell+1)},\, \{ \cB_i^{(\ell)} \}_{i \in \cI(\ell)},\, \epsilon)$
\STATE $\cB_i^{(\ell-1)} \leftarrow \bigcup_{j \in \mathrm{ch}(i)} \cS_j^{(\ell)}$ \textbf{for} $i \in \cI(\ell-1)$.
\ENDFOR
\item[]
\RETURN $\{ \mPi^{(\ell)},\, \mL^{(\ell)},\, \{ \mD_i^{(\ell)} \}_{i \in \cI(\ell)} \}_{\ell=2}^L,\, \mS^{(2)}$.
\end{algorithmic}
\end{algorithm}

\subsection{Estimation Algorithm}

\indent

The estimator we describe in this section computes a modified recursive skeletonization factorization of the estimate $\widehat{\mC}_\mathrm{HM}$ produced by \text{HCov}. Unlike in the standard version of this factorization, the matrix skeletonized at the $\ell\nth$ level is not $\mS^{(\ell+1)}$, but rather a matrix $\mP^{(\ell)} \approx \mS^{(\ell+1)}$ constructed such that the resulting factorization is of an SPD matrix. This would be guaranteed if $\mP^{(\ell)}$ itself were SPD, but this condition is challenging to enforce at deeper levels of the cluster tree where $\mS^{(\ell+1)}$ is potentially quite large. Instead we enforce a weaker condition, namely, that
\begin{equation}
\mP^{(\ell)}(\cB_i^{(\ell)},\cB_i^{(\ell)}) \succ \mZero \quad\text{and}\quad \kappa_2(\mP^{(\ell)}(\cB_i^{(\ell)},\cB_i^{(\ell)})) \leq \alpha \label{eqn:diagblock_spectral_restrictions}
\end{equation}
for each $i \in \cI(\ell)$, where $\alpha \geq 1$ is a condition number bound set by the user. We construct $\mP^{(\ell)}$ using the call
\begin{equation*}
\mP^{(\ell)} \leftarrow \text{modifyDiag}(\mS^{(\ell+1)};\: \{ \cB_i^{(\ell)} \}_{i \in \cI(\ell)},\, \alpha),
\end{equation*}
which modifies the eigenvalue spectra of diagonal blocks of $\mS^{(\ell+1)}$ so that \cref{eqn:diagblock_spectral_restrictions} is satisfied. This is detailed in \cref{alg:modify_diag}.
\begin{algorithm}
\caption{\text{modifyDiag}} \label{alg:modify_diag}
\begin{algorithmic}[1]
\STATE \textbf{input:} matrix $\mS \in \R^{b \times b}$, index sets $\{ \cB_i \}$ partitioning $1 \mcol b$, condition number bound $\alpha \geq 1$.
\STATE \textbf{output:} matrix $\mP \approx \mS$ such that, for each $i$, $\mP(\cB_i,\cB_i) \succ \mZero$ and $\kappa_2(\mP(\cB_i,\cB_i)) \leq \alpha$.
\item[]
\STATE $\mP \leftarrow \mS$.
\STATE $\beta \leftarrow \max_i \lambda_\mathrm{max}(\mS(\cB_i,\cB_i))$.
\FOR{each $i$}
\STATE $\mS(\cB_i,\cB_i) = \sum_j \lambda_j \vv_j\vv_j\tp$. \qquad\textit{\# spectral decomposition}
\STATE $\mP(\cB_i,\cB_i) \leftarrow \sum_j \max\{ \alpha^{-1}\beta,\, \lambda_j \}\vv_j\vv_j\tp$.
\ENDFOR
\item[]
\RETURN $\mP$.
\end{algorithmic}
\end{algorithm}

If $\cB_i^{(\ell)} = \cR_i^{(\ell)} \dot{\cup}\, \cS_i^{(\ell)}$ is the partition of $\cB_i^{(\ell)}$ resulting from the interpolative decomposition (cf.\ \cref{eqn:generic_id}), then the steps for decoupling the $(\cR_i^{(\ell)},\, \cR_i^{(\ell)})$ block of $\mP^{(\ell)}$ proceed as
\begin{equation*}
\begin{bmatrix}
\mP^{(\ell)}(\cR_i^{(\ell)},\cR_i^{(\ell)}) & \mP^{(\ell)}(\cR_i^{(\ell)},\cS_i^{(\ell)}) \\
\mP^{(\ell)}(\cS_i^{(\ell)},\cR_i^{(\ell)}) & \mP^{(\ell)}(\cS_i^{(\ell)},\cS_i^{(\ell)}) \\
\end{bmatrix} \mapsto \begin{bmatrix}
\mD_i^{(\ell)} & \mZero \\
\mZero & \mX_i^{(\ell)}
\end{bmatrix},
\end{equation*}
where
\begin{align}
\begin{bmatrix}
\mD_i^{(\ell)} & \mZero \\
\mZero & \mX_i^{(\ell)}
\end{bmatrix} &= \mL_i^{(\ell)} \begin{bmatrix}
\mP^{(\ell)}(\cR_i^{(\ell)},\cR_i^{(\ell)}) & \mP^{(\ell)}(\cR_i^{(\ell)},\cS_i^{(\ell)}) \\
\mP^{(\ell)}(\cS_i^{(\ell)},\cR_i^{(\ell)}) & \mP^{(\ell)}(\cS_i^{(\ell)},\cS_i^{(\ell)}) \\
\end{bmatrix} (\mL_i^{(\ell)})\tp, \label{eqn:pos_skelblock_congruence} \\
\mL_i^{(\ell)} &\defeq \begin{bmatrix}
\mI & -(\mD_i^{(\ell)})^{-1}\mY_i^{(\ell)} \\
\mZero & \mI
\end{bmatrix} \begin{bmatrix}
\mI & \mZero \\
-\mT_i^{(\ell)} & \mI
\end{bmatrix}, \nonumber
\end{align}
with $\mT_i^{(\ell)}$ being the interpolation matrix from $\cR_i^{(\ell)}$ to $\cS_i^{(\ell)}$, and $\mD_i^{(\ell)},\mY_i^{(\ell)},\mX_i^{(\ell)}$ being given in \cref{eqn:rskel_elimination_1,eqn:rskel_elimination_2,eqn:rskel_elimination_3,eqn:rskel_elimination_4}. \Cref{eqn:pos_skelblock_congruence} shows that $\mathrm{diag}(\mD_i^{(\ell)},\mX_i^{(\ell)})$ is congruent to $\mP^{(\ell)}(\cB_i^{(\ell)},\cB_i^{(\ell)})$ via a nonsingular transformation; hence, the positive semidefiniteness of $\mP^{(\ell)}(\cB_i^{(\ell)},\cB_i^{(\ell)})$ implies that $\mD_i^{(\ell)}$, too, is SPD. After $L-1$ levels of skeletonization, we arrive at an approximate factorization of $\widehat{\mC}_\mathrm{HM}$:
\begin{equation}
\widehat{\mC}_\mathrm{HM} \approx \widehat{\mC}_\mathrm{RS} \defeq \mL\tp \begin{bmatrix}
\mD^{(L)} &             &        &           & \\
          & \mD^{(L-1)} &        &           & \\
          &             & \ddots &           & \\
          &             &        & \mD^{(2)} & \\
          &             &        &           & \mP^{(1)}
\end{bmatrix} \mL. \label{eqn:rscov_factorization}
\end{equation}
Each of $\mD^{(L)},\, \mD^{(L-1)},\, \ldots,\, \mD^{(2)}$ is SPD by construction, and is $\mP^{(1)}$ is SPD as well. This last point is important to spell out in detail: skeletonization of $\mP^{(2)}$ yields $\mD^{(2)}$, which is SPD, and a residual $\mS^{(2)}$ which is potentially indefinite. The construction of $\mP^{(1)}$ from $\mS^{(2)}$ guarantees that $\mP^{(1)}(\cB_i^{(1)},\cB_i^{(1)})$ is SPD for each $i \in \cI(1)$; however, $\cI(1)$ contains only the root index of the cluster tree. Hence, $\mP^{(1)} = \mP^{(1)}(\cB_i^{(1)},\cB_i^{(1)})$ up to a permutation of indices, meaning $\mP^{(1)}$ itself is SPD.

From \cref{eqn:rscov_factorization}, we find that the procedure we have described produces, in factorized form, an SPD covariance estimator $\widehat{\mC}_\mathrm{RS}$ ``nearby'' to $\widehat{\mC}_\mathrm{HM}$. We call this procedure \text{RSCov}; pseudocode is given in \cref{alg:rscov}.
\begin{algorithm}
\caption{\text{RSCov} (recursively skeletonized SPD covariance estimator).}\label{alg:rscov}
\begin{algorithmic}[1]
\STATE \textbf{input:} samples $\rvX_1,\, \ldots,\, \rvX_m \in \R^n$, approximation degree $k \geq 1$, spatial domain $\Omega \subseteq [-1,1]^d$, spatial dimension $d \geq 1$, admissibility parameter $\eta > 0$, leaf diameter $\mu > 0$, accuracy tolerance $\epsilon > 0$, condition number bound $\alpha \geq 1$.
\STATE \textbf{output:} a recursively skeletonized PSD matrix $\widehat{\mC}_\mathrm{RS}$ approximating $\cov[\rvX]$.
\item[]
\STATE $\cT \leftarrow \text{cluster\_tree}(\Omega,d,n,\eta,\mu)$.
\STATE $L \leftarrow \text{depth}(\cT)$.
\STATE $\widehat{\mC}_\mathrm{HM} \leftarrow \text{HCov}(\rvX_1,\, \ldots,\, \rvX_m;\, \Omega,k,d,\eta,\mu)$. \qquad\emph{\# see \cref{alg:hcov}.}
\item[]
\STATE $\mS^{(L+1)} \leftarrow \widehat{\mC}_\mathrm{HM}$.
\STATE \Tstrut $\cB_i^{(L)} \leftarrow \idx(\Omega_i)$ \textbf{for} all leaves $\Omega_i$ of $\cT$.
\STATE $\mP^{(L)} \leftarrow \text{modifyDiag}(\mS^{(L+1)} ;\: \{ \cB_i^{(L)} \}_{i \in \cI(L)},\, \alpha)$ \qquad\textit{\# see \cref{alg:modify_diag}} \label{line:modifydiag_1}
\FOR{$\ell = L,\, L-1,\, \ldots,\, 2$}
\STATE $\mPi^{(\ell)},\, \mL^{(\ell)},\, \{ \mD_i^{(\ell)} \}_{i \in \cI(\ell)},\, \mS^{(\ell)},\, \{ \cS_i^{(\ell)} \}_{i \in \cI(\ell)} \leftarrow \mathrm{skeletonize}(\mP^{(\ell)},\, \{ \cB_i^{(\ell)} \}_{i \in \cI(\ell)},\, \epsilon)$
\STATE $\cB_i^{(\ell-1)} \leftarrow \bigcup_{j \in \mathrm{ch}(i)} \cS_j^{(\ell)}$ \textbf{for} $i \in \cI(\ell-1)$.
\STATE $\mP^{(\ell-1)} \leftarrow \text{modifyDiag}(\mS^{(\ell)} ;\: \{ \cB_i^{(\ell-1)} \}_{i \in \cI(\ell-1)},\, \alpha)$. \label{line:modifydiag_2}
\ENDFOR
\item[]
\RETURN $\{ \mPi^{(\ell)},\, \mL^{(\ell)},\, \{ \mD_i^{(\ell)} \}_{i \in \cI(\ell)} \}_{\ell=2}^L,\, \mP^{(1)}$.
\end{algorithmic}
\end{algorithm}
    \section{Numerical Examples} \label{section:hrs_cov_numerics}

\indent

Now we will experimentally compare \text{HCov} (\cref{alg:hcov}) and \text{RSCov} (\cref{alg:rscov}) with other covariance estimation strategies on a suite of test cases. The estimators that we compare to are the most widely-used representatives from each of the algorithm categories described in \cref{subsection:regularization_methods}. They are the following.
\begin{enumerate}
\item \underline{Unregularized.} This is simply the standard estimator, $\widetilde{\mC}_m \defeq \frac{1}{m-1} \sum_{s=1}^m (\rvX_s - \tildemu_m)(\rvX_s - \tildemu_m)\tp$, where $\tildemu_m \defeq \frac{1}{m} \sum_{s=1}^m \rvX_s$.
\item \underline{Linear shrinkage.} This estimator is given by $\widehat{\mC} \defeq \widetilde{\mD}_m(\beta\mI + (1-\beta) \widetilde{\mR}_m)\widetilde{\mD}_m$, where $\widetilde{\mD}_m$ is the diagonal matrix of empirical standard deviations, $\widetilde{\mR}_m = \widetilde{\mD}_m^{-1}\widetilde{\mC}_m\widetilde{\mD}_m^{-1}$ is the empirical correlation matrix, and $\beta \in [0,1]$ is a shrinkage parameter. We calculate $\beta$ from the sample vectors using formulas given in \cite{ledoit_wolf_shrinkage} for estimating the optimal shrinkage parameter. We use the centered and normalized samples $\widetilde{\mD}_m^{-1}(\rvX_s - \tildemu_m)$ as inputs to these formulas, so that we are approximating the optimal parameter for shrinking $\widetilde{\mR}_m$ towards the true correlation matrix.

\item \underline{Spatial localization.} This estimator is given by $\widehat{\mC} \defeq \mL \circ \widetilde{\mC}_m$ where $\mL$ is a spatial localization matrix. For 1D examples, $\mL$ is constructed from the Gaspari-Cohn function with an empirically tuned cutoff radius. For 2D examples on uniform $n_x \times n_y$ grids we set $\mL = \mL_y \otimes \mL_x$, where $\mL_x \in \R^{n_x \times n_x}$ and $\mL_y \in \R^{n_y \times n_y}$ are 1D localization matrices constructed from the Gaspari-Cohn function and $\otimes$ is the Kronecker product. We allow $\mL_x$ and $\mL_y$ to have separate cutoff radii, and we empirically tune for the pair of cutoff radii that yields the lowest error.

\item \underline{Power law correction (PLC).} This estimator is given by $\widehat{\mC} \defeq \cL_\text{PLC}^{(\beta)}(\rvX_1,\, \ldots,\, \rvX_m) \circ \widetilde{\mC}_m$, where $\cL_\text{PLC}^{(\beta)}$ is defined in \cref{eqn:plc_estimator}. We choose $\beta$ adaptively based on the samples using a procedure described in \cref{section:plc_tuning}.

\item \underline{Graphical lasso.} This is the sparse precision matrix estimator discussed in \cref{subsection:regularization_methods}. Because the available implementations of the graphical lasso are quite computationally expensive, we include comparisons with it only on a specialized lower dimensional test case. This comparison is discussed separately in \cref{section:glasso_tests}.
\end{enumerate}

We measure each strategy's performance with respect to the Frobenius norm error of the covariance estimate $\widehat{\mC}$ from the true covariance $\mC$:
\begin{equation*}
\| \widehat{\mC} - \mC \|_\frob \defeq \sqrt{ \sum_{i=1}^n \sum_{j=1}^n |\widehat{\mC}(i,j) - \mC(i,j)|^2 } = \sqrt{ \,\E_{\,\vz \sim \cN(\vzero,\mI)} [\| \widehat{\mC}\vz - \mC\vz \|_2^2] }.
\end{equation*}
The first equality is the standard definition of the Frobenius norm, while the second may be useful in settings where $\widehat{\mC}$ and $\mC$ are best thought of as linear operators. See \cref{subsec:frobnorm_identity} for its proof. We also measure each strategy's performance with respect to the spectral properties of the covariance matrix. In particular, we will measure how well the eigenspaces of each estimator capture the true covariance structure using a quantity that we call the \emph{variance fraction in subspace}, or VFIS. For a covariance matrix $\mC \in \R^{n \times n}$ and a $k$-dimensional linear subspace $\cQ \subseteq \R^n$, we define the VFIS of $\cQ$ with respect to $\mC$ as
\begin{equation*}
\mathrm{VFIS}(\cQ,\mC) \defeq \frac{\trace(\mQ\tp\mC\mQ)}{\trace(\mC)} \in [0,1],
\end{equation*}
where $\mQ \in \R^{n \times k}$ is any\footnote{If $\mQ_1$ and $\mQ_2$ are orthonormal bases for the same subspace $\cQ$, then they define the same projector $\mP = \mQ_1\mQ_1\tp = \mQ_2\mQ_2\tp$. Hence $\trace(\mQ_1\tp\mC\mQ_1) = \trace(\mC\mQ_1\mQ_1\tp) = \trace(\mC\mQ_2\mQ_2\tp) = \trace(\mQ_2\tp\mC\mQ_2)$, meaning the definition of $\mathrm{VFIS}(\cQ,\mC)$ does not depend on the particular choice of basis.} orthonormal basis for $\cQ$. If $\rvX$ is a random variable on $\R^n$ with $\E[\rvX] = \vzero$ and $\cov[\rvX] = \mC$, then $\mathrm{VFIS}(\cQ,\mC)$ measures how well a random sample $\rvZ \sim \rvX$ can be reconstructed from its orthogonal projection into $\cQ$. Indeed,
\begin{align*}
\frac{\E [\| (\mI - \mQ\mQ\tp)\rvZ \|_2^2]}{\E[\| \rvZ \|_2^2]} = \frac{\E\trace(\rvZ\stp\rvZ - \rvZ\stp\mQ\mQ\tp\rvZ)}{\E\trace(\rvZ\stp\rvZ)} &= \frac{\E\trace(\rvZ\rvZ\stp - \mQ\tp\rvZ\rvZ\stp\mQ)}{\E\trace(\rvZ\rvZ\stp)} \\
&= 1 - \mathrm{VFIS}(\cQ,\mC).
\end{align*}
Note that if $\mC = \mV\mLambda\mV\tp$ is a spectral decomposition with $\mLambda = \mathrm{diag}(\lambda_1,\, \ldots,\, \lambda_n)$ and $\lambda_1 \geq \ldots \geq \lambda_n$, then the $k$-dimensional subspace maximizing $\mathrm{VFIS}(\cQ,\mC)$ is the leading eigenspace, $\cQ = \range \mV(\mcol, 1 \mcol k)$. We prefer $\mathrm{VFIS}(\cQ,\mC)$ as a measure of subspace quality, as opposed to geometric error measures such as maximum principle angle to the true leading eigenspace, because the latter can be extremely ill conditioned with respect to tiny (even practically insignificant) perturbations of the matrix \cite{davis_kahan}. Additionally, $\mathrm{VFIS}(\cQ,\mC)$ is more operationally relevant for areas like model order reduction where eigenspaces of a covariance estimate are used to reconstruct state vectors, rather than as estimates of the true covariance's eigenspaces \emph{per se}.

Finally, we also record the storage complexity of RSCov and HCov. We specifically plot the number of floating point numbers needed to represent each of these estimators. For HCov, this means that we record the number of floating point numbers needed to represent the matrices at each leaf of the block tree, accounting for the use of low-rank formats and symmetry. We do not, however, include the memory footprint of storing the index and block trees themselves; these costs are relatively small, and can change depending on which data structure is used to represent the trees. For RSCov we record the number of floating points needed to store the nonzero entries of each matrix in the recursive skeletonization factorization, excluding the unit diagonal for block Gaussian elimination matrices. We do not record the cost of storing permutation matrices (represented in practice as vectors), nor the cost of representing the cluster tree.

\subsection{Synthetic Covariance Kernel} \label{subsec:tidal_covar}

\indent

We begin with a synthetic test case illustrating the type of covariance structure that \text{HCov} and \text{RSCov} are best suited for. In $d=1$, we consider a random spatial process $F : [-1,1) \to \R$ given by
\begin{equation}
F(x) = \alpha Z(x) + \beta \sqrt{2} \sin(2\pi x + T), \label{eqn:tidal_process}
\end{equation}
where $Z$ is drawn from a mean-zero Gaussian process with $\cov[Z(x),Z(y)] = \exp(-\frac{1}{2\sigma^2}(x-y)^2)$, $T \sim \cU[0,2\pi]$ is a random phase shift independent of $Z$,\, $\alpha^2 + \beta^2 = 1$, and $0 < \sigma \ll 1$. One can imagine $F$ as representing the height of water in a fictional 1D ``ocean''; the first term represents small-scale variations from waves, while the second term creates two opposing poles of ``high tide'' and ``low tide.'' A sample from this process is plotted in \cref{fig:tidal_covar}. In this figure, and throughout this section, we take $\alpha = \sqrt{0.8},\, \beta = \sqrt{0.2}$, and $\sigma = 10^{-2}$.

The covariance kernel associated with this process is
\begin{equation}
\cC(x,y) \defeq \cov[F(x),F(y)] = \alpha^2 \exp\left(-\frac{(x-y)^2}{2\sigma^2}\right) + \beta^2 \cos(2\pi(x-y)). \label{eqn:tidal_covar_kernel}
\end{equation}
The defining feature of this kernel is its widely separated length scales: ``shortwave'' oscillations that become decorrelated for $|x-y| \gg \sigma$ are superimposed upon ``longwave'' oscillations whose correlation structure spans all of $[-1,1)$. This kernel is asymptotically smooth with growth parameter $\gamma = 4\pi$; we prove this in \cref{section:asymptotic_smoothness_proof}. We discretize the kernel onto a uniform grid of $n=2{,}000$ points, producing a covariance matrix
\begin{equation}
\mC(i,j) \defeq \cC(x_i,x_j),\qquad 1 \leq i,j \leq n, \label{eqn:tidal_covar_matrix}
\end{equation}
where $x_i \defeq \frac{2(i-1)}{n} - 1$. This covariance matrix is shown in \cref{fig:tidal_covar}.
\begin{figure}
    \centering
    \includegraphics[scale=.8]{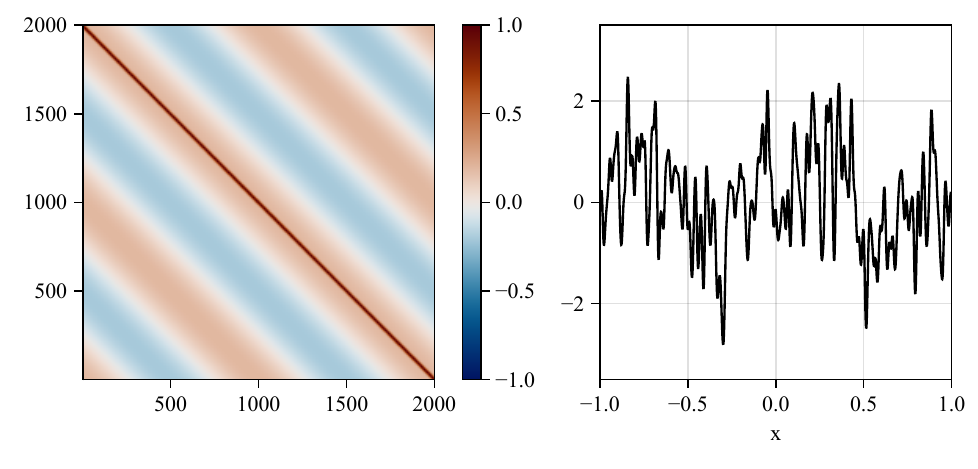}
    \caption{\underline{Left:} the covariance matrix defined in \cref{eqn:tidal_covar_matrix}. \underline{Right:} a sample from the spatial process defined in \cref{eqn:tidal_process}.} \label{fig:tidal_covar}
\end{figure}

\Cref{fig:tidal_localization} illustrates the difficulty of estimating this covariance structure using traditional Gaspari-Cohn spatial localization. The eigenvalue spectrum of $\mC$ reveals that two modes dominate the covariance structure, corresponding to the rank-2 ``longwave'' term\footnote{This term is ``rank-2'' in the sense that it has a two-term separable expansion: $\cos(2\pi(x - y)) = \cos(2\pi x)\cos(2\pi y) + \sin(2\pi x)\sin(2\pi y)$.} in \cref{eqn:tidal_covar_kernel}. If too short a localization radius is used, as in the middle panel of \cref{fig:tidal_localization}, then this longwave structure is erased. If, however, the localization radius is made longer as in the rightmost panel, then shortwave oscillations corrupt the picture; these correspond to the relatively long tail of smaller eigenvalues past the leading two.
\begin{figure}
    \centering
    \includegraphics[scale=.6]{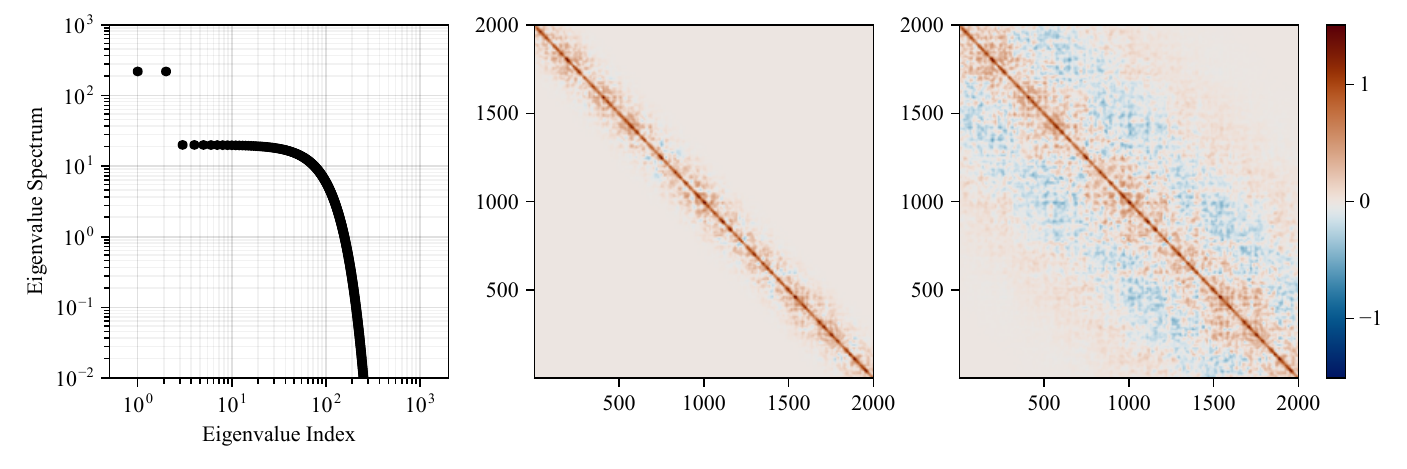}
    \caption{\underline{Left:} Eigenvalue spectrum of the covariance matrix defined in \cref{eqn:tidal_covar_matrix}. \underline{Middle:} an estimate of this covariance matrix from 40 i.i.d.\ samples from the corresponding spatial process, regularized via spatial localization with radius $0{.}2$. \underline{Right:} the same, but with localization radius $1$.} \label{fig:tidal_localization}
\end{figure}

\Cref{fig:tidal_estimation_results} shows how \text{HCov} and \text{RSCov} perform relative to other algorithms for this covariance structure. Each algorithm was used to estimate the matrix defined in \cref{eqn:tidal_covar_matrix} from a set of i.i.d.\ samples of the spatial process given in \cref{eqn:tidal_process}. For \text{HCov} and \text{RSCov}, a bisection-based cluster tree was used with leaf-level diameter $\mu = 0{.}125$, and the corresponding block tree used $\eta=1$ for the admissibility criterion. For each sample size, the polynomial approximation degree $k$ was selected by minimizing the mean error (in the Frobenius norm) over 15 independent trials. Similarly, the localization radius was tuned by minimizing the mean error over 15 independent trials. In each of this paper's experiments, spatial localization matrices were constructed using the Gaspari-Cohn function \cite{gaspari_cohn}.
\begin{figure}
    \centering
    \includegraphics[scale=0.8]{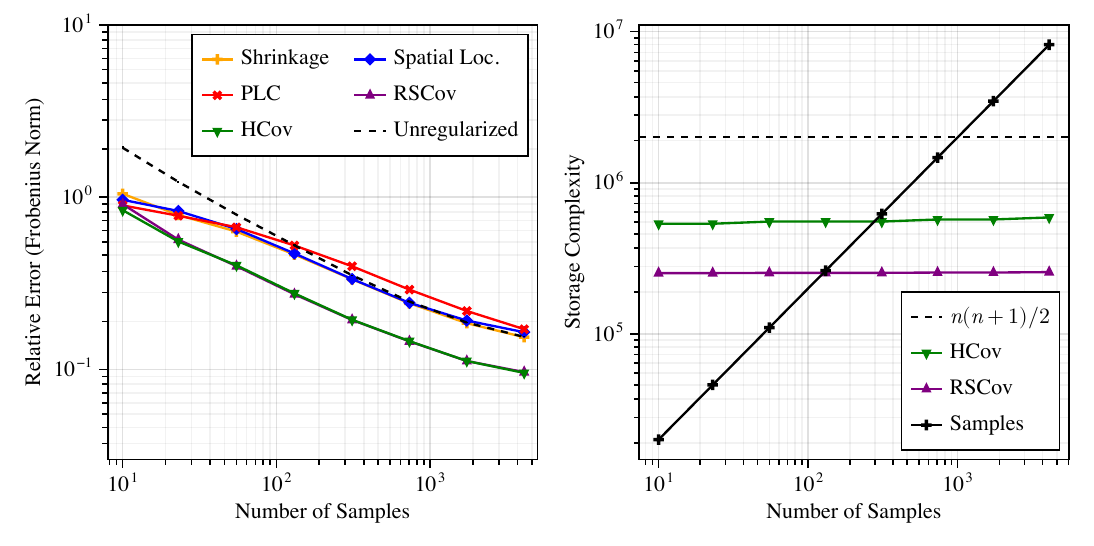}
    \caption{Performance of different covariance estimators on the example described in \cref{subsec:tidal_covar}. Each point shown represents the mean over 30 independent trials; 95\% confidence intervals for the mean are plotted in the left-hand panel, but in most cases are too small to see.} \label{fig:tidal_estimation_results}
\end{figure}

It is notable in \cref{fig:tidal_estimation_results} that shrinkage, spatial localization, and PLC perform almost as poorly as the unregularized estimator. \text{HCov} and RSCov outperform all other estimators by a wide margin, because the smoothing procedure is able to filter out short-lengthscale ``interference'' on off-diagonal blocks without destroying the much wider-lengthscale long distance correlation structure. Because \text{RSCov} modifies diagonal submatrices to preserve positive definiteness (cf.\ lines \ref{line:modifydiag_1} and \ref{line:modifydiag_2} of \cref{alg:rscov}), a small amount of extra error is introduced (visible, for example, at the 10 sample mark). However, these modifications are conservative enough that its overall performance is still quite close to \text{HCov}.

With regards to storage complexity, both \text{HCov} and \text{RSCov} occupy less memory than a full representation of the matrix. At larger sample sizes, storing these estimators is even cheaper than storing the samples themselves. For \text{HCov}, this results from the relatively small memory footprint required to store low-rank factorizations of the admissible blocks. For \text{RSCov}, this results from the high degree of sparsity in the matrices representing the elimination and permutation steps of recursive skeletonization (cf.\ \cref{subsec:rskel}).

\Cref{fig:tidal_spectra} shows the performance of \text{HCov} and \text{RSCov} with regards to the spectral properties of the covariance matrix. Similar to the results for covariance estimation error, the best performing algorithms are \text{HCov} and \text{RSCov} for smaller subspace dimensions, though PLC overtakes \text{HCov} in terms of VFIS for larger subspace dimensions. The spectral properties of the unregularized estimator and \text{RSCov} are particularly illustrative; the unregularized estimator has a rank that is limited by the sample size (55, in this case), while \text{RSCov} has a ``plateau'' in its spectrum that is likely an artifact from applications of the \text{modifyDiag} function  (cf.\ lines \ref{line:modifydiag_1} and \ref{line:modifydiag_2} of \cref{alg:rscov}).
\begin{figure}
    \centering
    \includegraphics[scale=.77]{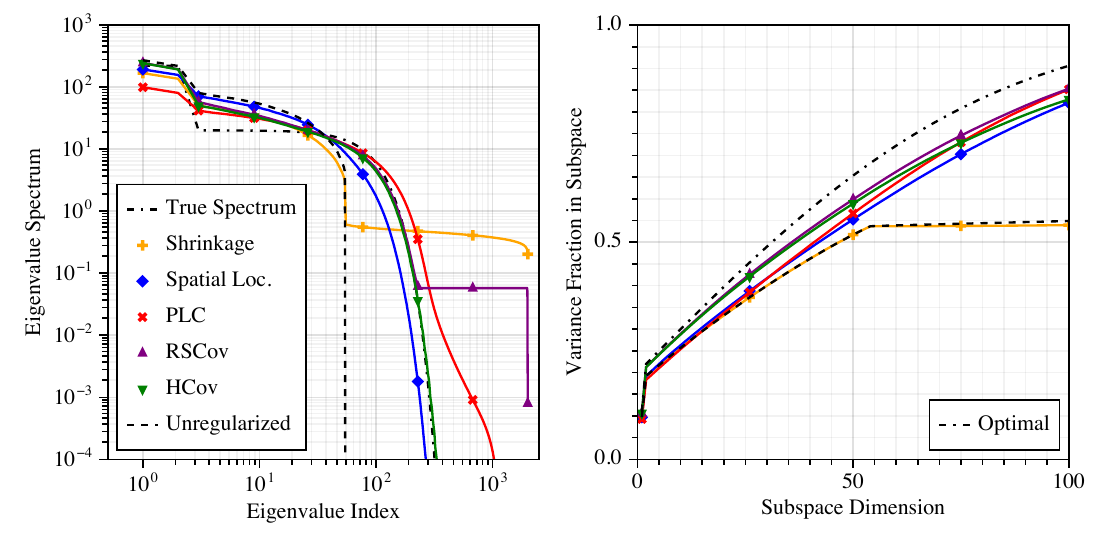}
    \caption{\underline{Left:} true and estimated eigenvalue spectra for the covariance matrix in \cref{subsec:tidal_covar}. Negative eigenvalues of the POLO and \text{HCov} estimators are not plotted. \underline{Right:} VFIS of the leading eigenspaces of different estimators of the covariance matrix, compared to the VFIS of the true leading eigenspace (``optimal''). All data shown in these plots are averaged over 30 independent trials, and all trials used 55 samples for covariance estimation.} \label{fig:tidal_spectra}
\end{figure}

\subsection{Climatological Covariance of Quasigeostrophic Flow} \label{subsec:qg_climatology}

\indent

The quasigeostrophic (QG) equations are a simplified model of atmospheric dynamics at large scales, where the flow is governed mainly by a balance of rotational (i.e., Coriolis) and pressure gradient forces. In this section, we compare the performance of \text{HCov} and \text{RSCov} relative to other algorithms for estimating the covariance of the climatological (i.e., stationary) distribution associated with these dynamics. The covariance of this distribution is useful in model reduction techniques such as proper orthogonal decomposition \cite{berkooz_pod}, since its leading eigenvectors define a low-dimensional subspace capturing a maximal amount of the dynamics' variance.

To sample from this distribution, we simulate QG dynamics of a single-layer fluid on a 2D, square domain with $\beta$-plane curvature, discretized to a $64 \times 64$ grid. The equations of motion are 
\begin{align*}
(u,v) &= (-\partial_y \psi, \partial_x \psi), \\
\psi &= \frac{g h}{f_0}, \\
Q &= f_0 + \beta y + (\nabla^2 - \ell^{-1}) \psi + \frac{f_0 h}{H}, \label{eqn:pv_def} \\
D_t Q &= -\mu Q + F,
\end{align*}
where $(u,v)$ is the velocity field, $D_t = \partial_t + u \partial_x + v \partial_y$ is the material derivative, $\psi$ is the streamfunction, $g$ is gravitational acceleration, $f_0$ is the baseline Coriolis parameter, $\beta$ is the planetary potential vorticity gradient, $h$ is the layer thickness, $H$ is the baseline layer height, $\ell$ is the Rossby deformation radius, $\mu$ is the drag coefficient, and $F$ is a stochastic forcing. We use $Q$, the potential vorticity, as the state variable. We simulate these dynamics in Julia using the GeophysicalFlows package \cite{geophysical_flows}. A thorough treatment of quasigeostrophic theory can be found in standard texts on atmospheric dynamics, such as \cite{daley}. 

\Cref{fig:qg_climatology} shows a typical state vector for this system, as well as the covariance of the stationary distribution along the state vector's $32^\mathrm{nd}$ vertical column. Although we plot only this 1D covariance slice for simplicity, our experiment involves estimating the covariance matrix of the full 2D state vector. This system is of interest to us because of the presence of relatively high-magnitude long distance correlations, clearly visible in the left-hand panel of \cref{fig:qg_climatology}. The state vector plotted in the right-hand panel gives some hint to the origins of this covariance structure; here, one can see two horizontal bands of positive and negative potential vorticity induced by rotation effects. The off-diagonal negative covariance in the left-hand panel corresponds to the cross-covariance block between these two bands. Because this covariance structure features a superposition of ``shortwave'' features from the stochastic forcing and ``longwave'' features from the banding effects, it is qualitatively similar to the kernel examined in \cref{subsec:tidal_covar}.
\begin{figure}
    \centering
    \includegraphics[scale=.8]{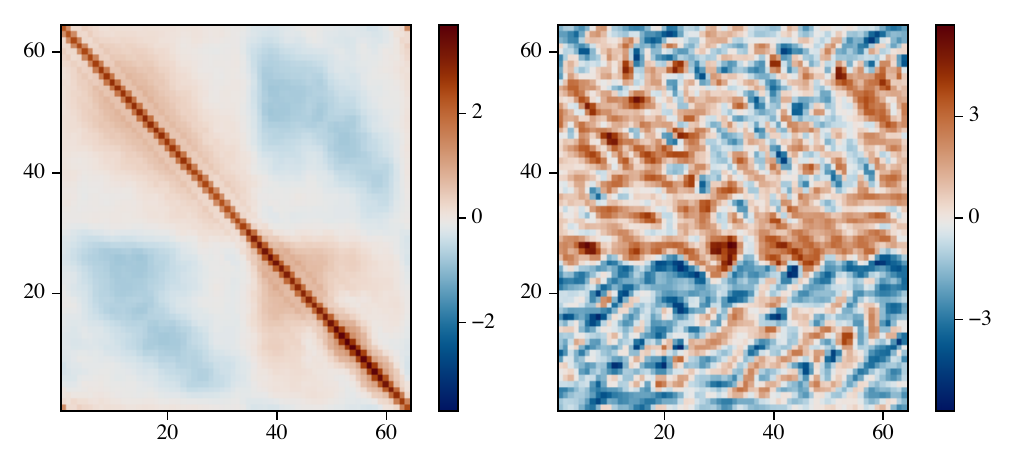}
    \caption{\underline{Left:} potential vorticity covariance of the quasigeostrophic system described in \cref{subsec:qg_climatology} along the $32^\mathrm{nd}$ vertical column of the state vector. \underline{Right:} a typical potential vorticity state vector for this system.} \label{fig:qg_climatology}
\end{figure}

\Cref{fig:qg_climatology_covarestim_results} shows how different methods perform at estimating the full covariance matrix (i.e., along both spatial dimensions) given samples from the climatology. These samples were generated by simulating a long time series of the dynamics, recording the state every $\Delta t = 20$ units of model time, and then drawing a random set of vectors from the uniform distribution over the time series. A total of $2n$ samples were created, where $n = 64^2 = 4{,}096$ is the state vector dimension. Covariance estimation errors were calculated with respect to a ``ground truth'' matrix formed by taking the empirical covariance of all $2n$ samples. For spatial localization, the horizontal and vertical localization radii\footnote{Because this problem involves a uniform 2D Cartesian grid, we use a localization matrix of the form $\mL = \mL_y \otimes \mL_x$, where $\mL_x, \mL_y$ are localization matrices constructed from the univariate Gaspari-Cohn function and $\otimes$ denotes the Kronecker product. ``Horizontal and vertical localization radii'' refers to the 1D localization radii that define $\mL_x$ and $\mL_y$.} were separately tuned for each sample size. For \text{HCov} and \text{RSCov}, the approximation rank was optimally tuned for each sample size. The hierarchical rank structure was constructed by recursively dividing the domain into quadrants, terminating when leaf-level domains had at most $8$ grid points per sidelength, and $\eta = \sqrt{2}$ was used for the admissibility criterion.
\begin{figure}
    \centering
    \includegraphics[scale=.8]{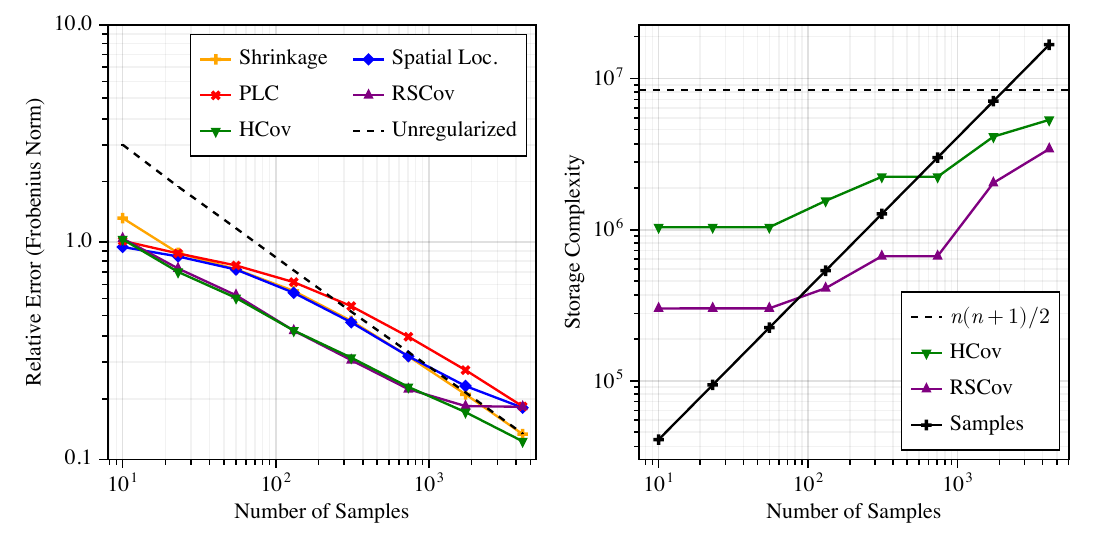}
    \caption{Covariance estimation results for the climatological distribution of the QG system described in \cref{subsec:qg_climatology}. Each point in the left-hand plot is the mean over 30 independent trials. Confidence intervals at the 95\% level in the left-hand panel are too narrow to resolve in the plot.} \label{fig:qg_climatology_covarestim_results}
\end{figure}

As in the synthetic example from \cref{subsec:tidal_covar}, the results in \cref{fig:qg_climatology_covarestim_results} highlight the difficulty of preserving smooth long-distance correlation structures with spatial localization. \text{HCov} and \text{RSCov}, on the other hand, are very successful at preserving these structures while filtering out short-wavelength noise introduced by the stochastic forcing. With respect to storage complexity, the \text{RSCov} estimator is particularly impressive; it is more than an order of magnitude cheaper to store than the full matrix for small sample sizes, and cheaper than even the samples themselves for larger sample sizes. Comparable storage efficiency is observed from \text{HCov}, though the gains are less dramatic.

\Cref{fig:qg_spectra} shows estimation results for the eigenvalue spectrum and eigenspaces given 55 samples from the climatology. Interestingly, in this case \text{HCov} lags behind spatial localization in terms of VFIS for moderate-to-large subspace dimensions, despire outperforming spatial localization in terms of covariance estimation error. It is also interesting to observe the \text{RSCov} outperforms \text{HCov} in terms of VFIS, as opposed to in \cref{fig:tidal_spectra} where the two perform nearly identically.
\begin{figure}
    \centering
    \includegraphics[scale=.77]{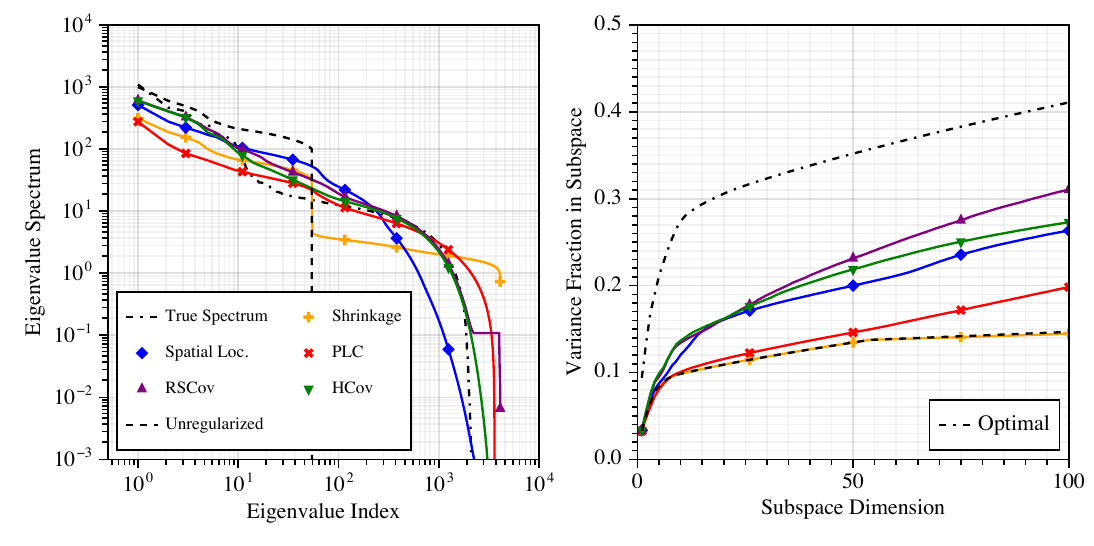}
    \caption{The same experiment as in \cref{fig:tidal_spectra}, repeated for the example in \cref{subsec:qg_climatology}.} \label{fig:qg_spectra}
\end{figure}

\subsection{Multiscale Lorenz Forecast Covariance} \label{subsec:l2s_fc_experiment}

\indent

The previous example featured a covariance matrix derived from a ``climatological'' distribution on the system state, similar to the distribution that would be sampled from to construct a general reduced-order model. This is distinct from the kind of covariance estimation problem that arises in data assimilation. In this setting, the covariance matrix being estimated arises not from a climatological distribution on the system state, but instead from a Bayesian prior distribution (or ``forecast distribution'') on the state given a past history of observations. This so-called ``forecast covariance'' matrix is strongly influenced by the dynamics governing the evolution of the state vector between observation times, producing complex covariance structures different from the ones encountered in the previous two examples.

In this section, we consider a covariance matrix arising from a forecast-like distribution for a multiscale dynamical system. The state vector at time $t$ is
\begin{equation}
\vz(t) \defeq \alpha_1 \vx(\beta_1 t) + \alpha_2 \vy(\beta_2 t), \label{eqn:twoscale_lorenz}
\end{equation}
where $\alpha_1^2 + \alpha_2^2 = 1$, $\beta_1,\, \beta_2 > 0$, and $\vx(t),\, \vy(t) \in \R^n$ evolve according to chaotic systems of ordinary differential equations. The system dimension is $n = 2{,}000$. The dynamics of $\vx(t)$ are governed by the Lorenz '96 equations \cite{lorenz_96}:
\begin{equation}
\begin{split}
\vx(t) &\defeq [x_1(t),\, \ldots,\, x_n(t)], \\
\frac{dx_i(t)}{dt} &= (x_{i+1}(t) - x_{i-2}(t))x_{i-1}(t) - x_i(t) + F,
\end{split} \label{eqn:lorenz_96}
\end{equation}
where $F = 8$, and where the state vector index $i$ is periodic. The dynamics of $\vy(t)$ are governed by the equations of ``model II'' from Lorenz '05 \cite{lorenz_05}:
\begin{equation}
\begin{split}
\vy(t) &\defeq [y_1(t),\, \ldots,\, y_n(t)], \\
\frac{dy_i(t)}{dt} &= \{ \vy(t),\vy(t) \}_{k,i} - y_i(t) + G,
\end{split} \label{eqn:lorenz_05ii}
\end{equation}
where the index is again periodic, $G = 10,\, k = 250$, and
\begin{align*}
\{ \vy(t),\vy(t) \}_{k,i} &\defeq \frac{1}{k^2} \sum_{p = -\lfloor k/2 \rfloor}^{\lfloor k/2 \rfloor} \sum_{q = -\lfloor k/2 \rfloor}^{\lfloor k/2 \rfloor} \phi_k(p) \phi_k(q) \left( y_{n-k+p-q}(t)y_{n+k+p}(t) -y_{i-2k-q}(t)y_{n-k-p}(t) \right), \\
\phi_k(p) &\defeq \begin{cases}
\frac{1}{2} & \text{$k$ is even and $|p| = k/2$} \\
1 & \text{else}
\end{cases}.
\end{align*}
The structure of $\vx(t)$ is dominated by short-wavelength oscillations spanning only a few grid points, while $\vy(t)$ shows longer-wavelength oscillations spanning hundreds of grid points. The superposition of the two in \cref{eqn:twoscale_lorenz} produces a ``multiscale'' system, the structure of which is shown in \cref{fig:l2s_fc_vis}.
\begin{figure}
    \centering
    \includegraphics[scale=1]{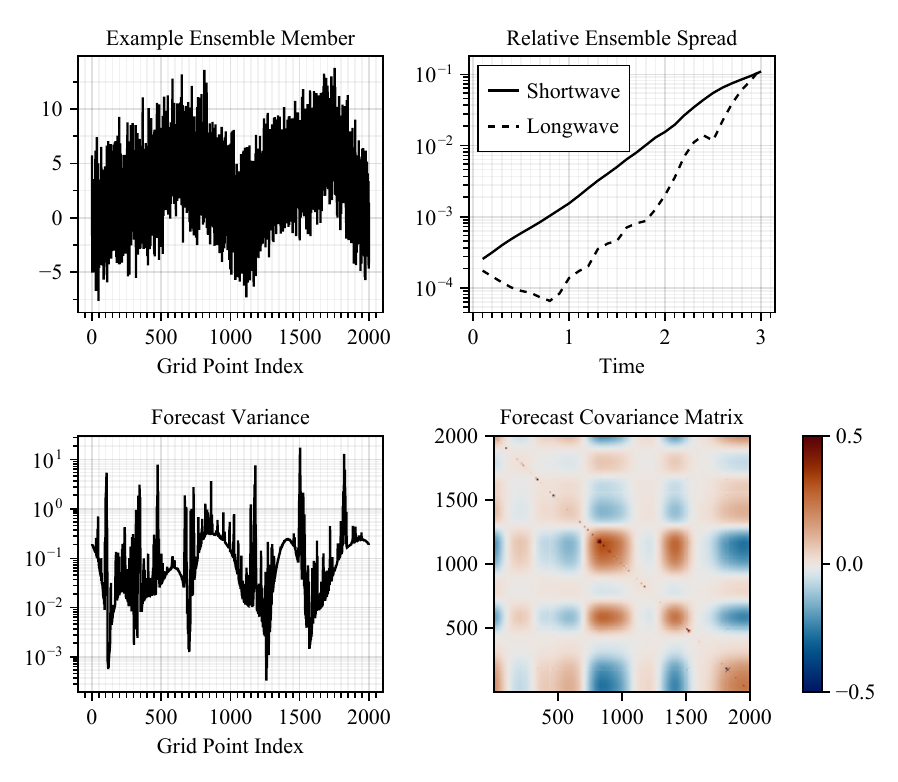}
    \caption{\underline{Top left:} an example state vector drawn from the prior distribution considered in \cref{subsec:l2s_fc_experiment}. \underline{Top right:} relative ensemble spreads during spinup time. \underline{Bottom left:} the forecast variance at each gridpoint. \underline{Bottom right:} the forecast covariance matrix. The color bar is calibrated to show the smaller-magnitude long-distance correlations produced by the longwave component, but the diagonal features spikes in variance up to $\approx 17{.}5$, and covariance magnitudes off the diagonal reach as high as $\approx 6{.}4$.} \label{fig:l2s_fc_vis}
\end{figure}

The forecast distribution is constructed as follows: initial states $\vx(0)$ and $\vy(0)$ are generated by integrating a random ``seed'' vector through \cref{eqn:lorenz_96,eqn:lorenz_05ii} for $t \in [0,10]$. We perturb these initial states with Gaussian random vectors of magnitude $10^{-4}$, producing ensembles $\{ \vx^{(i)}(0) \}_{i=1}^{M}$ and $\{ \vy^{(i)}(0) \}_{i=1}^{M}$ with $M = 4{,}000$ concentrated around $\vx(0)$ and $\vy(0)$. These ensembles are integrated through \cref{eqn:lorenz_96,eqn:lorenz_05ii} until their relative variances,
\begin{align*}
v_x(\beta_1 t) &\defeq \frac{1}{\| \vx(0) \|_2^2} \sum_{j=1}^n \var[x_j^{(1)}(\beta_1 t),\, \ldots,\, x_j^{(M)}(\beta_1 t)] \qquad\text{and} \\
v_y(\beta_2 t) &\defeq \frac{1}{\| \vy(0) \|_2^2} \sum_{j=1}^n \var[y_j^{(1)}(\beta_2 t),\, \ldots,\, y_j^{(M)}(\beta_2 t)],
\end{align*}
reach approximately $10^{-1}$. We tune $\beta_1$ and $\beta_2$ so that $v_x(\beta_1 t)$ and $v_y(\beta_2 t)$ grow at approximately the same rate; this growth in variance is illustrated in \cref{fig:l2s_fc_vis}. Finally, we combine the resulting ensembles via \cref{eqn:twoscale_lorenz} with $\alpha_1 = \sqrt{0{.}75}$ and $\alpha_2 = \sqrt{0{.}25}$, resulting in a forecast ensemble $\{ \vz^{(j)}(t) \}_{j=1}^{M}$.

Covariance estimation results for this ensemble are shown in \cref{fig:l2s_fc_results}. The block and cluster trees defining the hierarchical rank structure were constructed in an identical manner to \cref{subsec:tidal_covar}, and as before, localization radii and polynomial approximation degrees were tuned to minimize approximation error for each sample size. The overall comparison between the algorithms in the same as in \cref{subsec:tidal_covar,subsec:qg_climatology}, but the gap between \text{HCov} and \text{RSCov} is somewhat more substantial, and exhibits moderate growth for large sample sizes.
\begin{figure}
    \centering
    \includegraphics[scale=.8]{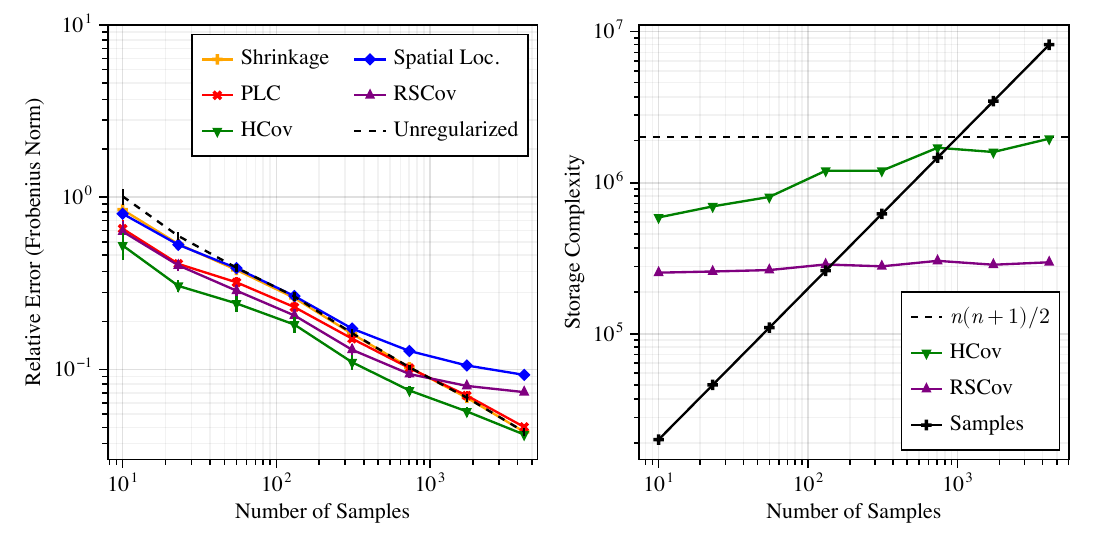}
    \caption{Covariance estimation results for the multiscale Lorenz forecast distribution of \cref{subsec:l2s_fc_experiment}. Each point in the left-hand plot is the mean over 30 independent trials. Confidence intervals at the 95\% level in the left-hand panel are too narrow to resolve in the plot.} \label{fig:l2s_fc_results}
\end{figure}

With respect to the spectral structure of the covariance matrix, \cref{fig:l2s_fc_spectra} shows that each estimation strategy performs similarly on this example. For low subspace dimensions, PLC and HCov have the best performance by a small margin in terms of VFIS. As before, covariance estimates in \cref{fig:l2s_fc_spectra} were constructed from 55 ensemble members.
\begin{figure}
    \centering
    \includegraphics[scale=.77]{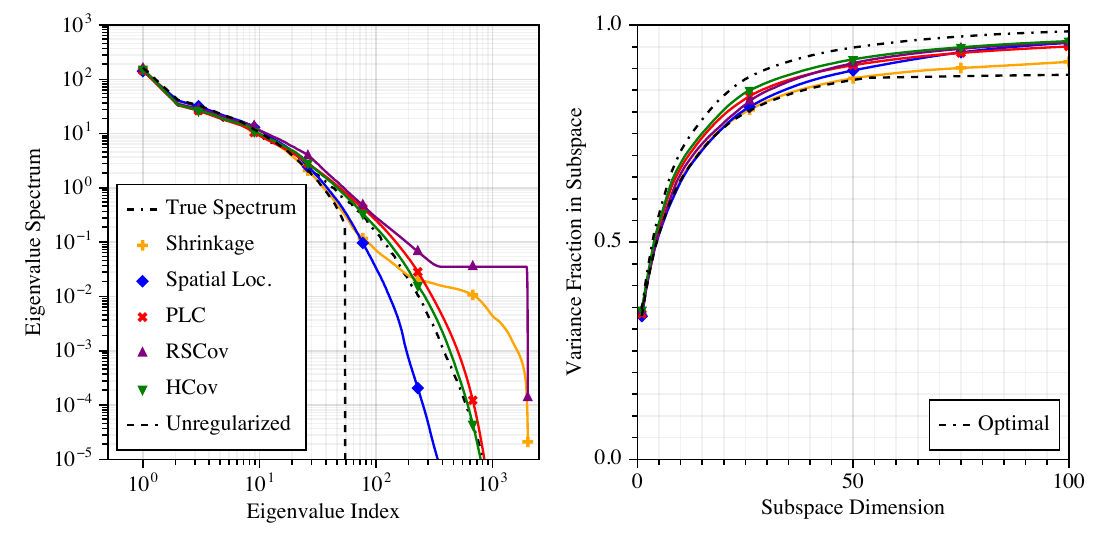}
    \caption{The same experiment as in \cref{fig:tidal_spectra}, repeated for the example in \cref{subsec:l2s_fc_experiment}.} \label{fig:l2s_fc_spectra}
\end{figure}

\subsection{Surface Quasigeostrophic Forecast Covariance} \label{subsec:sqg_forecast}

\indent

Surface quasigeostrophic (SQG) dynamics are a variant of quasigeostrophic dynamics (cf.\ \cref{subsec:qg_climatology}) involving a thin fluid of uniform potential vorticity bounded between two surfaces \cite{blumen_sqg,held_sqg}. Buoyancy evolves turbulently under these dynamics, and because of this, weather prediction and DA algorithms are often tested on SQG forecasting problems. In this section, we compare the performance of \text{HCov} and \text{RSCov} to other algorithms in estimating the covariance matrix of a forecast-like distribution for SQG turbulence.

To generate samples from this distribution, we use an open-source model\footnote{We refer here to Jeffrey S.\ Whitaker's \texttt{sqgturb} model: \url{https://github.com/jswhit/sqgturb}.} that simulates two-layer SQG turbulence on a square $f$-plane with periodic boundaries. The state variable is buoyancy, both layers of which are represented on uniform $64 \times 64$ grids. An initial condition is generated by integrating a random ``seed'' vector through 100 days of model time. Uncertainty on the initial condition is then introduced by perturbing the vector with a small amount of Gaussian noise, creating an ensemble of $8{,}192$ nearly-identical states. We integrate this ensemble forward in time, allowing its variance to grow under the influence of turbulence; this continues until the ensemble variance reaches 10\% of its steady-state value. Once this occurs, we extract the first layer of each buoyancy vector, yielding an $8{,}192$-member ensemble of 2D state vectors with dimension $n = 64^2 = 4{,}096$.

\Cref{fig:sqg_fc_vis} shows buoyancy plots for the first three members of this ensemble, and shows the forecast covariance matrix along various 1D slices of the state vector. In contrast to the previous examples, this covariance matrix is relatively sparse (though not completely zero) off the diagonal. We can therefore expect spatial localization perform fairly well. The nonzero structures that do exist off the diagonal are more highly oscillatory, and we can therefore expect some performance degradation in HCov and RSCov.
\begin{figure}
    \centering
    \includegraphics[scale=.85]{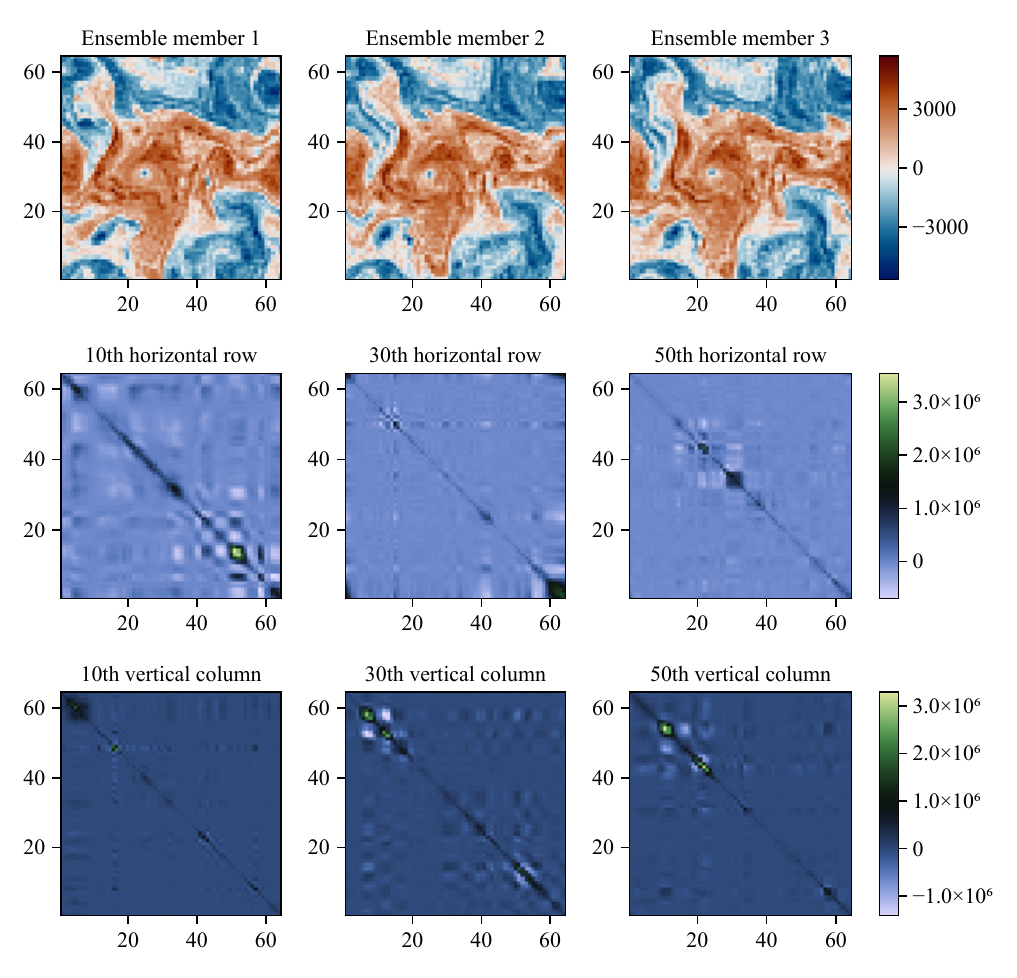}
    \caption{\underline{Top row:} the first three members of the $8{,}192$-member forecast distribution generated for the SQG system. \underline{Bottom rows:} the covariance matrix of the forecast distribution along 1D slices of the state vector. Figures in the same row use the same colorbar.} \label{fig:sqg_fc_vis}
\end{figure}

\Cref{fig:sqg_fc_polyappx} shows covariance estimation results for this example. The cluster and block trees for HCov and RSCov were constructed in the same manner as in \cref{subsec:qg_climatology}. Spatial localization outperforms HCov and RSCov, as suggested by the discussion above. At extremely low sample sizes, the covariance structure is so difficult to resolve that the regularizing bias of polynomial approximation in HCov and RSCov lowers the overall error. Thus, the optimal approximation degrees for HCov and RSCov are relatively small, allowing these estimators to have low storage complexity. As the number of samples increases, however, higher-degree approximations are needed to resolve the covariance structure, meaning that the storage complexity of HCov and RSCov quickly increases (even to the point of exceeding the $\cO(n^2)$ storage cost of the full covariance matrix). The plateau in storage complexity likely results from HCov and RSCov reaching the maximum approximation degree allowed in this experiment.
\begin{figure}
    \centering
    \includegraphics[scale=.8]{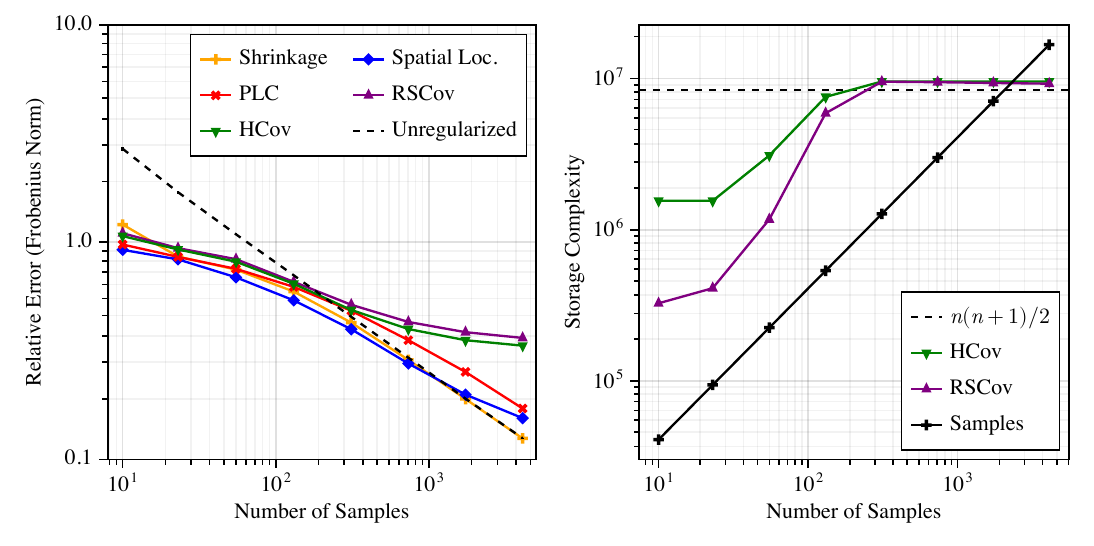}
    \caption{Covariance estimation results for the SQG forecast covariance matrix described in \cref{subsec:sqg_forecast}, with \text{HCov} and \text{RSCov} using \cref{alg:polyappx_crosscovar_estimator} to estimate cross-covariance submatrices. Each point on the left-hand plot is the mean over 30 independent trials. Confidence intervals at the 95\% level in the left-hand panel are too narrow to resolve in the plot.} \label{fig:sqg_fc_polyappx}
\end{figure}

\Cref{fig:sqg_fc_spectra} shows that for this example, spatial localization and PLC outperform HCov and RSCov in terms of the quality of the estimated subspace. This is consistent with the results in \cref{fig:sqg_fc_polyappx}, where these two estimators also outperform HCov and RSCov in terms of estimating the covariance matrix itself.
\begin{figure}
    \centering
    \includegraphics[scale=.77]{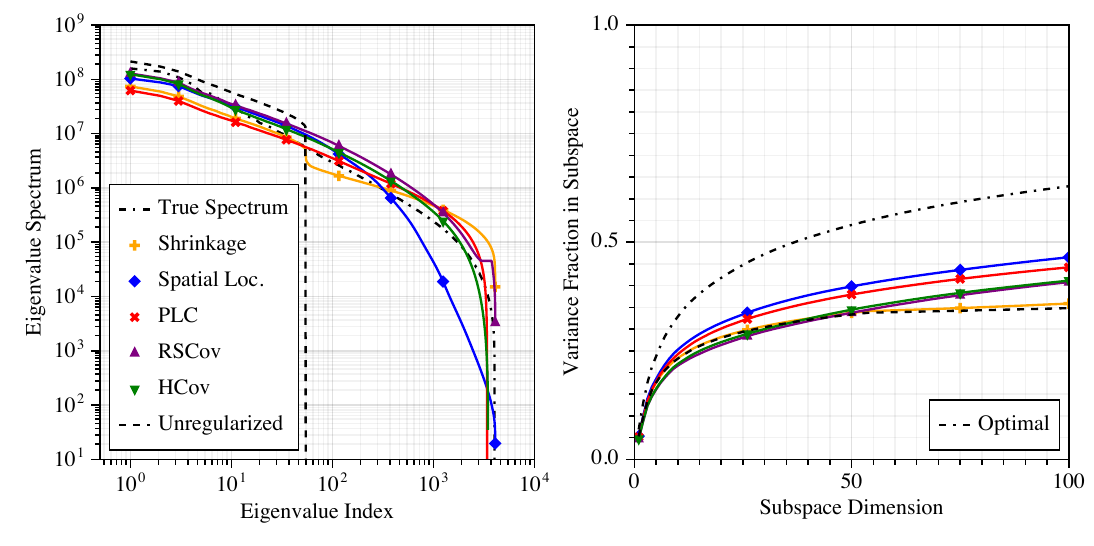}
    \caption{The same experiment as in \cref{fig:tidal_spectra}, repeated for the example in \cref{subsec:sqg_forecast}.} \label{fig:sqg_fc_spectra}
\end{figure}

    \section{Conclusions}

\indent

We have shown that for covariance matrices arising from asymptotically smooth covariance kernels, estimation error can be reduced by imposing hierarchical rank structure as a regularization. We have implemented this idea using low-degree polynomial approximation, leading to the HCov algorithm whose performance is supported by a theoretical error analysis. By combining HCov with a modified recursive skeletonization factorization, we have shown that it is possible to enforce positive definiteness in the resulting estimator.

Our experiments indicate that HCov and RSCov perform quite well at estimating covariance matrices for spatial processes defined by a superposition of ``shortwave'' and ``longwave'' features. This includes examples where asymptotic smoothness cannot be verified analytically. However, we also find that our algorithms are not universally superior to standard techniques, especially in situations where a simpler localization structure is appropriate.

\section{Acknowledgements}

\indent

Robin Armstrong and Anil Damle were funded through the Department of Energy Office of Science award DE-SC0025453. This work benefited substantially from the help of our colleagues. In particular, Pavel Sakov (Australian Bureau of Meteorology) helped guide our experiments by suggesting that we examine covariance structures arising from tidal-type forcings. Matthias Morzfeld (Scripps Institution of Oceanography) gave invaluable constructive criticism on early versions of our experiments that helped us greatly improve their quality.

    \bibliographystyle{plain}
    \bibliography{text/sources}

    \appendix

\section{Comparison With the Graphical Lasso} \label{section:glasso_tests}

\indent

We include here a numerical experiment to compare HCov and RSCov with the graphical lasso \cite{friedman_graphical_lasso}, along with the other estimators considered in \cref{section:hrs_cov_numerics}. The system considered here is a discretization of the 1D spatial process from \cref{subsec:tidal_covar}, but this time onto a uniform grid with only 50 points. We use the implementation in the GraphicalLasso.jl package\footnote{See \url{https://www.juliapackages.com/p/graphicallasso}.}, tuning the regularization strength optimally for each sample size. The tuning procedures for the other estimators are the same as in \cref{section:hrs_cov_numerics}.

\Cref{fig:tidal_small} shows covariance estimation results for this test case. The graphical lasso performs better than the unregularized estimator, but except for at high sample sizes, every other estimator does better. Interestingly, however, it performs comparatively well in terms of spectrum and eigenspace estimation; see \cref{fig:tidal_small_spectra}. \Cref{fig:precision_sparsity} gives some indication for why the graphical lasso underperforms in terms of overall accuracy; it relies on the assumption of a sparse precision matrix, whereas \cref{fig:precision_sparsity} shows that this example generates a \emph{non}-sparse precision matrix. This is intuitive, since off-diagonal covariances in this system arise from genuine long-range couplings between the crests and troughs of long-distance waves, rather than from neighbor-to-neighbor interactions that could be modeled by a sparse graphical Markov structure. This suggests that sparse precision matrix estimators in general (not just the graphical lasso) may be inappropriate for estimating covariance structures characterized by superpositions of widely separated lengthscales, including the covariance matrices in \cref{subsec:tidal_covar,subsec:l2s_fc_experiment} and, to a lesser extent, the matrix in \cref{subsec:qg_climatology}. However, numerically studying the sparsity of the precision matrix is difficult in these examples, as the extremely high condition numbers of these matrices renders us unable to invert them in a stable manner.
\begin{figure}
    \centering
    \includegraphics[scale=.8]{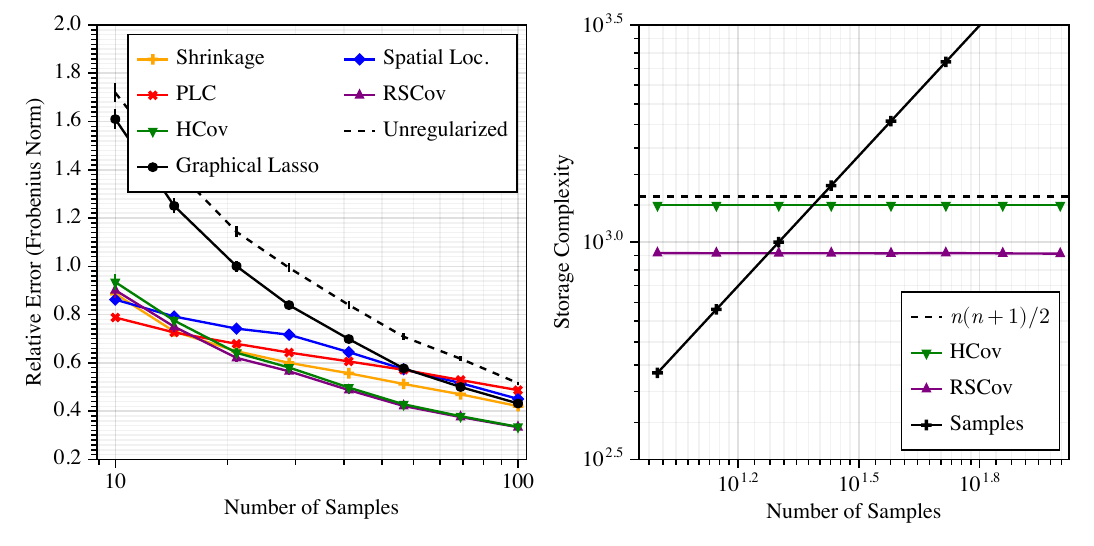}
    \caption{Covariance estimation results for a smaller version of the test case from \cref{subsec:tidal_covar}, including results for the graphical lasso. Each point in the left-hand plot is the average over 30 independent trials, and with some barely visible exceptions at the lowest sample sizes, the 95\% confidence intervals are too small to be seen.} \label{fig:tidal_small}
\end{figure}
\begin{figure}
    \centering
    \includegraphics[scale=.85]{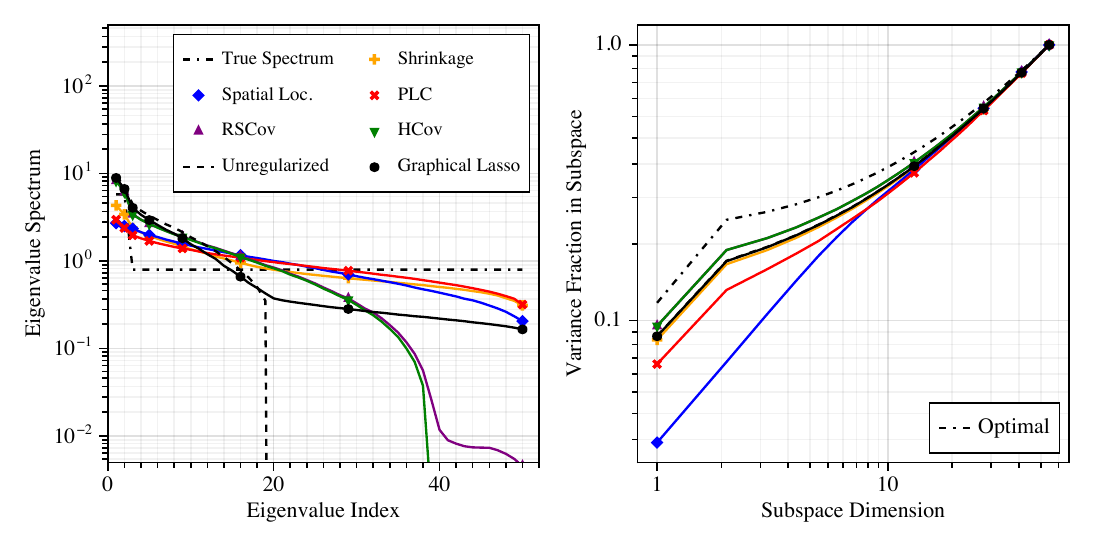}
    \caption{Spectrum estimation results for the system in \cref{section:glasso_tests} using 20 samples. Each point is the mean over 30 independent trials.} \label{fig:tidal_small_spectra}
\end{figure}
\begin{figure}
    \centering
    \includegraphics[scale=.9]{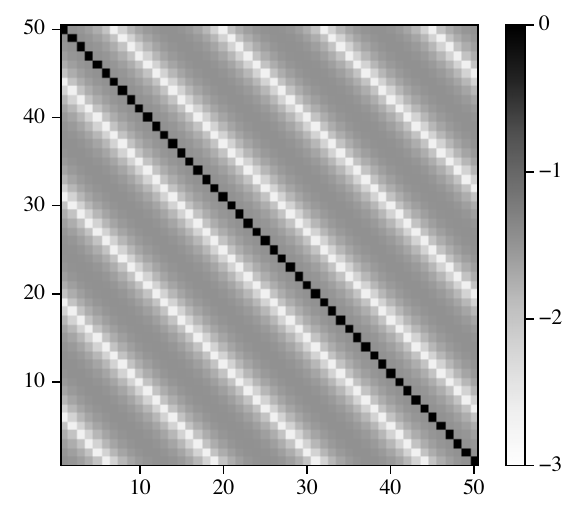}
    \caption{Heatmap of $\log_{10}\left(\frac{|\mP(i,j)|}{\max_{1 \leq i,j \leq n}|\mP(i,j)|}\right)$, where $\mP = \mC^{-1}$ and $\mC$ is the covariance matrix obtained by discretizing the kernel in \cref{eqn:tidal_covar_kernel} onto a uniform grid of 50 points.} \label{fig:precision_sparsity}
\end{figure}

\section{Proof of Frobenius Norm Identity} \label{subsec:frobnorm_identity}

\indent

This is a proof that for any matrix $\mA \in \R^{m \times n}$, $\| \mA \|_\frob^2 = \E_{\,\vx \in \cN(\vzero,\mI)} [\| \mA\vx \|_2^2]$. Consider the thin singular value decomposition $\mA = \mU\mSigma\mV\tp$, where $\mU,\, \mV$ have orthonormal columns and $\mSigma = \mathrm{diag}(\sigma_1,\, \ldots,\, \sigma_{\min\{m,n\}})$. By orthogonal invariance of the standard Gaussian distribution, $\vx \sim \cN(\vzero,\mI)$ implies that $\vz \defeq \mV\tp\vx \sim \cN(\vzero,\mI)$ as well. Thus,
\begin{equation*}
\E_{\,\vx \in \cN(\vzero,\mI)} [\| \mA\vx \|_2^2] = \E_{\,\vx \in \cN(\vzero,\mI)} \trace(\vx\tp \mA\tp\mA \vx) = \E_{\,\vz \in \cN(\vzero,\mI)} \trace(\vz\tp \mSigma^2 \vz) = \E_{\,\vz \sim \cN(\vzero,\mI)} \left[ \sum_{i=1}^{\min\{ m,n \}} \sigma_i^2 z_i^2 \right].
\end{equation*}
Because $\E[z_i^2] = 1$, it follows that $\E_{\,\vx \in \cN(\vzero,\mI)} [\| \mA\vx \|_2^2] = \sum_{i=1}^{\min\{ m,n \}} \sigma_i^2 = \| \mA \|_\frob^2$, thus proving the claim.

\section{Asymptotic Smoothness Proof for \Cref{subsec:tidal_covar}} \label{section:asymptotic_smoothness_proof}

\indent

We differentiate $k$ times in $x$:
\begin{align*}
\partial_x^k \cC(x,y) &= \alpha^2 (-\sigma\sqrt{2})^{-k} H_k\left( \frac{x-y}{\sigma \sqrt{2}} \right) \exp\left(-\frac{(x-y)^2}{2\sigma^2}\right) \\
&\qquad\qquad+ \beta^2 (2\pi)^k \cdot \begin{cases}
(-1)^{(k+1)/2} \sin(2\pi(x-y)) & \text{$k$ odd} \\
(-1)^{k/2} \cos(2\pi(x-y)) & \text{$k$ even}
\end{cases},
\end{align*}
where $H_k$ is the $k\nth$ Hermite polynomial. Using the triangle inequality and \cite[eq.\ 8.22.8]{szego_orthogonal_polynomials},
\begin{align*}
|\partial_x^k \cC(x,y)| &\leq \alpha^2 \frac{(\sigma \sqrt{2})^k \Gamma(k+1)}{\Gamma(\frac{k}{2} + 1)}\exp\left( -\frac{(x-y)^2}{4\sigma^2} \right) + \beta^2 (2\pi)^k + \cO(k^{-1/2}) \\
&\leq \frac{c_1 k! (2^{3/2}\sigma^2)^k}{|x-y|^k} + \beta^2 (2\pi)^k + \cO(k^{-1/2}) \\
&\leq \frac{c_2 k! ((2^{3/2}\sigma^2)^k + (4\pi)^k)}{|x-y|^k} + \cO(k^{-1/2}) \\
&\leq \frac{c_3 k!(4\pi)^k}{|x-y|^k} \| \cC \|_\infty + \cO(k^{-1/2}).
\end{align*}
The second line used $\Gamma(k+1)=k!,\, \Gamma(\frac{k}{2}+1) \geq 1$, and $\exp(-t^2) \leq c_1 \alpha^{-2} t^{-k}$. The third line used $2\pi \leq 4\pi k!  |x-y|^{-1}$. The final line used $a^k + \beta^2 b^k \leq c_3\max\{ a,b \}^k,\, \sigma \leq 1,\, 2^{3/2} \leq 4\pi$, and $\| \cC \|_\infty = 1$. We have shown that \cref{eqn:tidal_covar_kernel} defines an asymptotically smooth kernel with growth parameter $\gamma = 4\pi$.

\section{Proof of \Cref{lemma:local_polynomial_approximation}} \label{section:local_poly_approx_proof}

\indent

The following proof is adapted, nearly detail-for-detail, from the proof of \cite[lemma 3.15]{bebendorf}. All that has changed is the definition of the scale factor, equal here to $\| \cC \|_\infty$, and equal instead to $|\cC(\vx,\vy)|$ in \cite{bebendorf}. We show the details for the sake of having a complete presentation.

Fix a point $\vxi \in R_i$ (cf.\ \cref{subsec:hrs}), and using the analyticity of $\cC_*$ (cf.\ \cref{assumption:asymptotic_smoothness}), expand $\cC(\cdot,\vy)$ about $\vxi$ in a Taylor series as
\begin{align*}
\cC(\vx,\vy) &= T_\vy^{(k)}(\vx) + r_\vy^{(k)}(\vx), & T_\vy^{(k)}(\vx) &\defeq \sum_{|\alpha| < k} \frac{1}{\alpha!} \partial_\vx^\alpha \cC_*(\vxi,\vy)(\vx - \vxi)^\alpha, \\
&& r_\vy^{(k)}(\vx) &\defeq \sum_{|\alpha| \geq k} \frac{1}{\alpha!} \partial_\vx^\alpha \cC_*(\vxi,\vy)(\vx - \vxi)^\alpha,
\end{align*}
for $(\vx,\vy) \in R_i \times R_j$, where $\alpha$ denotes a multi-index and $|\alpha|$ is its order. Asymptotic smoothness (\cref{assumption:asymptotic_smoothness}) implies the following estimate for $r_\vy^{(k)}(\vx)$:
\begin{align*}
|r_\vy^{(k)}(\vx)| &\leq \sum_{|\alpha| \geq k} \frac{1}{\alpha!}|\partial_\vx^\alpha \cC_*(\vxi,\vy)| |(\vx - \vxi)^\alpha| \\
&\leq c \| \cC \|_\infty \sum_{|\alpha| \geq k} \frac{|\alpha|! \gamma^{|\alpha|}}{\alpha!\| \vxi - \vy \|_2^{|\alpha|}} |(\vx - \vxi)^\alpha| \\
&= c \| \cC \|_\infty \sum_{\ell = k}^\infty \left( \frac{\gamma}{\| \vxi - \vy \|_2} \right)^\ell \sum_{|\alpha| = \ell} \binom{\ell}{\alpha} |(\vxi - \vy)^\alpha| \\
&= c \| \cC \|_\infty \sum_{\ell = k}^\infty \left( \frac{\gamma}{\| \vxi - \vy \|_2} \right)^\ell \sum_{|\alpha| = \ell} \binom{\ell}{\alpha} |(\vxi - \vy)^\alpha|,
\end{align*}
recalling that $\binom{\ell}{\alpha} \defeq \frac{|\alpha|!}{\alpha!}$. By the multinomial theorem, $\sum_{|\alpha| = \ell} \binom{\ell}{\alpha} |(\vxi - \vy)^\alpha| = \| \vxi - \vy \|_1^\ell \leq d^{\ell / 2} \| \vxi - \vy \|_2^\ell$. We therefore have
\begin{align*}
|r_\vy^{(k)}(\vx)| &= c \| \cC \|_\infty \sum_{\ell = k}^\infty \left( \frac{\gamma \sqrt{d} \| \vxi - \vx \|_2}{\| \vxi - \vy \|_2} \right)^\ell.
\end{align*}
Because $\vx,\vxi \in \Omega_i,\, \vy \in \Omega_j$, and $(\Omega_i,\Omega_j)$ is $\eta$-admissible (cf.\ \cref{def:admissibility}),
\begin{equation*}
\| \vxi - \vx \|_2 \leq \diam(\Omega_i) \leq \eta \cdot \dist(\Omega_i,\Omega_j) \leq \eta \| \vxi - \vy \|_2.
\end{equation*}
Thus,
\begin{equation*}
|r_\vy^{(k)}(\vxi)| \leq c \| \cC \|_\infty \sum_{\ell = k}^\infty (\gamma \eta \sqrt{d})^\ell = \frac{c (\gamma \eta \sqrt{d})^k}{1 - \gamma \eta \sqrt{d}} \| \cC \|_\infty,
\end{equation*}
completing the proof.

\section{Proof of \Cref{thm:hcov_bound_global}} \label{section:hcov_analysis}

\begin{lemma} \label{lemma:sqfrob_triangle}
For all matrices $\mX,\mY$ of equal dimension, $\| \mX + \mY \|_\frob^2 \leq 2(\| \mX \|_\frob^2 + \| \mY \|_\frob^2)$.
\end{lemma}

\begin{proof}
The Frobenius norm is generated by the inner product $\langle \mX,\mY \rangle_\frob \defeq \trace(\mX\tp\mY)$. Hence, by the Cauchy-Schwartz inequality,
\begin{equation}
\| \mX + \mY \|_\frob^2 = \| \mX \|_\frob^2 + \| \mY \|_\frob^2 + 2\langle \mX,\mY \rangle_\frob \leq \| \mX \|_\frob^2 + \| \mY \|_\frob^2 + 2\| \mX \|_\frob \| \mY \|_\frob. \label{eqn:frobenius_cosine_law}
\end{equation}
Because $\| \mX \|_\frob \| \mY \|_\frob$ is the geometric mean of $\| \mX \|_\frob^2$ and $\| \mY \|_\frob^2$, it is bounded above by the arithmetic mean $\frac{1}{2}(\| \mX \|_\frob^2 + \| \mY \|_\frob^2)$. Inserting this upper bound into \cref{eqn:frobenius_cosine_law} completes the proof.
\end{proof}

\begin{lemma} \label{lemma:gaussian_sample_covar}
Given a symmetric positive semidefinite $\mR \in \R^{k \times k}$ and $\rvY_1,\, \ldots,\, \rvY_m \sim \cN(\vzero,\mR)$ i.i.d.,
\begin{equation*}
{\textstyle \E\left[ \left\| \mR - \frac{1}{m} \sum_{s = 1}^m \rvY_s\rvY_s\stp \right\|^2\right]} \leq 20 \| \mR \|^2 \left( \frac{k}{m} + \sqrt{\frac{k}{m}} \right)^2,
\end{equation*}
where $\| \cdot \|$ is the spectral or Frobenius norm.
\end{lemma}

\begin{proof}
We will first assume that $\mR$ is invertible, and the case where $\mR$ is singular will be addressed at the end of the proof. Let $\rvW_s \defeq \mR^{-1/2}\rvY_s \sim \cN(\vzero,\mI)$. Then, since the spectral and Frobenius norms satisfy $\| \mX\mY \| \leq \| \mX \| \| \mY \|_2$ for all $\mX,\mY$ of conformal dimension,
\begin{equation}
\textstyle \E\left[\left\| \mR - \frac{1}{m}\sum_{s = 1}^m \rvY_s\rvY_s\stp \right\|^2\right] \leq \| \mR \|^2 \cdot \E\left[ \left\| \mI - \frac{1}{m} \sum_{s = 1}^m \rvW_s\rvW_s\stp \right\|_2^2\right]. \label{eqn:baseline_gaussian_submultiplicative_bound}
\end{equation}
Let $\mS_m \defeq \mI - \frac{1}{m}\sum_{s = 1}^m \rvW_s\rvW_s\tp$. We have 
\begin{equation}
\E[\| \mS_m \|_2^2] = \int_0^\infty \prob(\| \mS_m \|_2^2 \geq z)\,dz = 2\int_0^\infty y \cdot \prob(\| \mS_m \|_2 \geq y)\,dy. \label{eqn:sqgaussian_covar_layercake}
\end{equation}
Any $y \geq 2\sqrt{\frac{k}{m}} + \frac{k}{m}$ can be written as $y = \left( x + \sqrt{\frac{k}{m}} + 1 \right)^2 - 1$ for some $x \geq 0$, and by \cite[example 6.3]{wainwright_hds},
\begin{align*}
\prob\left(\| \mS_m \|_2 \geq {\textstyle \left( x + \sqrt{\frac{k}{m}} + 1 \right)^2 - 1 }\right) &= \prob\left( \| \mS_m \|_2^2 \geq {\textstyle 2x + 2\sqrt{\frac{k}{m}} + \left( x + \sqrt{\frac{k}{m}} \right)^2} \right)\\
&\leq 2e^{-mx^2/2}.
\end{align*}
Hence,
\begin{align*}
\E[\| \mS_m \|_2^2] &= 2\int_0^{\frac{k}{m} + 2\left(\frac{k}{m}\right)^{1/2}} y \cdot \prob(\| \mS_m \|_2 \geq y)\,dy + 2\int_{\frac{k}{m} + 2\left(\frac{k}{m}\right)^{1/2}}^\infty y \cdot \prob(\| \mS_m \|_2 \geq y)\,dy \\
&\leq \left( \frac{k}{m} + 2\sqrt{\frac{k}{m}} \right)^2 + 8\int_0^\infty ((x + \theta)^3 - (x + \theta)) e^{-mx^2/2}\,dx,
\end{align*}
where $\theta \defeq \sqrt{\frac{k}{m}} + 1$. Using standard formulas for the moments of a Gaussian distribution,
\begin{equation*}
\int_{0}^\infty x^j e^{-mx^2/2}\,dx = \begin{cases}
\sqrt{\frac{\pi}{2m}} & j = 0 \\
m^{-1} & j=1 \\
\sqrt{\frac{\pi}{2m^3}} & j=2 \\
\frac{1}{2m} & j=3
\end{cases},
\end{equation*}
and this gives us
\begin{equation*}
\E[\| \mS_m \|_2^2] \leq \left( \frac{k}{m} + 2\sqrt{\frac{k}{m}} \right)^2 + 8\left( \frac{1}{2m} + 3\theta\sqrt{\frac{\pi}{2m^3}} + \frac{3\theta^2 - 1}{m} + \theta(\theta^2 - 1)\sqrt{\frac{\pi}{2m}} \right).
\end{equation*}
Inserting $\theta \leq 2\sqrt{k}$, $\theta(\theta^2 - 1) \leq 2\left( \frac{k}{m} + \sqrt{\frac{k}{m}} \right)^2$, $m^{-3/2} \leq m^{-1}$, and various simplifying bounds on constants, we have
\begin{align*}
\E[\| \mS_m \|_2^2] \leq 20 \left( \frac{k}{m} + \sqrt{\frac{k}{m}} \right)^2.
\end{align*}
This proves the claim when $\mR$ is invertible. If instead $\rank \mR = r < k$, then there exists a factorization $\mR = \mQ\mT\mQ\tp$ where $\mQ \in \R^{k \times r}$ has orthonormal columns and $\mT \in \R^{r \times r}$ is symmetric, invertible, and positive-definite. Then, defining $\rvV_s \defeq \mQ\tp\rvW_s$, we have
\begin{equation*}
{\textstyle \E \left\| \mR - \frac{1}{m} \sum_{s = 1}^m \rvY_s\rvY_s\stp \right\|^2 = \E\left\| \mT - \frac{1}{m}\sum_{s = 1}^m \rvV_s\rvV_s\stp\right\|^2} \leq 20 \| \mT \|^2 \left( \frac{r}{m} + \sqrt{\frac{r}{m}} \right)^2.
\end{equation*}
Inserting the bounds $r < k$ and $\| \mT \| \leq \| \mR \|$ completes the proof.
\end{proof}

\begin{theorem}[Local Error Bound] \label{thm:local_error_bound}
Let $(\Omega_i,\Omega_j)$ be an admissible leaf of the block tree. Provided that $\eta < (\gamma \sqrt{d})^{-1}$, the estimator given in \cref{eqn:admissible_crosscovar_estimator} satisfies
\begin{align*}
\E\| \widehat{\mC}_{i,j} - \mC_{i,j} \|_\frob^2 &\leq \frac{8c^2 (\gamma \eta \sqrt{d})^{2k} \| \cC \|_\infty^2}{(1 - \gamma \eta \sqrt{d})^2} \cdot |\idx(\Omega_i)| \cdot |\idx(\Omega_j)| + 20 \left( \frac{t_{k,i}+t_{k,j}}{m} + \sqrt{\frac{t_{k,i}+t_{k,j}}{m}} \right)^2 \| \mC_{ij,ij} \|_\frob^2,
\end{align*}
where
\begin{equation*}
\mC_{ij,ij} \defeq \begin{bmatrix}
\mC_{i,i} & \mC_{i,j} \\
\mC_{j,i} & \mC_{j,j}
\end{bmatrix}.
\end{equation*}
\end{theorem}

\begin{proof}
We start by using \cref{lemma:sqfrob_triangle} to write
\begin{equation}
\E \| \widehat{\mC}_{i,j} - \mC_{i,j} \|_\frob^2 \leq 2\| \mC_{i,j} - \mP_i^{(k)}\mC_{i,j}\mP_j^{(k)} \|_\frob^2 + 2\E \| \widehat{\mC}_{i,j} - \mP_i^{(k)}\mC_{i,j}\mP_j^{(k)} \|_\frob^2. \label{eqn:localbound_triangle1}
\end{equation}
We will separately bound the two terms above. For the first term, the triangle inequality and $\| \mP_i^{(k)} \|_2 \leq 1$ imply that
\begin{align}
\| \mC_{i,j} - \mP_i^{(k)}\mC_{i,j}\mP_j^{(k)} \|_\frob &= \| (\mI - \mP_i^{(k)})\mC_{i,j} + \mP_i^{(k)}\mC_{i,j}(\mI - \mP_j^{(k)}) \|_\frob \nonumber \\
&\leq \| (\mI - \mP_i^{(k)})\mC_{i,j} \|_\frob + \| \mC_{i,j}(\mI - \mP_j^{(k)}) \|_\frob. \label{eqn:localbound_triangle2}
\end{align}
For each $q \in \idx(\Omega_j)$, \cref{lemma:local_polynomial_approximation} shows the existence of a  $T_q \in \poly(d, k-1)$ satisfying
\begin{equation*}
|\cC(\vx,\vr_q) - T_q(\vx)| \leq \frac{c  (\gamma \eta \sqrt{d})^k}{1 - \gamma \eta \sqrt{d}} \| \cC \|_\infty
\end{equation*}
for all $\vx \in \Omega_i$. Define $\vt_q \defeq [T_q(\vr_p)]_{p \in \idx(\Omega_i)} \in \R^{|\idx(\Omega_i)|}$. Because $\vt_q \in \cP_i^{(k)}$, \cref{eqn:polynomial_projector} implies that
\begin{align*}
\| (\mI - \mP_i^{(k)})\mC_{i,j} \|_\frob^2 &= \sum_{q \in \idx(\Omega_j)} \| \mC(\idx(\Omega_i), q) - \mP_i^{(k)}\mC(\idx(\Omega_i), q) \|_\frob^2 \\
&\leq \sum_{q \in \idx(\Omega_j)} \| \mC(\idx(\Omega_i), q) - \vt_q \|_\frob^2 \\
&= \sum_{p \in \idx(\Omega_i)} \sum_{q \in \idx(\Omega_j)} |\cC(\vr_p,\vr_q) - T_q(\vr_p)|^2 \\
&\leq \frac{c^2  (\gamma \eta \sqrt{d})^{2k} \| \cC \|_\infty^2}{(1 - \gamma \eta \sqrt{d})^2} \cdot |\idx(\Omega_i)||\idx(\Omega_j)|.
\end{align*}
Applying this to \cref{eqn:localbound_triangle2} (noting that the same bound holds for $\| \mC_{i,j}(\mI - \mP_j^{(k)}) \|_\frob = \| (\mI - \mP_j^{(k)})\mC_{j,i} \|_\frob$ by symmetry), we have
\begin{equation*}
\| \mC_{i,j} - \mP_i^{(k)}\mC_{i,j}\mP_j^{(k)} \|_\frob \leq \frac{2c(\gamma \eta \sqrt{d})^k  \| \cC \|_\infty}{1 - \eta \gamma \sqrt{d}} \sqrt{|\idx(\Omega_i)| |\idx(\Omega_i)|}
\end{equation*}
Now, to bound the second term in \cref{eqn:localbound_triangle2}, let $\mV_i \in \R^{|\idx(\Omega_i)| \times t_{k,i}}$ and $\mV_j \in \R^{|\idx(\Omega_j)| \times t_{k,j}}$ be orthonormal bases for $\cP_i^{(k)}$ and $\cP_j^{(k)}$, respectively. Define
\begin{align*}
\mV_{ij} &\defeq \begin{bmatrix}
\mV_i & \mZero \\
\mZero & \mV_j
\end{bmatrix},& \mP_{ij}^{(k)} &\defeq \mV_{ij}\mV_{ij}\tp \\
\rvZ_{ij,s} &\defeq \begin{bmatrix}
\rvZ_{i,s} \\
\rvZ_{j,s}
\end{bmatrix}, & \mC_{ij,ij} &\defeq \begin{bmatrix}
\mC_{i,i} & \mC_{i,j} \\
\mC_{j,i} & \mC_{j,j}
\end{bmatrix},& \widehat{\mC}_{ij,ij} &\defeq \frac{1}{m} \sum_{s = 1}^m (\mP_{ij}^{(k)} \rvZ_{ij,s})(\mP_{ij}^{(k)} \rvZ_{ij,s})\tp.
\end{align*}
We then have
\begin{align*}
\mP_{ij}^{(k)}\mC_{ij,ij}\mP_{ij}^{(k)} &= \mV_{ij}\cov[\mV_{ij}\tp\rvZ_{ij,s}]\mV_{ij}\tp,\quad\text{and} \\
\widehat{\mC}_{ij,ij} &= \mV_{ij}\left( \frac{1}{m} \sum_{s = 1}^m (\mV_{ij}\tp\rvZ_{ij,s})(\mV_{ij}\tp\rvZ_{ij,s})\tp \right)\mV_{ij}\tp.
\end{align*}
Note that $\widehat{\mC}_{i,j} - \mP_i^{(k)}\mC_{i,j}\mP_j^{(k)}$ is a submatrix of $\widehat{\mC}_{ij,ij} - \mP_{ij}^{(k)}\mC_{ij,ij}\mP_{ij}^{(k)}$. By \cref{lemma:gaussian_sample_covar},
\begin{align*}
\E \| \widehat{\mC}_{i,j} - \mP_i^{(k)}\mC_{i,j}\mP_j^{(k)} \|_\frob^2 &\leq \E \| \widehat{\mC}_{ij,ij} - \mP_{ij}^{(k)}\mC_{ij,ij}\mP_{ij}^{(k)} \|_\frob^2 \\
&= \E\left\| \cov[\mV_{ij}\tp\rvZ_{ij,s}] - \frac{1}{m} \sum_{s = 1}^m (\mV_{ij}\tp\rvZ_{ij,s})(\mV_{ij}\tp\rvZ_{ij,s})\tp \right\|_\frob^2 \\
&\leq 20 \| \mG_{ij,ij} \|_\frob^2 \left( \frac{t_{k,i}+t_{k,j}}{m} + \sqrt{\frac{t_{k,i}+t_{k,j}}{m}} \right)^2,
\end{align*}
where $\mG_{ij,ij} = \cov[\mV_{ij}\tp\rvZ_{ij,s}]$. Using $\| \mG_{ij,ij} \|_\frob \leq \| \mC_{ij,ij} \|_\frob$ and inserting into \cref{eqn:localbound_triangle1} completes the proof.
\end{proof}
\begin{lemma} \label{lemma:diagblock_overcounting}
For each node $\Omega_i$ of the cluster tree, let
\begin{equation*}
\cA_i = \{ j \suchthat \text{$(\Omega_i,\Omega_j)$ is a leaf of the block tree} \}.
\end{equation*}
Then,
\begin{equation*}
|\cA_i| \leq N(\eta,d) \defeq \frac{\pi^{d/2}}{\Gamma(\frac{d}{2} + 1)} [\textstyle{(2\eta^{-1} + 6)\sqrt{d}}\,]^d.
\end{equation*}
\end{lemma}

\begin{proof}
Supposing that $(\Omega_i,\Omega_j)$ is a leaf, let $(\Omega_s,\Omega_t)$ be its parent node on the block tree, with $\Omega_i \subseteq \Omega_s$ and $\Omega_j \subseteq \Omega_t$. Because $(\Omega_s,\Omega_t)$ is not a leaf, it is inadmissible; hence,
\begin{equation*}
\dist(\Omega_s,\Omega_t) < \eta^{-1}\max\{ \diam(\Omega_s),\diam(\Omega_t) \}.
\end{equation*}
Define sequences $\{ \vx_s^{(p)} \}_{p=1}^\infty \subseteq \Omega_s$ and $\{ \vx_t^{(p)} \}_{p = 1}^\infty \subseteq \Omega_t$ with $\lim_{p \to \infty} \| \vx_s^{(p)} - \vx_t^{(p)} \|_2 = \dist(\Omega_s,\Omega_t)$. For any $(\vx_i,\vx_j) \in \Omega_i \times \Omega_j$, the triangle inequality gives
\begin{align*}
\dist(\Omega_i,\Omega_j) \leq \| \vx_i - \vx_j \|_2 &\leq \| \vx_i - \vx_s^{(p)} \|_2 + \| \vx_s^{(p)} - \vx_t^{(p)} \|_2 + \| \vx_t^{(p)} - \vx_j \|_2 \\
&\leq \diam(\Omega_s) + \| \vx_s^{(p)} - \vx_t^{(p)} \|_2 + \diam(\Omega_t).
\end{align*}
Letting $p \to \infty$, we have
\begin{align*}
\dist(\Omega_i,\Omega_j) &\leq \diam(\Omega_s) + \dist(\Omega_s,\Omega_t) + \diam(\Omega_t) \\
&\leq (2 + \eta^{-1}) \max\{ \diam(\Omega_s),\diam(\Omega_t) \}.
\end{align*}
Let $R_i$ and $R_j$ (resp.\ $R_s$ and $R_t$) be the bounding rectangles for $\Omega_i$ and $\Omega_j$ (resp.\ $R_s$ and $R_t$), as constructed in \cref{subsec:hrs}. Then, the above inequalities imply that
\begin{equation*}
\dist(R_i,R_j) \leq (2 + \eta^{-1})\max\{ \diam(R_s),\diam(R_t) \}.
\end{equation*}
By construction, $R_i$ and $R_j$ are $d$-dimensional hypercubes with equal sidelength, which we denote $\ell$, and $R_s,\, R_t$ are $d$-dimensional hypercubes with sidelength $2\ell$. Hence, $R_i,R_j$ have diameter $\ell \sqrt{d}$, while $R_s,R_t$ have diameter $2\ell \sqrt{d}$. This gives
\begin{equation*}
\dist(R_i,R_j) \leq \rho \defeq (2\eta^{-1} + 4)\ell\sqrt{d}.
\end{equation*}
Let us now bound the number of sidelength-$\ell$ hypercubes within a distance $\rho$ of $R_i$. Let $\vc_i$ be the center of $R_i$; every such hypercube $R_j$ contains a point $\vx_j$ such that $\| \vx_j - \vc_i \|_2 \leq \rho + \ell\sqrt{d}$. Adding the diameter, we find that every such $R_j$ is contained within the ball of radius $\rho + 2\ell\sqrt{d}$ centered at $\vc_i$. The volume of this ball is
\begin{equation*}
V_d \defeq \frac{\pi^{d/2}}{\Gamma(\frac{d}{2} + 1)}({\textstyle \rho + 2\ell\sqrt{d}})^d = \frac{\pi^{d/2}}{\Gamma\left( \frac{d}{2} + 1 \right)} [{\textstyle (2\eta^{-1} + 6)\ell\sqrt{d}}\,]^d,
\end{equation*}
and we therefore find that the number of length-$\ell$ hypercubes within a distance $\rho$ of $R_i$ is at most
\begin{equation*}
\frac{V_d}{\ell^d} = \frac{\pi^{d/2}}{\Gamma(\frac{d}{2} + 1)}[{\textstyle (2\eta^{-1} + 6)\sqrt{d}}\,]^d = N(\eta,d),
\end{equation*}
completing the proof.
\end{proof}

\begin{theorem} \label{thm:hcov_global_bound_explicit}
Let $L$ denote the depth of the cluster and block trees, and recall the definition
\begin{equation*}
\cI(\ell) \defeq \{ i \suchthat \Omega_i \text{ is a node of the cluster tree at level $\ell$} \}\quad\text{for } 1 \leq \ell \leq L.
\end{equation*}
If $\eta < (\gamma \sqrt{d})^{-1}$, then the estimator $\widehat{\mC}_\mathrm{HM}$ produced by \cref{alg:hcov} satisfies
\begin{equation*}
\begin{split}
\E [&\| \widehat{\mC}_\mathrm{HM} - \mC \|_\frob^2] \leq \frac{8c^2(\eta \gamma \sqrt{d})^{2k}}{(1 - \eta \gamma \sqrt{d})^2} \cdot L n^2 \| \cC \|_\infty^2 \\
&\qquad+ (20N(\eta,d) + 40) \left( \frac{\max\{ 2p,2r_{k,d} \}}{m} + \sqrt{\frac{\max\{ 2p,2r_{k,d} \}}{m}} \right)^2 \cdot L \| \mC \|_\frob^2,
\end{split}
\end{equation*}
where $p = \max_{i \in \cI(L)} |\idx(\Omega_i)|$.
\end{theorem}

\begin{proof}
For $1 \leq \ell \leq L$, define
\begin{equation*}
\cX(\ell) \defeq \{ (i,j) \suchthat (\Omega_i,\Omega_j) \text{ is a leaf of the block tree at level $\ell$}\}.
\end{equation*}
Then,
\begin{equation*}
\E \| \widehat{\mC}_\mathrm{HM} - \mC \|_\frob^2 = \sum_{\ell = 1}^L \sum_{(i,j) \in \cX(\ell)} \E \| \widehat{\mC}_{ij} - \mC_{ij} \|_\frob^2.
\end{equation*}
If $(i,j) \in \cX(L)$, then $\max\{ |\idx(\Omega_i)|,\, \idx(\Omega_j) \} \leq p$. Submatrices corresponding to bottom-level block tree nodes are estimated with a standard outer-product estimator (cf.\ lines \ref{line:hcov_bottomlevel_estim_1} and \ref{line:hcov_bottomlevel_estim_2} of \cref{alg:hcov}). Hence, by \cref{lemma:gaussian_sample_covar},
\begin{align*}
\E\| \mC_{i,j} - \widehat{\mC}_{i,j} \|_\frob^2 &\leq \E\left\| \begin{bmatrix} \mC_{i,i} & \mC_{i,j} \\ \mC_{j,i} & \mC_{j,j} \end{bmatrix} - \frac{1}{m}\sum_{s = 1}^m \begin{bmatrix} \rvZ_i\rvZ_i\stp & \rvZ_i\rvZ_j\stp \\ \rvZ_j\rvZ_i\stp & \rvZ_j\rvZ_j\stp \end{bmatrix} \right\|_\frob^2 \\
&\leq 20 \left\| \begin{bmatrix} \mC_{i,i} & \mC_{i,j} \\ \mC_{j,i} & \mC_{j,j} \end{bmatrix} \right\|_\frob^2 \left( \frac{2p}{m} + \sqrt{\frac{2p}{m}} \right)^2 \\
&= (40\| \mC_{i,j} \|_\frob^2 + 20\| \mC_{i,i} \|_\frob^2 + 20\| \mC_{j,j} \|_\frob^2) \left( \frac{2p}{m} + \sqrt{\frac{2p}{m}} \right)^2.
\end{align*}
Summing over leaves at level $L$ and invoking \cref{lemma:diagblock_overcounting},
\begin{align*}
\sum_{(i,j) \in \cX(L)} &\E \| \widehat{\mC}_{i,j} - \mC_{i,j} \|_\frob^2 \leq \left( \frac{2p}{m} + \sqrt{\frac{2p}{m}} \right)^2 \sum_{(i,j) \in \cX(L)} (40 \| \mC_{i,j} \|_\frob^2 + 20\| \mC_{i,i} \|_\frob^2 + 20\| \mC_{j,j} \|_\frob^2 \|).
\end{align*}
Invoking \cref{thm:local_error_bound},
\begin{align*}
\sum_{(i,j) \in \cX(L)} (40 \| \mC_{i,j} \|_\frob^2 + &20\| \mC_{i,i} \|_\frob^2 + 20\| \mC_{j,j} \|_\frob^2 \|) \\
&\leq 40 \cdot \sum_{(i,j) \in \cX(L)} \| \mC_{i,j} \|_\frob^2 + 20N(\eta,d) \cdot \sum_{i \in \cI(L)} \| \mC_{i,i} \|_\frob^2 \\
&\leq (20N(\eta,d) + 40) \| \mC \|_\frob^2,
\end{align*}
implying that
\begin{equation}
\sum_{(i,j) \in \cX(L)} \E\| \widehat{\mC}_{i,j} - \mC_{i,j} \|_\frob^2 \leq (20N(\eta,d) + 40) \left( \frac{2p}{m} + \sqrt{\frac{2p}{m}} \right)^2  \| \mC \|_\frob^2. \label{eqn:globalbound_bottomlevel}
\end{equation}
Consider, now, $\ell < L$. Leaves of the block tree at level $\ell$ are admissible, meaning that by \cref{thm:local_error_bound},
\begin{align*}
\sum_{(i,j) \in \cX(\ell)} &\E\| \widehat{\mC}_{i,j} - \mC_{i,j} \|_\frob^2 \\
&\leq \frac{8c^2(\gamma \eta \sqrt{d})^{2k} \| \cC \|_\infty^2}{(1 - \eta \gamma \sqrt{d})^2} \sum_{(i,j) \in \cX(\ell)} |\idx(\Omega_i)| |\idx(\Omega_j)| \\
&\qquad+ 20 \left( \frac{t_{k,i}+t_{k,j}}{m} + \sqrt{\frac{t_{k,i}+t_{k,j}}{m}} \right)^2 \sum_{(i,j) \in \cX(\ell)} \| \mC_{ij,ij} \|_\frob^2 \\
&\leq \frac{8n^2c^2(\gamma \eta \sqrt{d})^{2k} \| \cC \|_\infty^2}{(1 - \eta \gamma \sqrt{d})^2} \\
&\qquad+ \left( \frac{t_{k,i}+t_{k,j}}{m} + \sqrt{\frac{t_{k,i}+t_{k,j}}{m}} \right)^2 \cdot\sum_{(i,j) \in \cX(\ell)} (40\| \mC_{i,j} \|_\frob^2 + 20\| \mC_{i,i} \|_\frob^2 + 20 \| \mC_{j,j} \|_\frob^2),
\end{align*}
having used $\sum_{(i,j) \in \cX(\ell)} |\idx(\Omega_i)| |\idx(\Omega_j)| \leq \sum_{\ell' = 1}^L \sum_{(i,j) \in \cX(\ell')} |\idx(\Omega_i)| |\idx(\Omega_j)| = n^2$. Once again invoking \cref{lemma:diagblock_overcounting},
\begin{equation*}
\sum_{(i,j) \in \cX(\ell)} (40\| \mC_{i,j} \|_\frob^2 + 20\| \mC_{i,i} \|_\frob^2 + 20 \| \mC_{j,j} \|_\frob^2) \leq (20N(\eta,d) + 40) \| \mC \|_\frob^2.
\end{equation*}
Inserting this bound and $\max\{ t_{k,i}, t_{k,j} \} \leq r_{k,d}$, we have
\begin{equation}
\begin{split}
&\sum_{(i,j) \in \cX(\ell)} \E \| \widehat{\mC}_{i,j} - \mC_{i,j} \|_\frob^2 \\
&\leq \frac{8n^2c^2(\gamma \eta \sqrt{d})^{2k} \| \cC \|_\infty^2}{(1 - \eta \gamma \sqrt{d})^2} + (20N(\eta,d) + 40)\left( \frac{2r_{k,d}}{m} + \sqrt{\frac{2r_{k,d}}{m}} \right)^2 \| \mC \|_\frob^2.
\end{split} \label{eqn:globalbound_upperlevels}
\end{equation}
Summing \cref{eqn:globalbound_upperlevels} over $\ell < L$ and adding \cref{eqn:globalbound_bottomlevel}, we find that
\begin{equation*}
\begin{split}
&\E \| \widehat{\mC}_\mathrm{HM} - \mC \|_\frob^2 \leq \frac{8Lc^2(\eta \gamma \sqrt{d})^{2k}}{(1 - \eta \gamma \sqrt{d})^2} \cdot n^2 \| \cC \|_\infty^2 \\
&\qquad+ L \cdot (20N(\eta,d) + 40) \left( \frac{\max\{ 2p,2r_{k,d}\}}{m} + \sqrt{\frac{\max\{ 2p,2r_{k,d} \}}{m}} \right)^2 \| \mC \|_\frob^2,
\end{split}
\end{equation*}
completing the proof.
\end{proof}

Now we will derive the asymptotic form of \cref{thm:hcov_global_bound_explicit} stated in the main text (\cref{thm:hcov_bound_global}). Note that $V \defeq \sup_{d \geq 1}\Gamma(\frac{d}{2}+1)^{-1} \pi^{d/2} < \infty$. Using this, as well as $\eta \leq 1$,
\begin{align*}
(20N(\eta,d) + 40) &\leq \left( 40 + \frac{20\pi^{d/2}}{\Gamma(\frac{d}{2} + 1)}[{\textstyle (2\eta^{-1} + 6)\sqrt{d}}]^d \right) \\
&\leq A(d) \eta^{-d},
\end{align*}
where $A(d) \defeq 40 + 20V({\textstyle 8\sqrt{d}})^d$. Combining the bounds so far,
\begin{equation}
\begin{split}
\E \| \widehat{\mC}_\mathrm{HM} - \mC \|_\frob^2 &\leq \| \cC \|_\infty^2 \cdot \frac{8Ln^2c^2 (\eta \gamma \sqrt{d})^{2k}}{(1 - \eta \gamma \sqrt{d})^2} \\
&\quad+ \| \mC \|_\frob^2 \cdot 3A(d)L\eta^{-d} \left( \frac{\max\{ r_{k,d},p \}}{m} + \sqrt{\frac{\max\{ r_{k,d},p \}}{m}} \right)^2,
\end{split} \label{eqn:globalbound_with_constants}
\end{equation}
which proves \cref{eqn:hcov_global_rate}, with $\max\{ 8c^2, A(d) \}$ as the hidden constant. The first term will be bounded by $\epsilon^2 \| \cC \|_\infty^2$ provided that
\begin{equation*}
(\eta \gamma \sqrt{d})^{2k} \leq \frac{\epsilon^2}{8Ln^2c^2}(1 - \eta \gamma \sqrt{d})^2,
\end{equation*}
or equivalently,
\begin{equation*}
k \log(\eta \gamma \sqrt{d}) \leq \log\left( \frac{\epsilon}{n \sqrt{L}} \right) + \log(1 - \eta \gamma \sqrt{d}) - \frac{1}{2}\log(8c^2).
\end{equation*}
Dividing by $\log\left( \frac{1}{\eta \gamma \sqrt{d}} \right)$ on both sides proves the first part of \cref{eqn:rank_and_samplesize_rates}. The second term in \cref{eqn:globalbound_with_constants} will be bounded by $\epsilon^2 \| \mC \|_\frob^2$ provided that
\begin{equation*}
\sqrt{\frac{\max\{ r_{k,d}, p \}}{m}} \leq \frac{\epsilon}{2\sqrt{A(d)L\eta^{-d}}},
\end{equation*}
or equivalently,
\begin{equation*}
m \geq 4A(d) \cdot \frac{L\eta^{-d}\max\{ r_{k,d}, p \}}{\epsilon^2}.
\end{equation*}
This completes the proof of \cref{eqn:rank_and_samplesize_rates}.

\section{Interpolative Decomposition Subroutine} \label{section:select_cpqr}

\indent

In \text{RSCov} (cf.\ \cref{alg:rscov}), it is necessary to select from a matrix $\mA$ a column subset that (a) is well-conditioned and (b) spans subspace that approximates $\range \mA$. In the pseudocode for these algorithms, this is generically denoted as $k,\, \mPi,\, \mT \leftarrow \text{InterpDecomp}(\mA, \epsilon)$, where $\epsilon > 0$ is a relative error tolerance which the decomposition
\begin{equation}
\mA\mPi \approx \mA\mPi(\mcol, 1 \mcol k) \begin{bmatrix} \mI & \mT \end{bmatrix} \label{eqn:select_cpqr_id}
\end{equation}
approximately satisfies. This is done by means of a column-pivoted QR (CPQR) factorization
\begin{equation*}
\mA\mPi = \mQ\mR,
\end{equation*}
where which chooses $\mPi$ through a strategy amounting to greedy optimization of $\sigma_\mathrm{min}(\mA\mPi(\mcol, 1 \mcol k))$. The exact procedure is given in pseudocode as \cref{alg:select_cpqr}. More information on this procedure can be found in numerical linear algebra references such as \cite{tropp_2023_randomized}.
\begin{algorithm}
\caption{\text{InterpDecomp}} \label{alg:select_cpqr}
\begin{algorithmic}[1]
\STATE \textbf{input:} matrix $\mA \in \R^{m \times n}$, relative error tolerance $\epsilon > 0$.
\STATE \textbf{output:} integer $k \geq 1$, permutation $\mPi \in \{0,1\}^{n \times n}$, and interpolation matrix $\mT \in \R^{k \times n}$ satisfying \cref{eqn:select_cpqr_id}.
\item[]
\STATE $\mPi,\mQ,\mR \leftarrow \text{cpqr}(\mA)$.\qquad \textit{\# column-pivoted QR factorization}
\STATE $k \leftarrow \max\{ i \geq 1 \suchthat |\mR(i,i)| > \epsilon \cdot |\mR(1,1)| \}$.
\STATE $\mT \leftarrow \mR(1 \mcol k, 1 \mcol k)^{-1}\mR(1 \mcol k, (k+1) \mcol n)$.
\item[]
\RETURN $k,\, \mPi,\, \mT$.
\end{algorithmic}
\end{algorithm}

\section{Tuning Procedure for PLC Localization} \label{section:plc_tuning}

\indent

For our numerical experiments, tune the PLC exponent $\beta$ in \cref{eqn:plc_estimator} using a method inspired by Vishny et al.'s NICE algorithm \cite{vishny_covariance}. Their tuning method can be described as follows: let $\widetilde{\mC}_m = \frac{1}{m} \sum_{s = 1}^m \rvX_s\rvX_s\stp$ be the unregularized sample covariance matrix, and let
\begin{equation*}
\widehat{\mC}_m[\beta] \defeq \begin{bmatrix}
|\hat{\rho}_{11}|^\beta \hat{c}_{11} & \cdots & |\hat{\rho}_{1n}|^\beta \hat{c}_{1n} \\
\vdots & \ddots & \vdots \\
|\hat{\rho}_{n1}|^\beta \hat{c}_{n1} & \cdots & |\hat{\rho}_{nn}|^\beta \hat{c}_{nn}
\end{bmatrix}
\end{equation*}
be the PLC estimate of $\cov[\rvX]$ with exponent $\beta$, where $\hat{c}_{ij}$ (resp.\ $\hat{\rho}_{ij}$) are the empirical covariance (resp.\ correlation) coefficients defined in \cref{eqn:empirical_corr}. Visnhy et al.\ set the PLC localization exponent as
\begin{equation}
\beta_* = \min\{ \beta = 1,\, 2,\, \ldots \suchthat \| \widetilde{\mC}_m - \widehat{\mC}_m[\beta] \|_\frob > \delta \xi \}, \label{eqn:morozov_discrep_tuning}
\end{equation}
where $\delta \in (0,1]$ is a tunable parameter and $\xi \geq 0$ is an \emph{a priori}
estimate for $\| \widetilde{\mC}_m - \cov[\rvX] \|_\frob$. They derive $\xi$ from problem-specific lookup tables, and they use $\delta = 1$ in most cases \cite{vishny_covariance}. \Cref{eqn:morozov_discrep_tuning} is derived from the Morozov discrepancy principle, which is used in many parameter tuning contexts besides this one.

Our tuning method for $\beta$ differs from that of Vishny et al.\ in four respects. First, for simplicity, we use $\delta = 1$ in all cases. Second, for efficiency, we do not evaluate $\| \widetilde{\mC}_m - \widehat{\mC}_m[\beta] \|_\frob$ in full, but rather
\begin{equation*}
\cJ(\beta) \defeq \| \widetilde{\mC}_m(\vi,\vi) - \widehat{\mC}_m[\beta](\vi,\vi) \|_\frob,
\end{equation*}
where $\vi$ is a vector of 100 randomly chosen indices from $1 \mcol n$. Third, rather than constructing lookup tables for $\xi$, we form an on-the-fly estimate $\hat{\xi} \approx \| \widetilde{\mC}_m(\vi,\vi) - \cov[\rvX](\vi,\vi) \|_\frob$ using a Monte-Carlo procedure based on Fisher's $z$-transformation. This is described by Vishny et al.\ in \cite[sec.\ 3.2.1]{vishny_covariance}. Finally, we replace the minimum over $\beta \in \{ 1,\, 2,\, \ldots \}$ in \cref{eqn:morozov_discrep_tuning} with an infimum over $\beta \in [0,10]$. Our tuned value for $\beta$ is, therefore,
\begin{equation*}
\beta_* = \inf\{ 0 \leq \beta \leq 10 \suchthat \cJ(\beta) > \hat{\xi} \}.
\end{equation*}
We approximate the infimum using a bisection search run for a fixed number of iterations.

\end{document}